\documentclass[12pt]{article}

\usepackage[T1]{fontenc}
\usepackage[utf8]{inputenc}

\usepackage{amsmath, amssymb, amsthm}
\usepackage{graphicx}
\usepackage{enumerate}
\usepackage{natbib}
\usepackage{url}
\usepackage{ragged2e}
\usepackage{caption}
\usepackage{setspace}
\usepackage{algorithm}
\usepackage{makecell}
\usepackage{changepage}
\usepackage{csquotes}
\usepackage{textcomp}

\setcitestyle{round}
\theoremstyle{plain}

\newtheorem{theorem}{Theorem}[section]
\newtheorem{lem}[theorem]{Lemma}
\theoremstyle{remark}

\newcommand{\indep}{\perp\!\!\!\perp}

\usepackage{rotating}
\usepackage{float}        
\usepackage{placeins}     
\usepackage{amssymb}
\usepackage{amsthm}
\newtheorem{assumption}{Assumption}

\newcommand{\blind}{1} 

\begin{document}

\def\spacingset#1{\renewcommand{\baselinestretch}{#1}\small\normalsize}
\spacingset{1}


\if1\blind
{
  \title{\bf Where Do the Returns to Schooling Come From? Educational Transitions and Labor Market Payoffs}
  \author{Aleksei Opacic\thanks{aopacic@g.harvard.edu. Many thanks to Clem Aeppli, Kosuke Imai, Ian Lundberg, 
  Michael Zanger-Tishler, Yi Zhang, and Xiang Zhou, 
  as well as members of the Harvard C.A.R.E.S. lab and a reviewer 
  from the Alexander and Diviya Magaro Peer Pre-Review program, 
  for helpful feedback and conversations. Many thanks also to the three anonymous referees and the associate editor for their invaluable suggestions.}\\
  Harvard University}
  \maketitle
} \fi

\if0\blind
{
  \begin{center}
    {\LARGE\bf Where Do the Returns to Schooling Come From? \\ Educational Transitions and Labor Market Payoffs}
  \end{center}
} \fi


\begin{abstract}
Conventional research on educational effects typically either employs a ``years of schooling'' measure of education, or dichotomizes attainment as a point-in-time treatment. Yet, such a conceptualization of education is misaligned with the sequential process by which individuals make educational transitions. In this paper, I propose a causal mediation framework for the study of educational effects on outcomes such as earnings. The framework considers the effect of a given educational transition as operating indirectly, via progression through subsequent transitions, as well as directly, net of these transitions. I demonstrate that the average treatment effect (ATE) of education can be additively decomposed into mutually exclusive components that capture these direct and indirect effects. The decomposition has several special properties which distinguish it from conventional mediation decompositions of the ATE, properties which facilitate less restrictive identification assumptions as well as identification of all causal paths in the decomposition. An analysis of the returns to high school completion in the NLSY97 cohort suggests that the payoff to a high school degree stems overwhelmingly from its direct labor market returns. Mediation via college attendance, completion and graduate school attendance is small because of individuals' low counterfactual progression rates through these subsequent transitions.
\end{abstract}

\noindent
\textit{Keywords:} causal inference, mediation, sequential ignorability, education

\vfill
\newpage
\spacingset{1.7}


\section{Introduction}

One of the most resilient social scientific findings across a range
of national contexts is the strong association between educational
attainment and a variety of life outcomes, including earnings, health,
social capital, and family stability \citep{Hout2012,Chetty2023}.
Conventionally, researchers have taken one of two approaches to evaluating
the social and economic returns to education: the first employs a
``years of schooling'' measure of educational attainment \citep{Angrist1991,Angrist1992,Kane1993,Card1994,Ashenfelter1997,Card1999,Card2001,Angrist2011},
while the second dichotomizes attainment as a point-in-time treatment.
This latter approach has been especially influential in the study
of the impact of postsecondary attainment on earnings, where the treatment
considered is often an indicator for whether an individual has attended,
or graduated from, college \citep{Brand2010,Carneiro2011,Zimmerman2014,Goodman2017,Smith2020,Bleemer2022,Mountjoy2022}.

Despite the important insights this literature has made into establishing
the causal effect of educational attainment on important social and
economic outcomes, extant work has been inattentive to the sequential
process by which people make educational transitions \citep{Mare1980}.\footnote{I use the term ``educational transition'' to refer both to vertical
transitions (e.g. enrollment at a secondary or tertiary institution),
as well as to the attainment of a qualification at a given level (e.g.
high school graduation or BA completion).} At the end of high school, individuals decide whether or not to enroll
in college. Among college enrollees, only 60\% receive a BA within
six years of initial college entry \citep{Snyder2016}, with an even
lower proportion for low-income students and students of color \citep{eller-diprete,Zhou2023}.
Moreover, amidst higher educational expansion in the US, college graduates
must increasingly choose whether to enter the labor market or to enroll
in postgraduate education. Increasingly, therefore, educational attainment
in the US has become a field of multiple levels with sequential transitions,
all of which are independently consequential for individuals\textquoteright{}
labor market outcomes, and therefore of independent scientific interest.

The sequential nature of educational transitions implies that a causal
mediation framework can be employed to study the causal paths by which
education's ``value-added'' occurs. Specifically, we can consider
the first transition in a sequence of educational levels of interest
as a treatment variable, $A$, and subsequent transitions as mediators
that ``transmit'' the effects of the treatment and of prior transitions,
$M_{k}$ ($1\leq k\leq K$). For example, if we are interested in
the total effect of high school completion on earnings, we may ask
to what extent this total effect operates indirectly, through the
effects of college attendance and college completion (putative mediators)
on earnings, or directly, through alternative causal pathways. The
insight that the total causal effect of education can be decomposed
into its direct and indirect effects opens up a range of important
research and policy-oriented questions. For example, tracing to what
extent an early-stage educational intervention boosts outcomes such
as earnings via its promotion of subsequent educational attainment
(its indirect effects), or via earnings directly, would enable policy-makers
to discern what drives the intervention's value and to hone subsequent
policy \citep[e.g.][]{Hurwitz2014,Sullivan2019,Castleman2020,Bird2021,Dynarski2021,Black2023,Turner2024}.
Relatedly, if the early intervention's effects are heterogeneous across
demographic groups, assessing the intervention's direct and indirect
effects could guide researchers to aspects of the educational experience
that either promote or inhibit upward mobility. Nevertheless, prior
empirical approaches are not well-suited to answering these questions:
a ``years-of-schooling'' approach captures the direct effect of each
additional year of schooling, while the dichotomous approach conflates
the direct and indirect effects.\footnote{A further strand of literature, especially prominent in labor economics,
explores labor market returns to horizontal aspects of differentiation
within a given educational level (e.g. college selectivity, as well
as specific colleges) or college types \citep[e.g.][]{Cohodes2014,Goodman2017,Mountjoy2021,Chetty2023,CioccaEller2023}.
While my proposed framework prioritizes the effects of different levels
of education, I discuss in the conclusion how the framework
could be extended to accommodate multivariate mediators.}

In this article, I introduce a causal mediation framework for analyzing
the effects of educational transitions. For the setting of $K$ ($\geq1$)
monotonic  binary or categorical mediators, I develop a general formula that decomposes the
total effect of any level of education into $K+1$ monotonic path-specific
effects (MPSEs): a direct effect net of $K$ subsequent educational
transitions, reflecting the path $A\to Y$, and $K$ mutually exclusive
``continuation'' or gross effects, reflecting the paths $A\to M_{1}\to Y$,
$A\to M_{1}\to M_{2}\to Y$, and $A\rightarrow M_{1}\dots\rightarrow M_{K}\rightarrow Y$.
Most importantly, this decomposition exploits a unique characteristic
of this empirical setting, in which mediators are characterized by ``monotonicity'':
that is, where an individual's potential $k+1$ mediator value is
deterministically zero if that individual's $k$th mediator value
is $0$. The resultant decomposition of the ATE into $K+1$ monotonic
path-specific effects (MPSEs) can be non-parametrically identified
under the assumption of sequential ignorability, which allows for
the effect of each educational level to be confounded by a distinct
set of (observed) intermediate covariates. I introduce several estimation
strategies for my proposed decomposition, including a simple linear
model-based regression-with-residuals (RWR) procedure, and a non-parametric
estimation strategy based on the efficient influence functions (EIFs)
of the target parameters \citep[see][]{Chernozhukov2017,Kennedy2022}.

This study makes three main contributions. Within the realm of education
research, the proposed decomposition draws on important work by \citet{Heckman2018},
who present a similar decomposition of the effect of schooling over
the early life course, but differs in two important respects. First,
I provide nonparametric definitions, identification results, and estimation
strategies for decomposing the total effect of schooling through its
direct and indirect components. Second, my decomposition accommodates
the presence of a distinct set of observed intermediate confounders
for each transition. While one limitation of my approach is that I
assume away the presence of \textit{unobserved} confounders for each
transition, I propose a sensitivity analysis that assesses the robustness
of the results to unobserved confounding, under a set of simplifying
assumptions.

More broadly, my framework speaks to the burgeoning field of causal
mediation analysis in the social, economic, and health sciences, targeted
at assessing the causal pathways by which a treatment affects an outcome.
While prior literature overwhelmingly focuses on single-mediator decompositions
of the ATE, a growing body of work examines mediation estimands in
settings with multiple mediators \citep{Avin2005,Albert2011,VanderWeele2014,Lin2017,Miles2017,Steen2017,Vansteelandt2017,Miles2020}.
In particular, in the case of two causally ordered mediators, \citet{Daniel2015}
show that the ATE can be decomposed into multiple path-specific effects
(PSEs), and outline the assumptions under which some of these effects
are identified. Most recently, \citet{Zhou2022b} generalized this
framework to the case of $K$ mediators, establishing a set of identifiable
PSEs and introducing several regression-based, weighting, and semiparametric
efficient estimators. I extend this literature by examining a special
empirical setting where the mediators are monotonic. Compared with
traditional mediation-based decompositions, monotonicity facilitates
PSE identification under weaker identification assumptions, enables
identification of all of the causal paths in question, as opposed
to just a strict subset of them, and further permits a finer-grained
decomposition. The general decomposition also extends previous literature
on mediation under monotonicity which has focused exclusively on the
case of a single mediator \citep[e.g. ][]{Zhou2022a}. 

Finally, I also contribute to a growing parallel literature that proposes
a range of nonparametric, and semi-parametric efficient estimators
for alternative mediation estimands, based on the efficient influence
functions (EIFs) of the causal quantities of interest \citep[e.g.][]{Miles2020,Farbmacher2022,Zhou2022b},
as well as to closely-related work that proposes semi-parametric
efficient estimators for dynamic treatment effects \citep{Lewis2020,Viviano2021,Bodory2022}.

In the following sections, I first introduce the decomposition for
the case of a single intermediate educational transition, before discussing
the general case of $K$ intermediate transitions and its identification
under the assumption of sequential ignorability (Section \ref{sec:Path-Specific-Effects}).
In Section \ref{sec:Estimation}, I introduce a semiparametric estimation
strategy for estimating the proposed decomposition, and in Section
\ref{sec:Empirical-Analysis}, I illustrate the proposed framework
and methods using data from the National Longitudinal Survey of Youth
(NLSY97) cohort. Section \ref{sec:Conclusion} concludes.

\section{Monotonic Path-Specific Effects\label{sec:Path-Specific-Effects}}
\subsection{A Single Intermediate Transition\label{subsec:A-Motivating-Example}}

I first consider the case of a single intermediate educational transition (monotonic
mediator). Suppressing subscripts $i$, let $A$ denote an indicator
for high school graduation (the initial educational transition), $M_{1}$, an indicator for college attendance (a monotonic mediator or
transition), and $Y$, a binary or continuous outcome of interest
such as earnings. A single-transition decomposition thus assesses
the educational sequence $A\rightarrow M\rightarrow Y$: high school
graduation$\to$college attendance$\to$earnings. In this way, I treat
college attendance as a mediator of the total effect of high school
graduation on earnings, in relation to which the total effect of high
school graduation can be decomposed into an indirect effect (that
``flows through'' college attendance), and a direct effect (net of
college attendance). Following \citet{Heckman2018}, I refer to this
latter term as the ``continuation'' value of educational transition
$A$.

Using potential outcomes notation, let $M(a)$ denote an individual's
potential value of the mediator if their treatment status were set
to $a$, and let $Y(a,m)$ denote that individual's potential outcome
if their treatment and mediator statuses were set to $a$ and $m$,
respectively. With our set of potential outcomes ($\{Y(1),Y(0),Y(1,0),Y(1,1),Y(0,1)\}$), we can then define potential outcomes that involve ``natural'' values of the mediator. For instance, $Y(1, M(0))$ represents the earnings an individual would have if they completed high school ($A=1$) but their college attendance were fixed to the level it would have been had they not completed high school ($A=0$).

With this notation, we can define the ``natural'' decomposition of the average treatment effect as follows \citep{Pearl01, ImaKeeYam10}:

\begin{align}
        \text{ATE} &= \mathbb{E}[Y(1)-Y(0)]\nonumber \\
    &= \text{NDE}(a) + \text{NIE}(1-a),
    \label{eq:nde-nie}
\end{align}

for $a=0,1$, where
$\text{NDE}(a) = \mathbb{E}[Y(1, M(a)) - Y(0, M(a))]$, and $\text{NIE}(a) = \mathbb{E}[Y(a, M(1)) - Y(a, M(0))]$. This decomposition allows a researcher to determine how much of the overall effect of high school completion is due to facilitating access to college, versus directly, net of college attendance.

Potential outcomes such as $Y(1, M(0))$ are central to mediation, but they create challenges in estimation and interpretation due to their ``cross-world'' nature. In particular, a unit cannot simultaneously complete high school and not complete high school at the same time, so $Y(1, M(0))$ is not observable, even in principle. This creates challenges for identification. In particular, it requires a ``cross-world independence'' assumption in order to identify the natural effect decomposition of the ATE: $Y(a', m) \indep M(a) \mid A = a, X$.

Many authors have expressed skepticism about this assumption since it requires that there are no (measured \emph{or} unmeasured) post-treatment confounders (confounders that are affected by the treatment and which affect the mediator and outcome). An alternative approach to mediation analysis involves estimating a quantity known as the controlled direct effect (CDE), which captures the causal effect of a treatment when the mediator is fixed to a given value:

$$\text{CDE}(m)=\mathbb{E}[Y(1,m)-Y(0,m)],$$

for $m=0,1$. This quantity is attractive because it is identified under a weaker assumption than that required for quantities like $Y(1, M(0))$, one that allows for the existence of observed post-treatment confounders. A drawback of this approach is that it does not quantify the extent of mediation through $M$ and only enables a researcher to rule out the existence of alternative mediators other than $M$  (see \cite{acharya2016explaining}).

In the context of educational effects, it is possible to relax the usual cross-world assumption by leveraging a key structural feature of educational transitions.  As I formalize in the following section, I assume that educational transitions are characterized by monotonicity. Here, this means that individuals who do not complete high school cannot attend college,
or $M(0)=0$. This sequential nature of educational transitions therefore
implies the following \textit{restricted} set of potential outcomes: $\{Y(1),Y(0),Y(1,0),Y(1,1)\}$ (i.e., ruling out $Y(0,1)$). Further, since by the composition assumption $Y(a) = Y(a,M(a))$ \citep{VanderWeele2009a}, under monotonicity $Y(0)=Y(0,M(0))=Y(0,0)$.

We can then apply these restrictions to Equation \ref{eq:nde-nie} as follows \citep[see][]{Zhou2022a}:

\begin{align}
\textup{ATE} & =\mathbb{E}[Y(1)-Y(0)]\nonumber \\
&= \underbrace{\mathbb{E}[Y(1, M(0)) - Y(0, M(0))]}_{= \text{NDE}(0)} + \underbrace{\mathbb{E}[Y(1, M(1)) - Y(1, M(0))]}_{=\text{NIE}(1) } \nonumber\\
&= \mathbb{E}[Y(1, 0) - Y(0, 0)] + \mathbb{E}[Y(1, M(1)) - Y(1, 0)] \nonumber\\
&=  \mathbb{E}[Y(1, 0) - Y(0, 0)] + \mathbb{E}\big[ M(1) [Y(1, 1) - Y(1, 0) ]  \big] \nonumber \\
 & =\mathbb{E}[Y(1,0)-Y(0,0)]+\mathbb{E}[M(1)]\mathbb{E}[Y(1,1)-Y(1,0)]+\textup{cov}[M(1),Y(1,1)-Y(1,0)]\label{eq:K=00003D1}\\
 & =\underbrace{\Delta_{0}}_{A\to Y}+\underbrace{\pi_{\textup{1}}\Delta_{\textup{1}}+\eta_{\textup{1}}}_{A\to M\to Y},\label{eq:K1-decomp}
\end{align}

where the third equality follows by monotonicity, the fourth, because $Y(1,M(1))-Y(1,0) = M(1) [Y(1,1)-Y(1,0)] + [1-M(1)] \cdot [Y(1,0)-Y(1,0)] = M(1) [Y(1,1)-Y(1,0)]$, and the fifth, by rules of covariance.  Here, $\Delta_{\textup{0}}$ and $\Delta_{1}$ denote the
direct effects of the first and intermediate transitions on the outcome,
$A\rightarrow Y$ and $M\rightarrow Y$, respectively, $\pi_{\textup{1}}$
denotes the total effect of the first transition on the intermediate
transition $A\to M$, and $\eta_{\textup{1}}$ denotes the covariance
between the effect of the initial transition on completion of the
second and the effect of the second transition on $Y$. Specifically, $\eta_{\textup{1}}$
is positive if those who would attend college given high school completion
(i.e., $M(1)=1$) benefit more from college attendance in terms of
their later earnings (i.e., have a larger $Y(1,1)-Y(1,0)$) than those
who do not (i.e., $M(1)=0$), and negative if the opposite is true.
Meanwhile, the composite term $(\pi_{\textup{1}}\Delta_{1}+\eta_{\textup{1}})$
captures the average indirect effect of the treatment via the
intermediate transition ($A\rightarrow M\rightarrow Y$), comprising
the sum of (i) the probability of college enrollment if an individual
graduated high school, multiplied by the direct effect of college enrollment, and (ii) the covariance between college enrollment and its direct effect
on earnings.

While motivated by the natural effect decomposition of the ATE, the monotonicity constraint on $M$ leads to several differences from Equation \ref{eq:nde-nie}. First, whereas Equation \ref{eq:nde-nie} can be written for $a=0$ and $a=1$, Equation \ref{eq:K1-decomp} is the algebraically unique natural effect decomposition of the ATE under monotonicity. This is because the natural effect decomposition for $a=0$ hinges on the counterfactual term $Y(1,M(0)) = Y(1,0)$, whereas the natural effect decomposition for $a=1$ hinges on the counterfactual term $Y(0,M(1))$. Under monotonicity, only the first of these quantities is well-defined. Second, under mediator monotonicity, $\text{NDE}(0) = \text{CDE}(0)$. Therefore, while in general the CDE does not quantify the extent of mediation through $M$ (since the difference between the ATE and CDE indicates the portions of the total effect due to interaction without mediation, and due to mediation, both with and without interaction \citep{vanderweele2014unification}) in this special case the CDE completely characterizes the extent of mediation.\footnote{An alternative way to see this is that, as \cite{vanderweele2014unification} shows, the individual-level natural direct effect can be written as $(Y(1,0)-Y(0,0)+ M(0) (Y(1,1)-Y(0,1)  - Y(1,0)+ Y(0,0))$, but under monotonicity, $M(0)=0$ for all individuals, and so the individual-level natural direct effect reduces to the first term. Interestingly, under monotonicity there is no ``interaction effect'' between the treatment and mediator (since the mediator can only take a value of 1 when the treatment is activated), and so the NIE reduces to pure indirect effect (PIE) discussed in \cite{vanderweele2014unification}. Hence, under monotonicity, ATE = CDE + PIE, and there are no interaction terms in the decomposition.} Third, and relatedly, because $Y(1,M(0)) = Y(1,0)$, Equation \ref{eq:K1-decomp} does not depend on any cross-world counterfactuals, and therefore can be identified under weaker assumptions than Equation \ref{eq:nde-nie}.

Finally, the decomposition in Equation \ref{eq:K1-decomp} can be compared with a ``randomized effect'' decomposition of treatment effects. An alternative approach to avoiding cross-world counterfactuals and treatment-induced confounding is to use \emph{randomized intervention analogues} to the natural direct and indirect effects (rNDE and rNIE) \citep{VanderWeele2014}. Rather than setting the mediator to the level it would ``naturally'' take under an alternative treatment level, these estimands conceptualize interventions that randomly draw the mediator from its population distribution under each treatment regime.

Analogous to the natural effect decomposition in Equation~\ref{eq:nde-nie}, we can define a \emph{randomized effect decomposition} of a treatment effect as follows:
\begin{align}
    \text{rATE} 
    &= \text{rNDE}(a) + \text{rNIE}(1-a),
    \label{eq:rnde-rnie}
\end{align}
for $a=0,1$, where $\text{rNDE}(a) = \mathbb{E}[Y(1, G_{a\mid X}) - Y(0, G_{a\mid X})]$ and $\text{rNIE}(a) = \mathbb{E}[Y(a, G_{1\mid X}) - Y(a, G_{0\mid X})]$, and rATE is defined to be the sum of these two components. Here, $G_{a\mid X} \equiv \text{Pr}(M=m|A=a,x)$ denotes a value of the mediator randomly drawn from its conditional distribution under treatment $A=a$ given baseline covariates $X$. Importantly, Equation \ref{eq:rnde-rnie} decomposes not the ATE but a randomized average treatment effect (rATE), which may differ from the ATE. By doing so, it avoids the cross-world counterfactual $Y(a, M(a'))$ by substituting a randomized intervention on the mediator distribution. This substitution identifies the indirect and direct pathways, even when post-treatment confounders of the mediator--outcome relationship are present.

When $a=0$, the rNDE captures exactly the NDE under a monotonicity constraint (i.e., $\text{rNDE}(0)=\Delta_0$). This is because $\text{Pr}(M=1|A=0,x) = 0$ and $\text{Pr}(M=0|A=0,x) = 1$ and so the randomized effect reduces to the CDE. Using results from \cite{yu2024natural}, the difference between the NIE under monotonicity and the rNIE is equal to

$$
\begin{aligned}
\pi_1\Delta_1 + \eta_1 - \text{rNIE} &= \mathbb{E}[\text{cov}[M(1)-M(0),Y(1,1)-Y(1,0)|X]] \\
&= \mathbb{E}[\text{cov}[M(1),Y(1,1)-Y(1,0)|X]]  \text{     (since } M(0)=0).\\
\end{aligned}
$$

This expression highlights that the NIE–rNIE gap measures residual within-$X$-stratum covariance between the mediator and its causal effect on $Y$, i.e., dependence induced by intermediate confounding. Thus, deviations between the ATE and rATE will stem from this within-stratum covariance. The randomized intervention framework parallels the monotonic decomposition in Equation~\ref{eq:K1-decomp}, which eliminates logically impossible counterfactuals (e.g., $Y(0,M(1))$) through structural constraints on the mediator rather than redefining the estimand itself.

\subsection{Generalization to $K$ Intermediate Transitions}

I now generalize the approach introduced in the preceding section
to the case of $K$ intermediate transitions. As previously, I denote
the treatment (``initial transition'') of high school graduation by
$A$, and use $M_{1},\dots M_{K}$ to refer to the $K$ subsequent
transitions of interest (``intermediate transitions''), where I
assume that all of $M_{1},\dots,M_{K}$ are binary and that for any
$i<j$, $M_{i}$ temporally precedes $M_{j}$. For instance, we may
wish to decompose the total effect of high school completion on earnings
via college attendance ($M_{1}$), college completion ($M_{2}$) and
graduate school attendance ($M_{3}$). Let an overbar denote a vector
of variables, such that $\overline{M}_{k}=(M_{1},M_{2},\dots M_{k})$
and $(1,\overline{1}_{k-1})=(A=1,M_{1}=1,\dots,M_{k-1}=1)$. Further, let
$[K]$ denote the set $\{0,1,\dots,K\}$. In addition, I denote by
$X$ a vector of pretreatment confounders of the effect of $(A,\overline{M}_{k})$
on ($M_{k+1},Y)$, and by $\overline{Z}_{k}=(Z_{1},\dots Z_{k})$
a vector of intermediate confounders that may confound the causal
effect of $M_{k}$ on ($M_{k+1},Y)$. Using potential outcomes notation,
$Y(1,\overline{1}_{k-1},m_{k})$ thus denotes an individual's potential
earnings if they completed, possibly contrary to fact, the treatment
and the first $k-1$ intermediate transitions, and then either completed
($m_{k}=1)$ or did not complete ($m_{k}=0$) the $k$th \textit{intermediate}
transition. Similarly, $M_{k+1}(1,\overline{1}_{k})$ denotes an individual's
potential value of the $k+1$th intermediate transition were that individual
to complete the treatment and the first $k$ intermediate transitions. As is standard in the mediation literature, I make the following composition assumption \citep{VanderWeele2009a}:

\begin{assumption}
Composition:
$Y(1,\overline{1}_{k-1}, m_k)
=
Y(1,\overline{1}_{k-1}, m_k, M_{k+1}(1,\overline{1}_{k-1}, m_k)),
\forall k \in [K-1], M_0 \equiv A$.
\label{assu:Composition}
\end{assumption}

In words, Assumption \ref{assu:Composition} states that 
a person's potential outcome under $(1,\overline{1}_{k-1}, m_k)$ is equal to their potential outcome under $A=1,\dots,M_{k-1}=1,m_k$ and under the value $M_{k+1}$ would naturally take under $A=1,\dots,M_{k-1}=1,m_k$. I also invoke the following constraint on units' potential transition values:

\begin{assumption}
Monotonicity:
$M_{k+1}(M_{k}=0)=0\forall k\in[K-1],M_{0}\equiv A.$
\label{assu:Monotonicity}
\end{assumption}

Informally, Assumption \ref{assu:Monotonicity} (\textit{monotonicity})
states that an individual's potential $k+1$th transition value is
deterministically $0$ if that individual fails to complete the prior
($k$th) transition. It is analogous to a one-sided non-compliance
assumption within an instrumental variables (IV) framework, which precludes the presence of both ``defiers'' as well as ``always-takers'' principal strata. We can then use
this assumption to decompose the ATE of $A$ on $Y$, which I denote
by $\tau_{0}$. Specifically, let $\tau_{k}$ denote the gross effect
of the $k$th mediator on $Y$, i.e., 

\begin{eqnarray}
\tau_{k}
=
\mathbb{E}[Y(1,\overline{1}_{k})-Y(1,\overline{1}_{k-1},0)], \nonumber
\end{eqnarray}

let $\Delta_{0}$ denote the direct effect of $A$ on $Y$, and let
$\Delta_{k}$ denote the direct effect of the $k$th mediator on $Y$,
i.e., 

\begin{eqnarray}
\Delta_{k}
=
\mathbb{E}[Y(1,\overline{1}_{k-1},1,0)-Y(1,\overline{1}_{k-1},0)]. \nonumber
\end{eqnarray}

To explicate my approach, note that the gross effect of the $k$th
mediator, $\tau_{k}$, includes not only the direct effect $M_{k}\rightarrow Y$,
net of subsequent educational transitions $\Delta_{k}$, but also
the indirect effects of $M_{k}$ via subsequent transitions ($M\rightsquigarrow Y$,
where a squiggly arrow denotes a combination of multiple paths). This
insight motivates us to further decompose $\tau_k$ into its direct
and indirect components. Under the composition assumption, $\tau_{k}$
can be decomposed as 

\begin{equation}
\tau_{k}
=
\Delta_{k}
+
\pi_{k+1}\tau_{k+1}
+
\eta_{k+1},
\label{eq:tau_k}
\end{equation}

where 

\begin{align}
\pi_{k+1}
&=
\mathbb{E}[M_{k+1}(1,\overline{1}_{k})], \nonumber \\
\eta_{k+1}
&=
\textup{cov}[M_{k+1}(1,\overline{1}_{k}),Y(1,\overline{1}_{k+1})-Y(1,\overline{1}_{k},0)]. \nonumber
\end{align}

For $k=1,\dots,K-1$, iteratively substituting equation \ref{eq:tau_k}
into the corresponding expression for $\tau_{k-1}$ yields 

\begin{equation}
\tau_{0}
=
\underbrace{\Delta_{0}}_{A\to Y}
+
\sum_{k=1}^{K}
\underbrace{
(\Pi_{j=1}^{k}\pi_{j})\Delta_{k}
+
(\Pi_{j=1}^{k-1}\pi_{j})\eta_{k}
}_{\theta_{k}\triangleq A\to M_{1}\ldots\to M_{k}\to Y},
\label{eq:K-decomposition}
\end{equation}

where $\pi_0=1$. Further, $\Delta_{K}=\tau_{K}$ is a \textit{gross}
or continuation effect, since this latter path is a composite one
that contains all residual paths omitted in the decomposition (i.e.,
through educational transitions subsequent to $K$, if they exist). Thus, the $\theta_{k}$ terms capture how much of the total effect
of high school completion flows through each intermediate transition
considered (i.e., via college attendance, via college completion,
and via graduate school attendance), while $\Delta_{0}$ captures
that portion of the total effect that operates directly, net of the
$K$ intermediate transitions considered.

\subsection{A Comparison with Conventional Mediation Analysis with Multiple Causally
Ordered Mediators}

The decomposition introduced in the previous section has an analog in the context of a mediation--based
decomposition of the ATE with multiple ordered mediators, but differs
from these conventional mediation analyses in important ways. To illustrate
the differences, consider a binary treatment, $A$, an outcome of
interest, $Y$, and a vector of pretreatment covariates, $X$, and
let $M_{1},M_{2},\ldots M_{K}$ denote $K$ causally ordered mediators,
assuming that for any $i<j,M_{i}$ precedes $M_{j}$, as above. Moreover,
define $M_0\equiv\emptyset$, and let
\[
\overline{M}_{k} \equiv (M_{1},M_{2},\ldots,M_{k}),
\qquad
\overline{M}_{k}(a)
\equiv
\big(M_{1}(a),M_{2}(a,M_{1}(a)),\ldots,M_{k}(a,\overline{M}_{k-1}(a))\big),
\]
with $\overline{M}_{0}(a)\equiv\emptyset$. Using the potential
outcomes notation as above, define the following expectation of a nested
counterfactual,
\[
\psi(a,\overline{M}_{k}(a^{\ast}))
\triangleq
\mathbb{E}\!\left[Y\!\big(a,\overline{M}_{k}(a^{\ast})\big)\right].
\]

Under Pearl’s \citeyearpar{Pearl2009} nonparametric structural equation
model (NPSEM), \citet{Zhou2022b} demonstrates that the ATE of $A$
on $Y$ can be decomposed into $K+1$ identifiable path-specific effects (PSEs)
corresponding to each of the causal paths $A\rightarrow Y$ and
$A\rightarrow M_{k} \rightsquigarrow Y$ $(k\in[K])$
(see also \citealp{Daniel2015}):

\begin{equation}
\mathrm{ATE}
=
\mathbb{E}[Y(1)-Y(0)]
=
\underbrace{
\psi(1,\overline{M}_{K}(0))-\psi(0,\overline{M}_{K}(0))
}_{A\rightarrow Y}
+
\sum_{k=1}^{K}
\underbrace{
\big(
\psi(1,\overline{M}_{k}(0))-\psi(1,\overline{M}_{k-1}(0))
\big)
}_{A\rightarrow M_{k} \rightsquigarrow Y}.
\label{eq:pse-zhou}
\end{equation}

This decomposition holds algebraically when Assumption \ref{assu:Monotonicity}
does not hold (i.e., when the mediators are not monotonic), and differs from the proposed decomposition (Equation \ref{eq:K-decomposition}) in several ways.\footnote{
\cite{Vansteelandt2017} propose a decomposition of the total effect in settings with multiple, potentially causally ordered mediators using interventional (in)direct effects. A key advantage of this approach is that, like the MPSE decomposition, it permits exposure-induced intermediate confounding by defining effects in terms of stochastic interventions on mediator distributions. When mediators are causally dependent, however, mediation through a given mediator is not, in general, fully summarized by the corresponding interventional indirect effect alone; an additional component capturing exposure-induced changes in the dependence between mediators (their Equation~(8)) is required to complete the decomposition. Accordingly, in settings with causally dependent mediators, mediation along a particular sequential path may be zero even when the associated interventional indirect effect is nonzero. By contrast, the framework proposed here exploits monotonicity to restrict the support of downstream mediators to that of upstream transitions, so that all mediation operates along well-defined sequential paths and no separate dependence term is required. Interestingly, under monotonicity, the interventional direct effect in \cite{Vansteelandt2017} coincides exactly with $\Delta_0$ in the MPSE decomposition, echoing the result that the randomized natural direct effect at zero equals $\Delta_0$, as discussed in Section~\ref{subsec:A-Motivating-Example}.
}

First and most importantly, the monotonic decomposition departs from conventional mediation decompositions in distinguishing a different set of causal pathways. As shown in the previous section, the monotonic and standard mediation decompositions both begin with a natural effect decomposition in the case of a single mediator. Conventional mediation decompositions then obtain a 2-mediator decomposition by decomposing the direct path $A \to Y$ net of $M_1$ into its direct component net of $M_2$ and indirect path $A \to M_2 \to Y$. This is shown in Figure \ref{fig:decomp-tree}, Panel A.
Notably, to assess the mediating role of $M_{1}$, only the composite path 
$A\rightarrow M_{1}\rightsquigarrow Y=(A\rightarrow M_{1}\rightarrow Y)+(A\rightarrow M_{1}\rightarrow M_{2}\rightarrow Y)$
is identified. The reason for this is that identification of the pure path-specific effects
$A\rightarrow M_{1}\rightarrow Y$ or $A\rightarrow M_{1}\rightarrow M_{2}\rightarrow Y$ fails when the 
\emph{recanting witness criterion} is violated \citep{Avin2005}. 
Specifically, $M_{1}$ acts as a ``recanting witness'' because it lies on the path 
$A\rightarrow M_{1}\rightarrow Y$ but also has an additional path to $Y$ through 
$M_{2}$ ($M_{1}\rightarrow M_{2}\rightarrow Y$) that is not contained in the path of interest. 
The nested counterfactual $\mathbb{E}[Y(a,M_{1}(a_{1}),M_{2}(a_{2},M_{1}(a_{12})))]$
is identified if and only if $a_{1}=a_{12}$. Consequently, the individual PSEs 
for $A\rightarrow M_{1}\rightarrow Y$ and $A\rightarrow M_{1}\rightarrow M_{2}\rightarrow Y$ 
are not identified, whereas the composite path 
$A\rightarrow M_{1}\rightsquigarrow Y$, which includes all possible paths from $M_{1}$ to $Y$, 
remains identifiable.
More generally, traditional decompositions do not disentangle the mediating
effects of $M_{k}$ that are direct (net of subsequent mediators)
and indirect (through different combinations of subsequent mediators). 

By contrast, as shown in Figure \ref{fig:decomp-tree}, Panel B, my proposed decomposition under mediator monotonicity begins with the natural effect decomposition with a single mediator as in Equation \ref{eq:nde-nie}, and then decomposes the path via $M_1$ into pathways that operate further via $M_2$ ($A\rightarrow M_{1}\rightarrow M_{2}\rightarrow Y$) and that operate directly to $Y$ ($A\rightarrow M_{1}\rightarrow Y$), leaving the direct effect ($\Delta_0$) untouched. In fact, $\Delta_0$ is not further decomposable under monotonicity because this assumption yields a more  restricted set of causal 
paths. To see why, consider the pathway  $A\to M_2 \to Y$ under
monotonicity in the two-mediator setting:
\[
\mathbb{E}\!\left[
Y\!\big(1,\,M_1(0),\,M_2(1,M_1(0))\big)
-
Y\!\big(1,\,M_1(0),\,M_2(0,M_1(0))\big)
\right].
\]
Under mediator monotonicity, $M_1(0)=0$ (e.g., without completing high school one
cannot attend college), and $M_2(\cdot,0)=0$ (e.g., without college attendance one
cannot complete a BA). Therefore
\[
M_2(0,M_1(0)) = M_2(0,0)=0.
\]
and the path $A\to M_2 \to Y$ equals
\[
Y(1,0,0)-Y(1,0,0)=0.
\]

More generally, consider the path-specific effect attributed to the pathway
$A \to M_k \to Y$ with $k \ge 2$, defined in the usual way by fixing all earlier
mediators at their natural levels under $A=1$ and varying only $M_k$ with respect
to $A$:
\[
\tau_{A\to M_k \to Y}(a)
=
\mathbb{E}\!\left[
Y\!\big(1,\overline M_{k-1}(0), M_k(1,\overline M_{k-1}(0))\big)
-
Y\!\big(1,\overline M_{k-1}(0), M_k(0,\overline M_{k-1}(0))\big)
\right],
\]
where $\overline M_k(a)$ denotes the vector $(M_1(a),\ldots,M_k(a,\overline M_{k-1}(a)))$. Under monotonicity, fixing earlier mediators at their baseline levels
implies that $\overline M_{k}(0)=0$, and that
$M_k(1,\overline M_{k-1}(0))=M_k(0,\overline M_{k-1}(0))=0$. In words, there is no direct effect of the treatment on later mediators once earlier transitions are held fixed at their natural value under $A=0$. As a result, the two potential
outcomes inside the expectation coincide, and
\[
\tau_{A\to M_k \to Y}(1)=0 \qquad \text{for all } k \ge 2.
\]
Any indirect effect involving $M_k$ must therefore
operate through the full causal chain
$A \to M_1 \to \cdots \to M_k \to Y$. Figure \ref{fig:MPSEs} illustrates the causal pathways defined and
identified under the proposed decomposition in the case of two monotonic
mediators.

Second, the PSE decomposition of the
ATE in general mediation settings is not algebraically unique, and
thus the PSEs defined under alternative decompositions will differ
if the effects of the treatment and each mediator vary across levels
of the other mediators. In fact, depending on the order in which the
paths $A\rightarrow Y$ and $A\rightarrow M_{k}\rightsquigarrow Y$
are considered, there are $(K+1)!$ identifiable different ways of
decomposing the ATE; the decomposition shown in Equation \ref{eq:pse-zhou}
is just one such decomposition. Consider the case of two causally
dependent mediators. In this setting, the causal pathway $A\to M_{2}\rightsquigarrow Y$
can be defined with respect to four different combinations of levels
of the treatment and first mediator: under (i) $a=1$ and $M_{1}(1)$,
(ii) $a=1$ and $M_{1}(0)$, (iii) $a=0$ and $M_{1}(1)$, or (iv)
$a=0$ and $M_{1}(0)$. By contrast, as a direct consequence of monotonicity,
the MPSE decomposition is the unique PSE decomposition of the ATE.

Finally, the sequential ignorability assumption required to identify
the MPSE decomposition is weaker than those required to identify a
generic PSE decomposition of the ATE. Specifically, the latter requires
Pearl's \citeyearpar{Pearl2009} non-parametric structural equation
model (NPSEM), which stipulates that $\ensuremath{\left(M_{k+1}\left(a_{k+1},\bar{m}_{k}\right),\ldots M_{K}\left(a_{K},\bar{m}_{K-1}\right),Y\left(a_{K+1},\bar{m}_{K}\right)\right)\perp\!\!\!\perp M_{k}\left(a_{k},\bar{m}_{k-1}^{*}\right)\mid X,A,\bar{M}_{k-1},\forall k\in[K]}.$
This assumption, sometimes referred to as the ``cross-world'' independence
assumption, is stronger than the sequential ignorability assumption
(\ref{assu:SI}) required to identify the MPSE decomposition since
it rules out the existence of confounders of the mediators, be they
observed or unobserved. By contrast, the MPSE decomposition identification
results accommodate \textit{observed} intermediate confounding without
altering the substance of the decomposition. I discuss this point in detail in the following section.

\subsection{Identification\label{sec:Identification}}

To identify the causal effects of interest, I rely on a series of
sequential ignorability assumptions. While most closely associated
with the dynamic treatment effects literature, which relies on observing
a complete set of time-varying confounders in order to identify longitudinal
effects \citep[see e.g.][]{Lewis2020,Viviano2021,Bodory2022}, these
assumptions can be transferred to a mediation context, given the fact
that the mediators of interest are all causally ordered. As discussed in the previous section, sequential ignorability identification assumptions
are distinct from -- and in fact weaker than -- the assumptions typically
employed in studies of causal mediation. 

Let $M_{k}=\emptyset\forall k<1$. In order to estimate the decomposition
shown in Equation \ref{eq:K-decomposition}, it suffices to identify
the expectation of two types of composite counterfactuals
($Y(1,\overline{1}_{k-1},m_{k})$
and $M_{k+1}(1,\overline{1}_{k})$), as well as covariance terms of
the form
$\mathrm{cov}[M_{k+1}(1,\overline{1}_{k}),Y(1,\overline{1}_{k+1})-Y(1,\overline{1}_{k},0)]$
$\forall k\in[K-1]$. I invoke the following three assumptions:

\begin{assumption}
Consistency: for any unit, if $A=a$, $Y=Y(a)$; if $(A,\overline{M}_{k-1},M_{k})=(1,\overline{1}_{k-1},m_{k})$,
then $Y=Y(1,\overline{1}_{k-1},m_{k})$ $\forall k\in[K]$, and if $(A,\overline{M}_{k})=(1,\overline{1}_{k})$,
then $M_{k+1}=M_{k+1}(1,\overline{1}_{k})\forall m_{k+1}\in\{0,1\},\forall k\in[K-1]$.
\label{assu:Consistency}
\end{assumption}

\begin{assumption}
Sequential ignorability: $(M_1(1),Y(a))\perp\!\!\!\perp A|X$; \\
$Y(1,\overline{1}_{k-1},m_{k})\perp\!\!\!\perp{M}_{k}|X,\overline{Z}_{k},\overline{M}_{k-1}=\overline{1}_{k-1}$
and
$M_{k+1}(1,\overline{1}_{k})\perp\!\!\!\perp{M}_{k}|X,\overline{Z}_{k},\overline{M}_{k-1}=\overline{1}_{k-1}$,
$\forall m_{k}\in\{0,1\},\forall k\in\{1,\dots,K\},M_{0}\equiv A$.
\label{assu:SI}
\end{assumption}

\begin{assumption}
Positivity: $p_{A|X}(a|x)>\epsilon>0$,
$p_{M_{k}|X,A,\overline{Z}_{k},\overline{M}_{k-1}}(m_{k}|x,a,\overline{z}_{k}, \overline{m}_{k-1} = \overline{1}_{k-1})>\epsilon>0$
$\forall k\in[K]$.
\label{assu:Positivity}
\end{assumption}

Assumption \ref{assu:Consistency} (\textit{consistency}) states that
a unit's observed outcome equals its potential outcome under a given
treatment sequence. Note that under Assumption \ref{assu:Composition}
(Composition), if $Y=Y(1,\overline{1}_{k-1}, m_k)$, then
$Y = Y(1,\overline{1}_{k-1}, m_k, M_{k+1}(1,\overline{1}_{k-1}, m_k))
= \dots = Y(1,\overline{1}_{k-1}, m_k, M_{k+1}(1,\overline{1}_{k-1}, m_k), \dots, M_{K}(\cdot))$.
In plain words, the $K-k$ mediators after mediator $k$ all take their
natural levels. Assumption \ref{assu:SI} (\textit{sequential ignorability})
is the no unmeasured confounding assumption for the treatment and all
mediators. It is considered plausible when sufficient pre-treatment and
intermediate covariates $(X,\overline{Z}_{K})$ are collected. Finally,
Assumption \ref{assu:Positivity} (\textit{positivity}) requires that
treatment assignment is not deterministic, and that mediator assignment
is not deterministic when the treatment and prior mediators $\overline{m}_{k-1}$ take on a
value of $1$.

Under Assumptions \ref{assu:Consistency}--\ref{assu:Positivity},
$\mathbb{E}[Y(1,\overline{1}_{k-1},m_{k})]$ and
$\mathbb{E}[M_{k+1}(1,\overline{1}_{k})]$
are identified, respectively, as 

\begin{eqnarray}
 & \mathbb{E}[Y(1,\overline{1}_{k-1},m_{k})]
=\int_{x}\int_{\overline{z}_{k}}
\mathbb{E}[Y|x,\overline{z}_{k},1,\overline{1}_{k-1},m_{k}]
\big[\prod_{j=1}^{k}dP(z_{j}|x,\overline{z}_{j-1},\overline{1}_{j-1})\big]dP(x)\\
 & \mathbb{E}[M_{k+1}(1,\overline{1}_{k})]
=\int_{x}\int_{\overline{z}_{k}}
\mathbb{E}[M_{k+1}|x,\overline{z}_{K},1,\overline{1}_{k}]
\big[\prod_{j=1}^{k}dP(z_{j}|x,\overline{z}_{j-1},\overline{1}_{j-1})\big]dP(x)
\end{eqnarray}

For a proof of the above formulas, see \citet{Robins1986}. The covariance
($\eta_{k}$) components in the decomposition are then identified
as the ``residual'' terms such as in Equation \ref{eq:tau_k}, which
follows directly from the fact that all other components in these
equations are identified. Thus, for $k\in\{1,\dots K\}$, we can identify
$\eta_{k}$ as 

\begin{eqnarray}
\eta_{k}=\tau_{k-1}-\Delta_{k-1}-\pi_{k}\tau_{k}.
\end{eqnarray}

\section{Estimation\label{sec:Estimation}}

The identification results outlined above suggest that the proposed
decomposition can be estimated via several approaches, including outcome-based modeling,
models for the treatment and mediators via inverse probability weighting, as well as doubly robust approaches. This section outlines two complementary estimation strategies (one semiparametric approach and one parametric approach) for implementing the MPSE decomposition. After providing a short summary of the approaches, I detail a semiparametric estimation approach in the main text, and refer readers to Supplementary Material A for further detail on the parametric approach. I also provide a simulation study comparing the performance of the two estimation strategies in Supplementary Material B.

The first estimation approach is a semiparametric, debiased machine learning (DML) estimator, which uses efficient influence functions and cross-fitting to estimate each component of the monotonic path-specific effect (MPSE) decomposition. The second is a simpler regression-with-residuals (RWR) estimator, which relies on parametric linear models and sequential residualization to translate the $\theta_{k}$ components in Equation
\ref{eq:K-decomposition} directly into regression coefficients that can be read off from these linear models. The RWR approach provides a transparent, low-computational alternative that directly links the decomposition to regression coefficients, making it useful both as a practical estimator in large datasets and as a parametric benchmark against which to compare the DML results. Nevertheless, when $X$ and $\overline{Z}_{K}$
are high-dimensional, the parametric models underlying RWR may be misspecified, which can in turn introduce bias. By contrast, the DML approach is robust to flexible, high-dimensional models for the treatment, mediators, and outcome. 

The DML approach is characterized by two components: first, the use of
a Neyman orthogonal estimating equation based on the efficient influence
function (EIF) for the target parameters, which makes estimates of
the parameter ``locally robust'' to estimates of the nuisance functions;
second, the use of a $K$-fold cross-fitting algorithm \citep{Chernozhukov2017}.

Let $O=(X,A,\overline{Z}_{K},\overline{M}_{K},Y)$ denote the observed
data, and $\mathcal{P}$ a nonparametric model over $O$ wherein all
laws satisfy the positivity assumption described in Section \ref{sec:Path-Specific-Effects}.
Before proceeding, I define the following auxiliary functions, as
introduced in Section \ref{sec:Path-Specific-Effects}: $\psi_{km_{k}}\triangleq\mathbb{E}[Y(\overline{1}_{k},m_{k})]$
and $\phi_{k}\triangleq\mathbb{E}[M_{k+1}(\overline{1}_{k+1})]$,
for all $k\in[K]$. Using the identification results
given in Section \ref{sec:Identification}, $\psi_{km_{k}}$ can
be written in terms of expectations of observed data:

\begin{eqnarray}
\psi_{km_{k}}
=
\mathbb{E}_{X}
\mathbb{E}_{Z_{1}\mid X,1}
\ldots
\mathbb{E}_{Z_{k}\mid X,\overline{Z}_{k-1},1,\overline{1}_{k-1}}
\mathbb{E}[Y\mid X,\overline{Z}_{k},1,\overline{1}_{k-1},m_{k}].
\label{Iterated-expectation}
\end{eqnarray}

For each $j\in[k]$, we can thus define
$\ensuremath{\mu_{jm_{k}}^{k}\left(X,\bar{Z}_{k}\right)}$
iteratively as

\begin{align*}
\mu_{km_{k}}^{k}\left(X,\bar{Z}_{k}\right)
&\triangleq
\mathbb{E}\left[Y\mid X,\bar{Z}_{k},1,\overline{1}_{k-1},m_{k}\right],\\
\mu_{jm_{k}}^{k}\left(X,\bar{Z}_{j}\right)
&\triangleq
\mathbb{E}\left[
\mu_{j+1m_{k}}^{k}\left(X,\bar{Z}_{j+1}\right)
\mid X,\bar{Z}_{j},1,\overline{1}_{j}
\right]
\forall j\in[k-1].
\end{align*}

This recursive definition of
$\ensuremath{\mu_{jm_{k}}^{k}\left(X,\bar{Z}_{k}\right)}$
is a compact way of expressing the nested expectations in
Equation \ref{Iterated-expectation}. For example, in the case of $k=2$,
the recursion becomes
\[
\mu^{2}_{2 m_2}(X,Z_1,Z_2)
=
\mathbb{E}[Y \mid X,Z_1,Z_2,1,\overline{1}_{1},m_2],
\]
\[
\mu^{2}_{1 m_2}(X,Z_1)
=
\mathbb{E}_{Z_2 \mid X,Z_1,1,\overline{1}_{1}}
\!\left[\mu^{2}_{2 m_2}(X,Z_1,Z_2)\right],
\]
\[
\mu^{2}_{0 m_2}(X)
=
\mathbb{E}_{Z_1 \mid X,1}
\!\left[\mu^{2}_{1 m_2}(X,Z_1)\right].
\]
Thus the counterfactual mean is
\[
\psi_{2 m_2}
=
\mathbb{E}_{X}\!\left[\mu^{2}_{0 m_2}(X)\right].
\]

These expressions make explicit that, at each step, we take the expectation
of the previous conditional expectation with respect to the distribution
of the next confounder, $Z_{j+1}$ given
$(X,\bar{Z}_{j},1,\overline{1}_{j})$.

Next, let
$\pi_{km_{k}}\big(X,\overline{Z}_{k}\big)
\triangleq
\mathrm{Pr}[M_{k}=m_{k}\mid X,\overline{Z}_{k},1,\overline{1}_{k-1}]$
$\forall k\in[K]$,
and
$\pi_{01}\big(X\big)\triangleq\mathrm{Pr}[A=1\mid X]$. The efficient
influence function (EIF) of $\psi_{km_{k}}$ is closely related to
the EIF for the g-formula, and can be written as

\begin{eqnarray}
\psi_{km_{k}}(O)=\sum_{j=0}^{k+1}\varphi_{j}(O),
\end{eqnarray}

where

\[
\begin{aligned}\varphi_{0}(O) & =\mu_{0m_{k}}^{k}(X)-\psi_{km_{k}}\\
\varphi_{j}(O) & =\frac{A}{\pi_{01}\big(X\big)}\left(\prod_{l=1}^{j-1}\frac{M_{l}}{\pi_{l1}\big(X,\overline{Z}_{l}\big)}\right)\left(\mu_{jm_{k}}^{k}\left(X,\overline{Z}_{j}\right)-\mu_{j-1m_{k}}^{k}\left(X,\bar{Z}_{j-1}\right)\right),\quad j\in\{1\dots,k\}\\
\varphi_{k+1}(O) & =\frac{A}{\pi_{01}\big(X\big)}\left(\frac{\mathbb{I}(M_{k}=m_{k})}{\pi_{km_{k}}\big(X,\overline{Z}_{k}\big)}\prod_{l=1}^{k-1}\frac{M_{l}}{\pi_{l1}\big(X,\overline{Z}_{l}\big)}\right)\left(Y-\mu_{km_{k}}^{k}\left(X,\bar{Z}_{k}\right)\right).
\end{aligned}
\]

For a proof, see \citet{Rotnitzky2017}. The expression above decomposes the efficient influence function for 
$\psi_{k m_k}$ into $k\!+\!2$ components, each corresponding to one layer of 
the iterated expectation representation in Equation~\ref{Iterated-expectation}. 
The first term, $\varphi_{0}(O)$, centers the EIF (around zero) by subtracting 
the target parameter $\psi_{k m_k}$ from its plug-in estimate 
$\mu^{k}_{0 m_k}(X)$. The next $k$ terms, $\varphi_{1}(O),\ldots,\varphi_{k}(O)$, are sequential 
bias–correction terms. Each $\mu^{k}_{j m_k}(X,\bar Z_j)$ is an estimated 
conditional expectation that appears in the iterated formula for 
$\psi_{k m_k}$, and estimation error in these nuisance regressions would 
normally introduce bias. Each $\varphi_{j}(O)$ therefore takes the form of a 
weighted residual that subtracts off the discrepancy between two successive 
levels of the recursion, 
$\mu^{k}_{j m_k}(X,\bar Z_j) - \mu^{k}_{(j-1)m_k}(X,\bar Z_{j-1})$, and weights 
it by the inverse probability of the relevant portion of the treatment-mediator 
history. The last term $\varphi_{k+1}(O)$ plays the same bias-correction role 
for the outcome regression. It incorporates the residual 
$Y - \mu^{k}_{k m_k}(X,\bar Z_k)$ and weights it by the inverse probability of 
the full treatment–mediator sequence required to identify $\psi_{k m_k}$. 
Together, these $k\!+\!1$ bias–correction terms ensure that the influence 
function is orthogonal to first–order errors in all nuisance functions, 
allowing the resulting DML estimator to achieve semiparametric efficiency under 
standard regularity conditions.

Since the expression above gives the efficient influence function for 
$\psi_{k m_k}$ under the nonparametric model $\mathcal{P}$, its variance 
determines the lowest achievable asymptotic variance for any regular, 
asymptotically linear estimator of $\psi_{k m_k}$. Consequently, the semiparametric efficiency bound for any 
asymptotically linear estimator of $\psi_{k m_k}$ is  
$\mathbb{E}\big[(\varphi_{k m_k}(O))^{2}\big]$.

This EIF motivates an EIF-based estimator for $\psi_{km_{k}}$, obtained
by solving the empirical moment condition $\mathbb{P}_{n}[\varphi_{km_{k}}(O;\hat{\eta})]=0$,
where $\mathbb{P}_{n}[\cdot]$ denotes an empirical average, and where
$\varphi_{km_{k}}(O;\hat{\eta})$ denotes the estimated EIF, evaluated
using plug-in estimators for the nuisance functions. Specifically,

\begin{equation}
\begin{aligned}\hat{\psi}_{km_{k}}^{\text{eif}} & =\mathbb{P}_{n}\bigg[\frac{A}{\hat{\pi}_{01}\big(X\big)}\left(\frac{\mathbb{I}(M_{k}=m_{k})}{\hat{\pi}_{km_{k}}\big(X,\overline{Z}_{k}\big)}\prod_{l=1}^{k-1}\frac{M_{l}}{\hat{\pi}_{l1}\big(X,\overline{Z}_{l}\big)}\right)\left(Y-\hat{\mu}_{km_{k}}^{k}\left(X,\bar{Z}_{k}\right)\right)\\
 & +\sum_{j=1}^{k}\frac{A}{\hat{\pi}_{01}\big(X\big)}\left(\prod_{l=1}^{j-1}\frac{M_{l}}{\hat{\pi}_{l1}\big(X,\overline{Z}_{l}\big)}\right)\left(\hat{\mu}_{jm_{k}}^{k}\left(X,\overline{Z}_{j}\right)-\hat{\mu}_{j-1m_{k}}^{k}\left(X,\bar{Z}_{j-1}\right)\right)\\
 & +\hat{\mu}_{0m_{k}}^{k}(X)\bigg].
\end{aligned}
\label{eq:eif-estimator}
\end{equation}

A similar EIF-based estimator can be used for $\phi_{k}$ to estimate
the $\pi_{k}$ terms in Equation \ref{eq:K-decomposition}. This estimator
is based on the following nuisance functions for estimation (see
Supplementary Material K for details):

\begin{align*}
\gamma_{k}\left(X,\bar{Z}_{k}\right) & \triangleq\mathbb{E}\left[M_{k+1}\mid X,\bar{Z}_{k},\overline{1}_{k+1}\right],\\
\gamma_{j}\left(X,\bar{Z}_{j}\right) & \triangleq\mathbb{E}\left[\gamma_{j+1}\left(X,\bar{Z}_{j+1}\right)\mid X,\bar{Z}_{k},\overline{1}_{j+1}\right]\forall j\in[k-1].
\end{align*}

Next, following \citet[p. 15]{Kennedy2022}, let $\mathbb{IF}:\Psi\rightarrow L_{2}(\mathbb{P})$
denote the operator mapping the functionals $\{\Delta_{k},\pi_{k},\eta_{k}\}:\mathcal{P}\rightarrow\mathbb{R},$
$\forall k\in[K]$ to their respective influence functions under the
nonparametric model $\mathcal{P}$. Because the $(\Delta_{k},\tau_{k})$
components of the decomposition are linear in $\psi_{km_{k}}$, by
linearity of the EIF, $(\mathbb{IF}(\Delta_{k}),\mathbb{IF}(\tau_{k}))$
can be expressed as linear combinations of $\varphi_{km_{k}}(O)$.
In particular, $\mathbb{IF}(\tau_{k})=\varphi_{(k+1)1}(O)-\varphi_{k,0}(O)$
and $\mathbb{IF}(\Delta_{k})=\varphi_{(k+1)0}(O)-\varphi_{k,0}(O)$.
The EIFs of $\eta_{k}$ and $\theta_{k}$, $\forall k\in[K]$ under
$\mathcal{P}$ are derived as in Theorem \ref{prop:eif-cov-theta}:

\begin{theorem}\label{prop:eif-cov-theta}

The EIFs of $\eta_{k}$, $\theta_{k}$ $\forall k\in[1,\dots,K]$
under $P$ are given, respectively, by

\begin{align*}
\mathbb{IF}(\eta_{k}) & =\mathbb{IF}(\tau_{k-1})-\mathbb{IF}(\Delta_{k-1})-\tau_{k}\mathbb{IF}(\pi_{k})-\pi_{k}\mathbb{IF}(\tau_{k}),\\
\mathbb{IF}(\theta_{k}) & =\mathbb{IF}(\Delta_{k})\prod_{j=1}^{k}\pi_{j}+\Delta_{k}\sum_{j=1}^{k}\mathbb{IF}(\pi_{j})\prod_{\substack{l=1\\
l\neq j
}
}^{k}\pi_{l}+\mathbb{IF}(\eta_{k})\prod_{j=1}^{k-1}\pi_{j}+\eta_{k}\sum_{j=1}^{k-1}\mathbb{IF}(\pi_{j})\prod_{\substack{l=1\\
l\neq j
}
}^{k-1}\pi_{l},
\end{align*}

for $k\in\{1,\dots K\}$, with $\theta_{0}=\Delta_{0}$, and where
$\mathbb{RIF}(\phi)=\mathbb{IF}(\phi)+\phi,$ denotes the recentered
EIF of a parameter (about the truth). Their corresponding EIF-based
estimators are (see Supplementary Material K for derivations):

\begin{align*}
\hat{\eta}_{k}^{\text{eif}} & =\mathbb{\widehat{\mathbb{RIF}}}(\tau_{k-1})-\widehat{\mathbb{RIF}}(\Delta_{k-1})-\hat{\tau}_{k}\widehat{\mathbb{RIF}}(\pi_{k})-\hat{\pi}_{k}\widehat{\mathbb{RIF}}(\tau_{k})+\hat{\pi}_{k}\hat{\tau}_{k},\\
\hat{\theta}_{k}^{\text{eif}} & =\widehat{\mathbb{RIF}}(\Delta_{k})\prod_{j=1}^{k}\hat{\pi}_{j}+\hat{\Delta}_{k}\sum_{j=1}^{k}\widehat{\mathbb{RIF}}(\pi_{j})\prod_{\substack{l=1\\
l\neq j
}
}^{k}\hat{\pi}_{l}+\widehat{\mathbb{RIF}}(\eta_{k})\prod_{j=1}^{k-1}\hat{\pi}_{j}+\hat{\eta}_{k}\widehat{\mathbb{RIF}}(\pi_{j})\prod_{\substack{l=1\\
l\neq j
}
}^{k-1}\hat{\pi}_{l}\\
 & \;-k\hat{\Delta}_{k}\prod_{j=1}^{k}\hat{\pi}_{j}-(k-1)\hat{\eta}_{k}\prod_{j=1}^{k-1}\hat{\pi}_{j}.
\end{align*}

where $\widehat{\mathbb{RIF}}(\phi)=\widehat{\mathbb{IF}}(\phi)+\phi,$
and $\widehat{\mathbb{IF}}(\phi)$ denotes the influence function
of a parameter evaluated at estimates of its component nuisance functions (see Supplementary Material K for derivations).
\end{theorem}

When machine learning estimators are used to compute the nuisance
functions, in order to ensure the convergence rates outlined in Theorem
\ref{prop:semiparametric-efficiency} below, one could assume Donsker-type
conditions for the nuisance function estimators, which restricts the
set of estimators available to use. Alternatively, to expand the
class of estimators that can be used for estimating the nuisance functions,
sample-splitting can be used. In particular, \citet{Chernozhukov2017}
suggest a ``cross-fitting'' procedure, which comprises the following
steps: (1) Randomly split data into $J$ folds: $\{S_{1},...S_{J}\}$;
(2) For each fold $S_{j}$, use the remaining ($j-1$) folds (training
sample) to fit a flexible machine-learning model for each of the nuisance
functions involved in the estimating equations; (3) For each observation
in $j$ (estimation sample), use estimates of the above models to
construct a set of estimated RIF functions for $\Delta_{k}\forall k\in\{0,\dots,K-1\}$, and for $(\pi_{k},\tau_{k},\eta_{k},\theta_{k})\forall k\in[K]$;
(4) Compute an estimate of the decomposition components by averaging
the estimated RIF functions across all subsamples $S_{1}$ through
$S_{J}$. When all nuisance functions are estimated via data-adaptive
methods and cross-fitting, the semiparametric efficiency of $\theta_{k}^{\mathrm{rEIF}}$
is given in the following Theorem:

\begin{theorem}[Semiparametric efficiency]\label{prop:semiparametric-efficiency}

Under Assumption \ref{assu:Positivity}, and under suitable regularity
conditions \citep[e.g.][]{Chernozhukov2018}, then $\hat{\theta}_{k}^{\text{eif}}$
is semiparametric efficient if $\sum_{j=k}^{k+1}\bigg[\sum_{l=0}^{j}R_{n}\big(\hat{\pi}_{l1})R_{n}\big(\hat{\mu}_{l0}^{j})\bigg]+\sum_{j=0}^{k-1}\bigg[R_{n}\big(\hat{\pi}_{j1})R_{n}\big(\hat{\mu}_{j0}^{k-1})+R_{n}\big(\hat{\pi}_{j1})R_{n}\big(\hat{\mu}_{j1}^{k-1})\bigg]+\sum_{j=0}^{k}\bigg[\sum_{l=0}^{j}R_{n}\big(\hat{\pi}_{l1})R_{n}\big(\hat{\gamma}_{l}^{j})\bigg]=o(n^{-1/2})$,
where $R_{n}\big(\cdot\big)$ denotes a mapping from a nuisance function
to its $L_{2}(P)$ convergence rate, and where $\hat{\mu}_{l0}^{K+1}\triangleq\hat{\mu}_{l1}^{K}$.
\end{theorem}

To gain some intuition for the result in Theorem \ref{prop:semiparametric-efficiency},
we can focus on $\theta_{1}=\pi_{\textup{1}}\Delta_{\textup{1}}+\eta_{\textup{1}}$,
i.e., the MPSE through $M_{1}$ when $K=1$. Note that estimation
of $\theta_{1}=\pi_{\textup{1}}\Delta_{\textup{1}}+\eta_{\textup{1}}$
requires estimating the following decomposition components: $(\pi_{\textup{1}},\Delta_{\textup{1}},\tau_{0,},\Delta_{0},\tau_{1})$.
To estimate these components, it suffices to estimate the following
quantities: $(\phi_{1},\psi_{01},\psi_{00},\psi_{10},\psi_{11})$.
In order for $\hat{\theta}_{1}^{\text{eif}}$ to be semiparametric
efficient, we require that the estimators employed for the set $(\phi_{1},\psi_{01},\psi_{00},\psi_{10},\psi_{11})$,
i.e., $(\hat{\phi}_{1}^{\text{eif}},\hat{\psi}_{01}^{\text{eif }},\hat{\psi}_{00}^{\text{eif}},\hat{\psi}_{10}^{\text{eif}},\hat{\psi}_{11}^{\text{eif }}),$
are themselves semiparametric efficient. Thus, a sufficient (but not
necessary) condition in order for $\hat{\theta}_{1}^{\text{eif}}$
to obtain the semiparametric efficiency bound is if, for any two nuisance
functions involved in $(\hat{\phi}_{1}^{\text{eif}},\hat{\psi}_{01}^{\text{eif }},\hat{\psi}_{00}^{\text{eif}},\hat{\psi}_{10}^{\text{eif}},\hat{\psi}_{11}^{\text{eif }})$,
the product of their convergence rates is $o(n^{-1/2})$. In this
way, $\hat{\theta}_{1}^{\text{eif}}$ will obtain the semiparametric
efficiency bound if all of its constituent nuisance functions converge
at a rate faster than $n^{-1/4}$ (although it will also obtain
the efficiency bound under a variety of alternative conditions).

Under the DML estimation procedure, inference for all components of the MPSE decomposition is conducted using the
efficient influence functions (EIFs) of the target parameters. When nuisance
functions are estimated using cross-fitting and data-adaptive methods, the resulting
DML estimators based on the EIFs derived above are asymptotically linear and converge at a $\sqrt{n}$-rate. In
particular, under standard regularity conditions \citep{Chernozhukov2017,Kennedy2022},
each estimator admits the expansion
\[
\sqrt{n}(\hat\theta - \theta)
=
\sqrt{n}(\mathbb{P}_n - \mathbb{P})\big[\phi(O)\big] + o_p(1),
\]
where $\phi(O)$ denotes the corresponding EIF. As a result, $\hat\theta$ is
asymptotically normal with variance $\mathbb{E}[\phi(O)^2]$, which is consistently
estimated by the empirical variance of the estimated EIF. Wald-type confidence
intervals reported in the main text are constructed using this plug-in variance
estimator. For example, inference on $\tau_1$ can be conducted
by estimating $\mathbb{P}_{n}[\big(\hat{\psi}_{11}^{\text{EIF}}-\hat{\psi}_{10}^{\text{EIF}}\big)^{2}]/n$.

\section{A Simulation Study\label{sec:Simulation}}

In this section, I evaluate the finite-sample performance of my two
estimation procedures via a simulation experiment. Supplementary Material \ref{app:Simulation} provides further details. Specifically, I
compare how the DML estimator proposed in Section \ref{sec:Estimation}
(as well as the parametric, RWR estimator described in Appendix \ref{App:rwr})
perform under different degrees of misspecification. Specifically, I consider the setting of two causally ordered monotonic mediators, with post-treatment confounding. I generate simulations of
observed data $O=(X_{1},X_{2},A,Z,M,Y)$ as follows:

\[
\begin{aligned}U_{1},U_{2},U_{3},U_{4} & \sim\text{MVN}\big(0_{4},I_{4}\big)\\
X_{1} & \sim\text{N}\big((U_{1},U_{2},U_{3},U_{4})\beta_{X_{1}},1\big)\\
X_{2} & \sim\text{N}\big((U_{1},U_{2},U_{3},U_{4})\beta_{X_{2}},1\big)\\
A & \sim\text{Bern}\big(g^{-1}\big[(1,X)\beta_{A}\big]\big)\\
Z|A=1 & \sim\text{N}\big((1,X)\beta_{Z},1\big)\\
M_1|A=1\sim\text{Bern}\big(g^{-1}\big[(1,X,Z)\beta_{M_1}\big]\big) & \quad M|A=0=0\\
M_2|M_1=1\sim\text{Bern}\big(g^{-1}\big[(1,X,Z)\beta_{M_2}\big]\big) & \quad M_2|M_1=0=0\\
Y & \sim\text{N}\big((1,A,X,AZ,AM_1, A M_1 M_2)\beta_{Y},1\big).
\end{aligned}
\]

The coefficients $(\beta_{X_{1}},\beta_{X_{2}},\beta_{Y})$ are drawn
from a $\text{Unif}\big(-1,1\big)$ distribution, while the coefficient $\beta_{A}$
is drawn from a $\text{Unif}\big(-0.5,0.5\big)$ distribution. $g(\cdot)$ is a link function as described in the Supplementary Material. Further, $X = (X_1,X_2)$. In
order to test how the DML and RWR methods perform when the relevant
models are misspecified, I also construct transformations of the observed
covariates ($X^{*}$) as follows, employing a similar
setup to \citet{Kang2007}:

\[
\begin{aligned}X_{1}^{*} & =(\text{exp}(X_{1}/2)-1)^{2}\\
X_{2}^{*} & =X_{2}/(1+\exp(X_{2}))+10\\
\end{aligned}
\]

For each simulated dataset, I construct two estimates of the path-specific
effects ($\theta_{0},\theta_{1},\theta_{2}$) via the RWR and
DML procedures described in Section \ref{sec:Estimation}. Standard
errors for the coverage rates are computed via the estimated variance
of the estimated EIFs for the DML approach, and via the nonparametric
bootstrap with 250 replications for the RWR procedure. For the DML
estimator, for each component involved in the decomposition, I construct
a Neyman-orthogonal ``signal'' using its EIF. The recentered EIFs
for each component are shown in Supplementary Material \ref{app:Simulation}.

I run 1000 replications of this DGP and compute the average bias and coverage of nominal 95\%
confidence intervals for sample sizes of 1000, 1500 and 2000 and using
either the ``correctly specified'' covariates $\left(X_{1},X_{2},Z\right)$
and the ``incorrectly specified'', transformed versions $\left(X_{1}^{*},X_{2}^{*},Z\right)$.
I calculate the true value of $\theta_{1}$ by recovering the true
values of the parameter set ($\theta_{0}$, $\theta_1,\theta_2$)
in each Monte Carlo simulation. 

Figure \ref{fig:Simulation} in Supplementary Material B presents the results from this simulation
exercise. Under correctly specified models, the DML and RWR estimators
perform similarly. In particular, both estimators exhibit negligible finite-sample bias
for $\theta_0$, $\theta_1$, and $\theta_2$, with absolute biases on the
order of $10^{-3}$ or smaller, and coverage rates close to the nominal
$95\%$ level. Under incorrectly specified models, however, the performance of the two
estimators diverges sharply.
The DML estimator remains stable: even when supplied with a
misspecified feature space, absolute biases remain small for all
MPSEs, and
coverage rates remain close to nominal, typically between $92\%$ and
$95\%$. By contrast, the RWR
estimator performs much more poorly, displaying a large amount of
bias that in fact grows with the sample size, a large RMSE, and coverage
rates that are not close to nominal. Overall, when the relevant models are correctly specified, both the
parametric and semiparametric approaches perform well.
Under misspecification, however, the advantages of the semiparametric estimation strategy become clear.

\section{Empirical Analysis\label{sec:Empirical-Analysis}}

To illustrate my approach empirically, I draw on data from the National
Longitudinal Survey of Youth 1997 (NLSY97). I parse out the direct
effect of high school graduation on adult earnings and its indirect
or continuation effects via (i) college attendance, (ii) college graduation,
and (iii) graduate school attendance. My analytic sample comprises
$N=7,305$ respondents.

I construct four types of variables: educational transitions, adult
earnings, a set of confounders for the effect of high school graduation
on subsequent transitions and earnings, and a single set of intermediate
confounders for the effect of college completion on subsequent transitions
and earnings. My educational transition variables contain a binary
treatment denoting whether a respondent had graduated high school
by age 22, and three binary mediators denoting whether the respondent
had attended a 4-year college by age 22, whether the respondent had
received a BA degree by age 29, and whether the respondent had enrolled
in a graduate level program by age 29, respectively. I assume that
all individuals who make a given educational transition have made
all previous educational transitions. Thus, by construction, my coding
strategy disallows for cases which violate the monotonicity assumption.\footnote{Assuming away cases in which an individual makes a particular educational
transition without having made \textit{all} previous transitions serves
as a reasonable approximation to reality. Among the set of individuals
who have non-missing earnings information in the NLSY97 (i.e., those
who comprise my analytic sample), 94\% of individuals observed to
attend graduate school by age 29 also completed a BA by age 29; 93\%
of respondents who completed a BA by age 29 had attended a 4-year
college by age 22 (6\% of those who completed a BA by age 29 first
attended a 4-year college between ages 23 and 26 inclusive), and 99\%
of respondents who attended a 4-year college by age 22 had also completed
high school.}

My outcome of interest is logged average annual earnings at ages 32–36. 
For each respondent and each survey year in this age range, I construct total 
annual labor-market income by summing wage, salary, and business income, and 
then compute the  (logged) average of these annual totals across this age range. This multi-year average yields a more stable measure of early-adult earnings that smooths over short-term income fluctuations.
\footnote{Age 36 is the latest age at which earnings are consistently observed 
in the NLSY97. Because earnings gains associated with graduate education may 
materialize later in the life course, I also assess robustness to using a later 
earnings window. In Supplementary Material H, I re-estimate the full MPSE 
decomposition using the NLSY79 cohort and measure logged earnings over ages 
35--44. The qualitative patterns are highly similar, and the contribution of 
the graduate-education pathway remains modest, for reasons discussed further in 
the supplement.}
Earnings are
adjusted for inflation to 2023 dollars using the personal consumption
expenditures (PCE) index. After dropping respondents with missing
earnings information, I accommodate those with zero earnings by adding
a small constant of $\$1,000$ to observed earnings (though in
Supplementary Material F, I replicate my main analyses under alternative
definitions of earnings).

In an effort to satisfy the sequential ignorability assumption (Assumption
\ref{assu:SI}), I include a large array of covariates in my models. This set of covariates is more expansive
than those used in previous observational studies of returns to education (see in particular
\citealp{Scott-Clayton}). In particular, in addition to including information
on respondent demographics (gender, race, ethnicity, age in 1997),
and observed pre-college performance such as overall high school GPA
and test score on the Armed Services Vocational Aptitude Battery (ASVAB),
I include detailed information on socioeconomic background. Since
my proposed decomposition also facilitates the inclusion of a distinct
set of observed intermediate confounders for each transition, I include two postsecondary characteristics
($Z$) to adjust for confounders of the effect of BA completion and
graduate school attendance on earnings: field of study and
college GPA. To assess the robustness of my main conclusion
to forms of unobserved confounding, in Supplementary Material C, I produce a set of ``bias-corrected'' estimates of the decomposition components under certain assumptions about the nature of the confounding.

A large proportion (just under $50\%$) of respondents are missing
information on covariates $X$ and $Z$. For my main analyses, I impute
missing values on these covariates via multiple imputation to increase
efficiency, but in Supplementary Material E, I replicate these
analyses restricted to the sample of respondents with complete information.
This exercise produces substantively similar results (for covariate
means for each of these analytic samples, see Supplementary Material D).
After constructing the analytical sample, I apply both the DML estimator
described in Section \ref{sec:Estimation} as well as a parametric,
regression-with-residuals (RWR) algorithm (described in Supplementary Material A)
to implement the proposed decomposition. For the DML approach, I estimate
all nuisance functions, using a super learner composed of the Lasso
and random forest and, following \citet{Chernozhukov2017}, use five-fold
cross-fitting. Supplementary Material J
gives further details about the particular models required given my
assumed data generation process.

I conduct two sets of analyses. First, I assume that positivity holds at every transition of the
education sequence, i.e., the estimated propensity of each transition lies strictly within
$(0,1)$. Second, I relax the positivity assumption by allowing units to have transition
propensities of exactly zero or one. In this scenario, the MPSE decomposition is unidentified, and therefore requires 
a redefinition of the target quantity as the decomposition among the common-support subpopulation across transitions. I present results under the first approach (assuming
positivity) in the main text, and compare them with results from the common-support, trimmed
sample under the second approach (relaxing positivity) in Supplementary Material~G. In practice the
two approaches yield substantively similar conclusions: only the
college-completion pathway attenuates modestly once units with extreme propensity scores are
excluded.

Figure \ref{fig:Decomp} shows my estimates of the average total effect
(ATE) on log earnings and its direct and continuation components under
both the DML and RWR procedures. Standard errors for RWR estimates are obtained via the non-parametric bootstrap, while standard errors for DML estimates are obtained via the variance of the estimated EIF for each MPSE. Both procedures return similar estimates,
though deviate in the estimated magnitude of MPSE $\theta_{1}$, and
DML estimates come expectedly with a significantly greater amount
of precision. The first column shows that the estimated ATE of graduating
high school on log earnings under DML (RWR) is $0.68$ ($0.63$),
which implies an earnings premium of approximately $96\%$. The next
two columns indicate that the vast majority ($69\%$ under DML and
$75\%$ under RWR) of the ATE operates directly, i.e. net of college
attendance, BA completion and graduate school attendance (MPSE $\theta_{0}$,
$A\rightarrow Y$). Specifically, high school graduates who do not
proceed to college can be expected to earn on average $0.47$
log earnings more than high school non-completers under DML,
an earnings premium of $59\%$ (RWR returns an identical estimate).

While the majority of the ATE is explained by the direct effect, a non-trivial portion occurs through mediation effects through later transitions. Under DML, the continuation effects of high school graduation via
college attendance without BA completion (MPSE $\theta_{1}$, $A\rightarrow M_{1}\rightarrow Y$)
and via BA completion without graduate school participation (MPSE
$\theta_{2}$, $A\rightarrow M_{1}\to M_{2}\rightarrow Y$) both mediate
roughly $15\%$ of the ATE, and correspond to an earnings premium
of approximately $10\%$. The RWR estimate of $\theta_{1}$ is notably
lower at $0.03$ and is also imprecisely estimated. Under both estimation
procedures, the continuation effect via graduate school attendance ($A\rightarrow M_{1}\rightarrow M_{2}\rightarrow M_{3}\rightarrow Y$)
is very small and fails to reach conventional levels of significance.
In sum, the total effect of high school graduation on earnings is
determined overwhelmingly by its direct effect
on earnings.

Table \ref{tab:Decomp-2} shows DML and RWR estimates of the various
components (the direct effects ($\Delta_{k}$), probabilities ($\pi_{k}$)
and covariance terms ($\eta_{k}$)) that constitute the continuation
effects $\theta_{k}$. Several points are of note. First, the components
in the table offer insights into the economic and educational returns
to different educational stages. The direct effects of each educational
transition ($\Delta_{k}$) are highly variable: they are largest for
high school graduation and for college completion (both at $0.47$
under DML and RWR), and lowest for college attendance and graduate school
participation (at $0.2$ and $0.14$, respectively, under DML). Note that the payoff
to graduate school attendance could be depressed by the fact that
I observe individuals at a maximum age of only 36, if graduate school
earnings premia materialize only much later in the life course. The
counterfactual continuation probabilities ($\pi_{k}$) also provide
insight into barriers in educational participation. In particular,
even if an individual were to complete high school (possibly contrary to fact), that individual
would have under a $50\%$ chance of continuing to a 4-year college without further intervention to increase individuals' college application,
admissions and enrollment rates. Further, even if individuals were to counterfactually both complete high school and attend a 4-year college, only a very
small proportion ($\pi_{1}\cdot\pi_{2}=0.24$) would be expected to
complete their BA degree without further intervention at the college-level.

Second, the fine-grained nature of the MPSE decomposition enables
us to trace the continuation effects to their constituent components.
In particular, while the direct effect of high school completion is
comparable to the direct effect of BA graduation on earnings, the overall continuation effect via BA completion that it informs
(MPSE $\theta_{2}$, $A\rightarrow M_{1}\to M_{2}\rightarrow Y$)
is small because $\theta_{2}$
is approximately (plus the small value of $\eta_{2}$) equal to $\Delta_{2}$
scaled by the product $\pi_{1}\cdot\pi_{2}=0.24$. In words, despite
the relatively large direct effect of BA completion on earnings, given
individuals' low counterfactual probability of BA completion, this
transition is not an important mediating pathway of the total effect
of high school completion on earnings. The result is that college
attendance without completion mediates high school graduation's earnings
effects as much as BA completion, despite the fact that college attendance
without completion yields a much smaller earnings return for high
school graduates than BA completion without graduate school attendance
does for college enrollees.

One instructive point of comparison for these results is instrumental variable (IV) estimates of returns
to years of schooling, typically estimated in the range of $6\%$ to $12\%$ \citep{Angrist1991,Angrist1992,Kane1993,Card1994,Ashenfelter1997,Angrist2011}.
While my estimate of the overall return to high school graduation
($\tau_{0})$ could appear large in this light, several factors could
reconcile this difference. First, $\tau_{0}$ captures the direct \textit{and}
continuation effects of high school completion (whereas IV estimates
of schooling returns capture schooling's direct effects). Further,
$\tau_{0}$ captures the effect of multiple additional years of schooling
(as the high school graduates and high school non-completers that
form the comparison group differ by multiple years of schooling),
as opposed to a single year's additional return. In fact, we can more
directly compare my DML estimate of the direct return to high school
graduation ($\Delta_{0}$) of $0.46$ (corresponding to an earnings
premium of $58\%$) using the fact that, in the NLSY97, high school
non-completers attained on average $3.7$ fewer years of schooling
than high school completers. An IV estimate of $12\%$, for example,
would therefore imply an earnings return to $3.7$ additional years
of approximately $52\%$ - broadly in line with my result. Still,
to assess the robustness of the above findings to potential violations
of Assumption \ref{assu:SI} (Sequential Ignorability), I implement
a sensitivity analysis in Supplementary
Material C. Under the stated assumptions about the pattern of unobserved
confounding, my primary finding that the ATE of high school graduation
is overwhelmingly mediated via its direct effect remains highly robust
to unobserved confounding.

\section{Conclusion\label{sec:Conclusion}}

In this article, I have developed a causal mediation framework for
analyzing education effects on earnings. First, I have demonstrated
that the total effect of any level of education can be decomposed
into a direct effect and $K$ mutually exclusive ``continuation'' effects. All of these effects are identifiable under the assumption of sequential ignorability.
Importantly, this property allows for the effect of each educational transition
to be confounded by a distinct set of observed covariates---a property
which allows for weaker identification conditions compared with conventional
mediation-based decompositions of the ATE \citep{Miles2017,Zhou2022b}.

Several directions for future research follow naturally from the proposed framework.
First, in this paper I have considered a decomposition
of the average treatment effect for the case of binary monotonic mediators, but many
educational processes are more finely graded. Extensions to settings with categorical
or multivalued transitions—such as different types of postsecondary institutions,
fields of study, or graduate degrees—would further broaden the applicability of the
framework. Supplementary Material I outlines one such extension, but generalizing the framework to categorical and continuous mediators remains an open area for future work.

Second, although the MPSE decomposition relaxes cross-world assumptions and permits
observed intermediate confounding, it still relies on sequential ignorability. In
practice, this assumption may be difficult to satisfy fully in observational
settings, particularly when selection into successive educational transitions is
driven by unobserved traits such as motivation or ability. Developing alternative
identification results that, for example, exploit instrumental variables—long used
in the education literature to address selection into schooling—would be a
particularly promising extension of the MPSE framework.

Finally, while this paper emphasizes educational attainment, the monotonic structure
exploited here arises in many other domains characterized by state-dependent
transitions, such as family formation, health progression, or criminal justice
contact. This characteristic is particularly
salient in demographic phenomena. Certain demographic events are rigid in their
monotonicity by definition. For example, researchers
may be interested in discerning the degree to which positive effects
of marriage on outcomes such as earnings and life satisfaction are
undermined by the negative effects of divorce and separation (and,
in turn, their mitigation via re-marriage) \citep{Kenney2004,Sweeney2004} -- monotonic transitions. Similarly, the effect of parenthood on earnings can be seen
as operating directly, through the effect of having a first child
net of subsequent children, as well as operating indirectly through
the effects of having multiple children. A similar perspective may be taken in a criminal
justice context: the total effect of early-stage police contact (such
as being searched for contraband) on educational and socio-psychological
outcomes can be decomposed into path-specific effects via subsequent
arrest and incarceration \citep{Weaver2010,Kirk2013,Sugie2017}. Applying the MPSE framework to these contexts, and comparing the resulting
decompositions across domains, may yield new insights into how life transitions
shape later outcomes through both direct and sequential mechanisms.

\vspace{2em}
\newpage
\section{Tables and figures}\hspace{1em}

\begin{table}[h!]
\caption{Direct Effects ($\Delta_{k}$), Probabilities ($\pi_{k}$) and Covariance
Terms ($\eta_{k}$) Involved in Decomposition via Debiased Machine-Learning
(DML) and Regression-With-Residuals (RWR). \label{tab:Decomp-2}}
\centering
\begingroup\fontsize{9.2pt}{9.2pt}\selectfont\setlength{\tabcolsep}{3pt}
\renewcommand{\arraystretch}{1.5} 
\begin{tabular}{lllllllllll}
   \hline & $\Delta_0$ & $\Delta_1$ & $\Delta_2$ & $\Delta_3$ & $\pi_1$ & $\pi_2$ & $\pi_3$ & $\eta_1$ & $\eta_2$ & $\eta_3$  \\ \hline 
     DML & 0.466 & 0.200 & 0.444 & 0.138 & 0.427 & 0.555 & 0.314 & 0.007 & 0.004 & -0.002 \\ 
   & (0.057) & (0.031) & (0.034) & (0.040) & (0.008) & (0.016) & (0.011) & (0.007) & (0.008) & (0.008) \\ 
   RWR & 0.469 & 0.116 & 0.466 & 0.163 & 0.374 & 0.515 & 0.159 & -0.016 & 0.083 & 0.028 \\ 
   & (0.115) & (0.093) & (0.097) & (0.114) & (0.015) & (0.077) & (0.022) & (0.007) & (0.015) & (0.051) \\ 
   \hline
\end{tabular}

\vspace{.75em} 

\fontsize{10pt}{12pt}\selectfont 
\noindent
Note: The $\Delta_{k}$ parameters capture the average effect of completing
the $k$th mediator \textit{but no subsequent mediator} on earnings,
relative to completing the $k-1$th mediator. For instance, $\Delta_{0}$
denotes the effect of completing high school ($M_{1}$) but not attending
college nor, under Assumption \ref{assu:Monotonicity}, completing
any subsequent mediators, relative to attending high school but not
completing it ($M_{0}\equiv A$). The $\pi_{k}$ terms capture the
average of individuals' counterfactual completion status of the $k$th
mediator under completion of all prior mediators $M_{0},\dots M_{k-1}$.
For example, $\pi_{1}$ denotes individuals' average counterfactual
college attendance, after---possibly contrary to fact---their completion
of high school. Finally, the $\eta_{k}$ terms refer to the covariance
between individuals' own counterfactual completion status of the $k$th
mediator, and their own ``gross'' effect of completing the $k$th
mediator on earnings. To recall, the ``gross'' effect of the $k$th
mediator captures the effect of completing that mediator, relative
to completing only the $k-1$th mediator, irrespective of whether
that effect operates directly (net of subsequent mediators) or via
subsequent transitions. For example, $\eta_{1}$ denotes the covariance
between each individual's counterfactual college attendance status
and their gross effect of college attendance on earnings.
\endgroup
\end{table}

\begin{figure}[p]
\centering
\includegraphics[width=\linewidth,keepaspectratio]{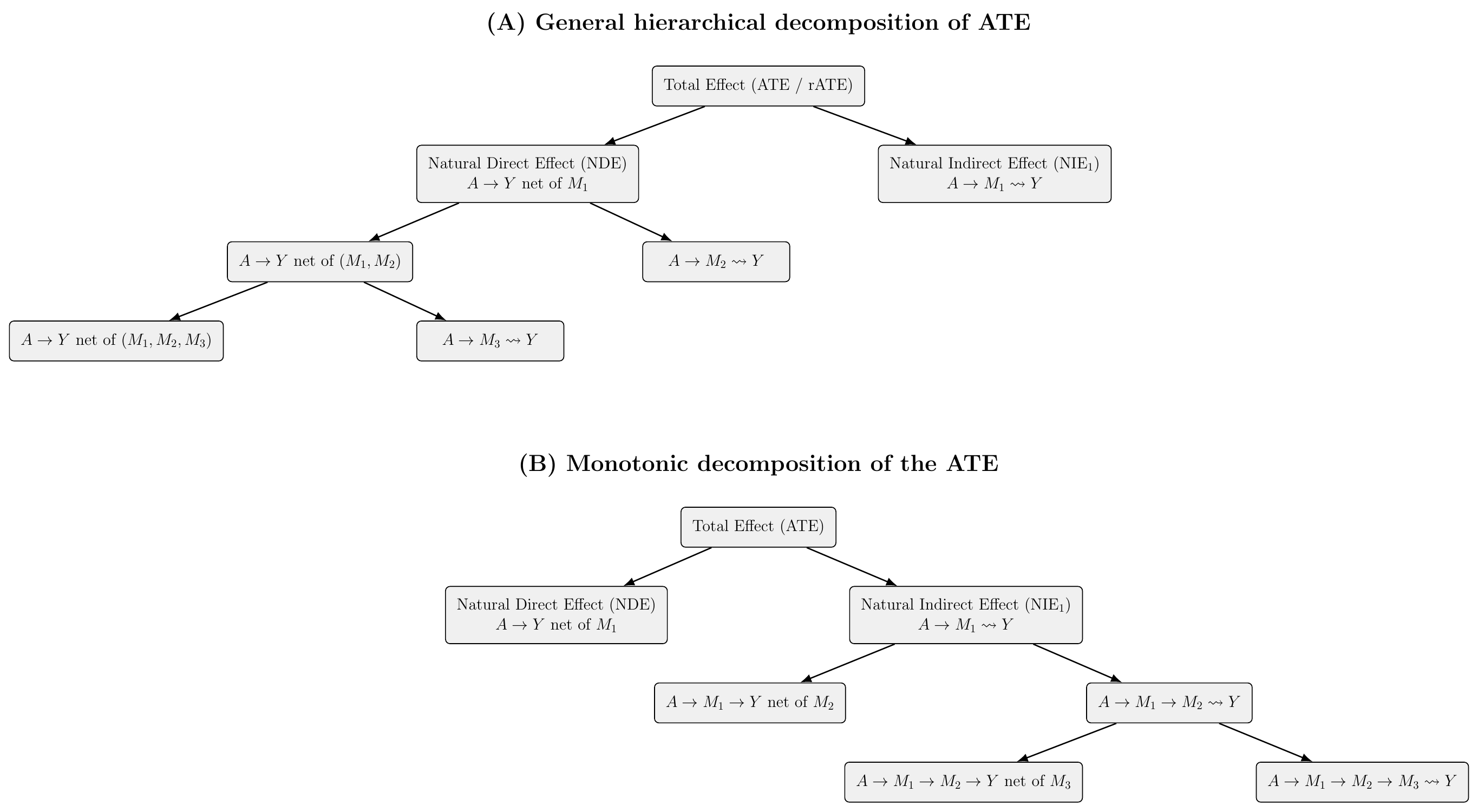}
\vspace{0.5em}
\caption{Nested decomposition of the Average Treatment Effect (ATE) into direct and sequential indirect pathways: general hierarchical decomposition (Panel A) and proposed monotonic decomposition (Panel B).}
\label{fig:decomp-tree}
\end{figure}

\newpage
\begin{figure}[h!]
\begin{centering}
\includegraphics[scale=0.6]{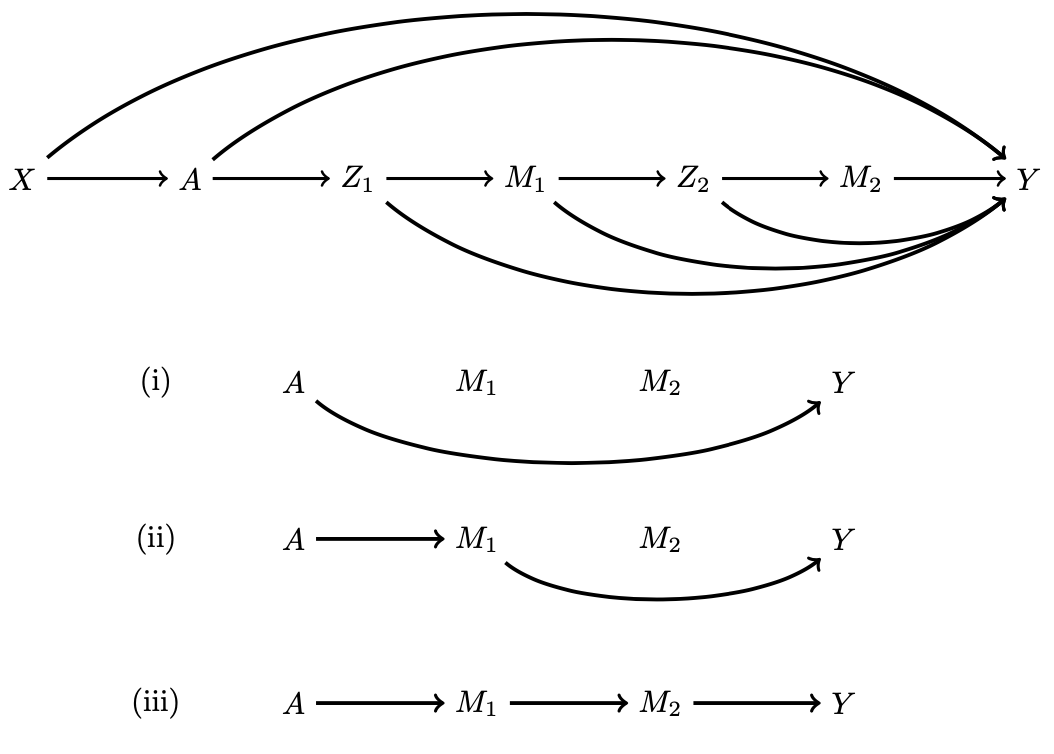}
\par\end{centering}
\caption{Causal Relationships Between Two Monotonic Mediators Shown in a Directed
Acyclic Graph (DAG) and the 3 Monotonic Path Specific Effects (MPSEs).
$A$ denotes an initial transition of interest, $Y$, an outcome,
and $M_{1}$ and $M_{2}$ are two causally ordered, monotonic mediators.
The set $(X,Z_{1},Z_{2})$ captures pre-treatment and intermediate
confounders. \label{fig:MPSEs}}
\end{figure}

\newpage
\begin{figure}[h!]
\begin{centering}
\includegraphics[scale=0.6]{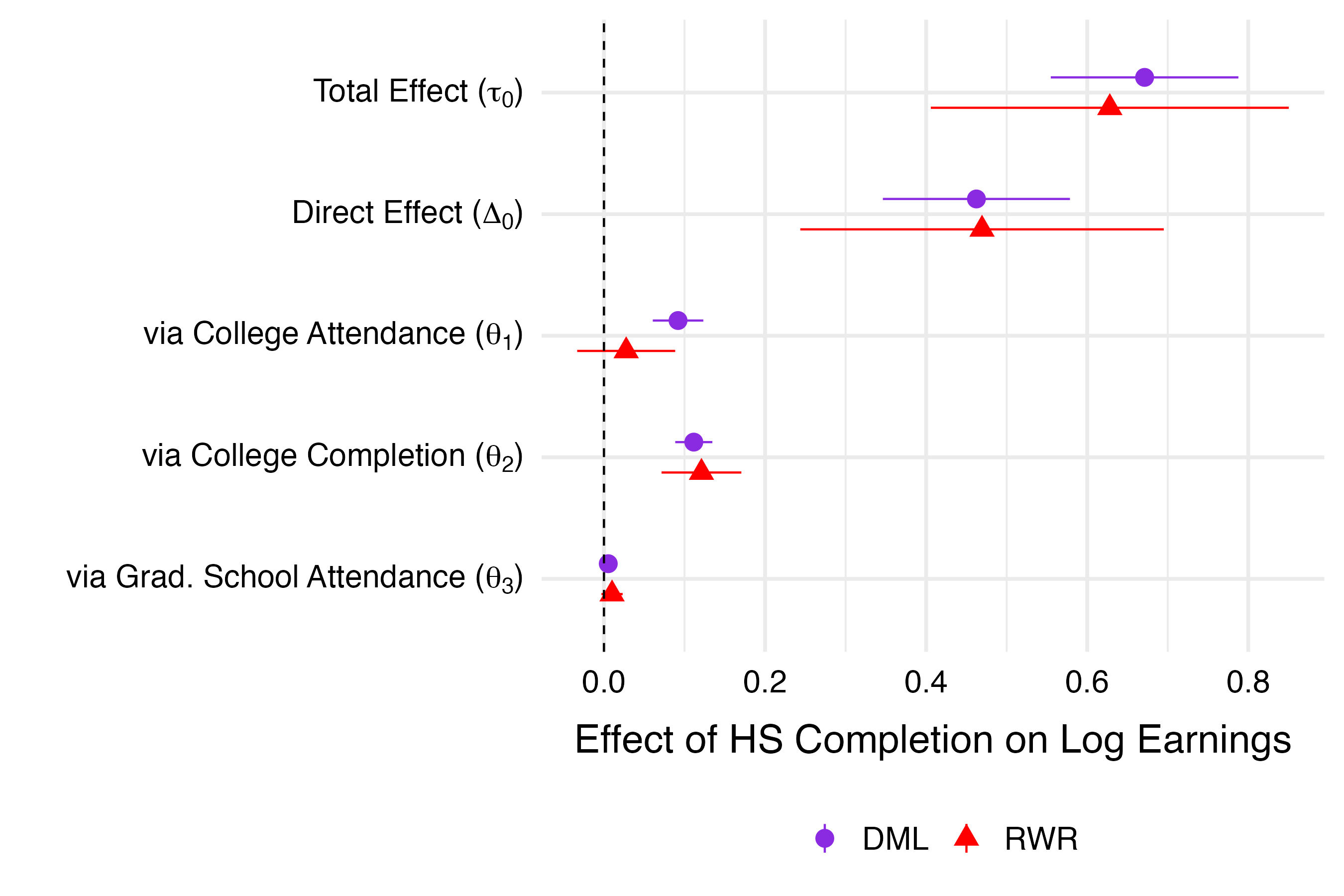}
\par\end{centering}
\caption{Decomposition of the Average Total Effect (ATE) of High School Graduation
on Logged Earnings via Debiased Machine-Learning (DML) and Regression-With-Residuals
(RWR).\label{fig:Decomp}}
\end{figure}

\FloatBarrier
\newpage



\bibliographystyle{plainnat} 
\bibliography{ed_decomp6}       

\begin{thebibliography}{68}
\providecommand{\natexlab}[1]{#1}
\providecommand{\url}[1]{\texttt{#1}}
\expandafter\ifx\csname urlstyle\endcsname\relax
  \providecommand{\doi}[1]{doi: #1}\else
  \providecommand{\doi}{doi: \begingroup \urlstyle{rm}\Url}\fi

\bibitem[Acharya et~al.(2016)Acharya, Blackwell, and
  Sen]{acharya2016explaining}
Avidit Acharya, Matthew Blackwell, and Maya Sen.
\newblock Explaining causal findings without bias: Detecting and assessing
  direct effects.
\newblock \emph{American Political Science Review}, 110\penalty0 (3):\penalty0
  512--529, 2016.

\bibitem[Albert and Nelson(2011)]{Albert2011}
Jeffrey~M Albert and Suchitra Nelson.
\newblock Generalized causal mediation analysis.
\newblock \emph{Biometrics}, 67:\penalty0 1028--1038, 2011.
\newblock ISSN 0006-341X.

\bibitem[Angrist and Chen(2011)]{Angrist2011}
Joshua~D Angrist and Stacey~H Chen.
\newblock Schooling and the vietnam-era gi bill: Evidence from the draft
  lottery.
\newblock \emph{American Economic Journal: Applied Economics}, 3:\penalty0
  96--118, 2011.
\newblock ISSN 1945-7782.

\bibitem[Angrist and Krueger(1991)]{Angrist1991}
Joshua~D Angrist and Alan~B Krueger.
\newblock Does compulsory school attendance affect schooling and earnings?
\newblock \emph{The Quarterly Journal of Economics}, 106:\penalty0 979--1014,
  1991.
\newblock ISSN 1531-4650.

\bibitem[Angrist and Krueger(1992)]{Angrist1992}
Joshua~D Angrist and Alan~B Krueger.
\newblock The effect of age at school entry on educational attainment: an
  application of instrumental variables with moments from two samples.
\newblock \emph{Journal of the American statistical Association}, 87:\penalty0
  328--336, 1992.
\newblock ISSN 0162-1459.

\bibitem[Ashenfelter and Zimmerman(1997)]{Ashenfelter1997}
Orley Ashenfelter and David~J Zimmerman.
\newblock Estimates of the returns to schooling from sibling data: Fathers,
  sons, and brothers.
\newblock \emph{Review of Economics and Statistics}, 79:\penalty0 1--9, 1997.
\newblock ISSN 0034-6535.

\bibitem[Avin et~al.(2005)Avin, Shpitser, and Pearl]{Avin2005}
Chen Avin, Ilya Shpitser, and Judea Pearl.
\newblock Identifiability of path-specific effects.
\newblock 2005.

\bibitem[Bird et~al.(2021)Bird, Castleman, Denning, Goodman, Lamberton, and
  Rosinger]{Bird2021}
Kelli~A Bird, Benjamin~L Castleman, Jeffrey~T Denning, Joshua Goodman, Cait
  Lamberton, and Kelly~Ochs Rosinger.
\newblock Nudging at scale: Experimental evidence from fafsa completion
  campaigns.
\newblock \emph{Journal of Economic Behavior \& Organization}, 183:\penalty0
  105--128, 2021.
\newblock ISSN 0167-2681.

\bibitem[Black et~al.(2023)Black, Denning, Dettling, Goodman, and
  Turner]{Black2023}
Sandra~E Black, Jeffrey~T Denning, Lisa~J Dettling, Sarena Goodman, and
  Lesley~J Turner.
\newblock Taking it to the limit: Effects of increased student loan
  availability on attainment, earnings, and financial well-being.
\newblock \emph{American Economic Review}, 113:\penalty0 3357--3400, 2023.
\newblock ISSN 0002-8282.

\bibitem[Bleemer(2022)]{Bleemer2022}
Zachary Bleemer.
\newblock Affirmative action, mismatch, and economic mobility after
  california’s proposition 209.
\newblock \emph{The Quarterly Journal of Economics}, 137:\penalty0 115--160,
  2022.
\newblock ISSN 0033-5533.

\bibitem[Bodory et~al.(2022)Bodory, Huber, and Lafférs]{Bodory2022}
Hugo Bodory, Martin Huber, and Lukáš Lafférs.
\newblock Evaluating (weighted) dynamic treatment effects by double machine
  learning.
\newblock \emph{The Econometrics Journal}, 25:\penalty0 628--648, 2022.
\newblock ISSN 1368-4221.

\bibitem[Brand and Xie(2010)]{Brand2010}
Jennie~E Brand and Yu~Xie.
\newblock Who benefits most from college? evidence for negative selection in
  heterogeneous economic returns to higher education.
\newblock \emph{American sociological review}, 75:\penalty0 273--302, 2010.
\newblock ISSN 0003-1224.

\bibitem[Card(1994)]{Card1994}
David Card.
\newblock Earnings, schooling, and ability revisited.
\newblock 1994.

\bibitem[Card(1999)]{Card1999}
David Card.
\newblock \emph{The causal effect of education on earnings}, volume~3, pages
  1801--1863.
\newblock Elsevier, 1999.
\newblock ISBN 1573-4463.

\bibitem[Card(2001)]{Card2001}
David Card.
\newblock Estimating the return to schooling: Progress on some persistent
  econometric problems.
\newblock \emph{Econometrica}, 69:\penalty0 1127--1160, 2001.
\newblock ISSN 0012-9682.

\bibitem[Carneiro et~al.(2011)Carneiro, Heckman, and Vytlacil]{Carneiro2011}
Pedro Carneiro, James~J Heckman, and Edward~J Vytlacil.
\newblock Estimating marginal returns to education.
\newblock \emph{American Economic Review}, 101:\penalty0 2754--2781, 2011.
\newblock ISSN 0002-8282.

\bibitem[Castleman et~al.(2020)Castleman, Deutschlander, and
  Lohner]{Castleman2020}
Benjamin~L Castleman, Denise Deutschlander, and Gabrielle Lohner.
\newblock Pushing college advising forward: Experimental evidence on intensive
  advising and college success.
\newblock \emph{EdWorkingPapers. com}, 2020.

\bibitem[Chernozhukov et~al.(2017)Chernozhukov, Chetverikov, Demirer, Duflo,
  Hansen, and Newey]{Chernozhukov2017}
Victor Chernozhukov, Denis Chetverikov, Mert Demirer, Esther Duflo, Christian
  Hansen, and Whitney Newey.
\newblock Double/debiased/neyman machine learning of treatment effects.
\newblock \emph{American Economic Review}, 107:\penalty0 261--265, 2017.
\newblock ISSN 0002-8282.

\bibitem[Chernozhukov et~al.(2018)Chernozhukov, Chetverikov, Demirer, Duflo,
  Hansen, Newey, and Robins]{Chernozhukov2018}
Victor Chernozhukov, Denis Chetverikov, Mert Demirer, Esther Duflo, Christian
  Hansen, Whitney Newey, and James Robins.
\newblock Double/debiased machine learning for treatment and structural
  parameters, 2018.

\bibitem[Chetty et~al.(2023)Chetty, Deming, and Friedman]{Chetty2023}
Raj Chetty, David~J Deming, and John~N Friedman.
\newblock Diversifying society’s leaders? the causal effects of admission to
  highly selective private colleges, 2023.

\bibitem[Cohodes and Goodman(2014)]{Cohodes2014}
Sarah~R Cohodes and Joshua~S Goodman.
\newblock Merit aid, college quality, and college completion: Massachusetts'
  adams scholarship as an in-kind subsidy.
\newblock \emph{American Economic Journal: Applied Economics}, 6:\penalty0
  251--285, 2014.
\newblock ISSN 1945-7782.

\bibitem[Daniel et~al.(2015)Daniel, Stavola, Cousens, and
  Vansteelandt]{Daniel2015}
Rhian~M Daniel, Bianca L~De Stavola, Simon~N Cousens, and Stijn Vansteelandt.
\newblock Causal mediation analysis with multiple mediators.
\newblock \emph{Biometrics}, 71:\penalty0 1--14, 2015.
\newblock ISSN 0006-341X.

\bibitem[Dynarski et~al.(2021)Dynarski, Libassi, Michelmore, and
  Owen]{Dynarski2021}
Susan Dynarski, C~J Libassi, Katherine Michelmore, and Stephanie Owen.
\newblock Closing the gap: The effect of reducing complexity and uncertainty in
  college pricing on the choices of low-income students.
\newblock \emph{American Economic Review}, 111:\penalty0 1721--1756, 2021.
\newblock ISSN 0002-8282.

\bibitem[Eller(2023)]{CioccaEller2023}
Christina~Ciocca Eller.
\newblock What makes a quality college? re-examining the equalizing potential
  of higher education in the united states.
\newblock 2023.
\newblock ISSN 0002-9602.

\bibitem[Eller and DiPrete(2018)]{eller-diprete}
Christina~Ciocca Eller and Thomas~A DiPrete.
\newblock The paradox of persistence: Explaining the black-white gap in
  bachelor’s degree completion.
\newblock \emph{American Sociological Review}, 83:\penalty0 1171--1214, 2018.
\newblock ISSN 0003-1224.

\bibitem[Farbmacher et~al.(2022)Farbmacher, Huber, Lafférs, Langen, and
  Spindler]{Farbmacher2022}
Helmut Farbmacher, Martin Huber, Lukáš Lafférs, Henrika Langen, and Martin
  Spindler.
\newblock Causal mediation analysis with double machine learning.
\newblock \emph{The Econometrics Journal}, 25:\penalty0 277--300, 2022.
\newblock ISSN 1368-4221.

\bibitem[Goodman et~al.(2017)Goodman, Hurwitz, and Smith]{Goodman2017}
Joshua Goodman, Michael Hurwitz, and Jonathan Smith.
\newblock Access to 4-year public colleges and degree completion.
\newblock \emph{Journal of Labor Economics}, 35:\penalty0 829--867, 2017.
\newblock ISSN 0734-306X.

\bibitem[Heckman et~al.(2018)Heckman, Humphries, and Veramendi]{Heckman2018}
James~J Heckman, John~Eric Humphries, and Gregory Veramendi.
\newblock Returns to education: The causal effects of education on earnings,
  health, and smoking.
\newblock \emph{Journal of Political Economy}, 126:\penalty0 S197--S246, 2018.
\newblock ISSN 0022-3808.

\bibitem[Hout(2012)]{Hout2012}
Michael Hout.
\newblock Social and economic returns to college education in the united
  states.
\newblock \emph{Annual review of sociology}, 38:\penalty0 379--400, 2012.

\bibitem[Hurwitz and Howell(2014)]{Hurwitz2014}
Michael Hurwitz and Jessica Howell.
\newblock Estimating causal impacts of school counselors with regression
  discontinuity designs.
\newblock \emph{Journal of Counseling \& Development}, 92:\penalty0 316--327,
  2014.
\newblock ISSN 0748-9633.

\bibitem[Imai et~al.(2010)Imai, Keele, and Yamamoto]{ImaKeeYam10}
Kosuke Imai, Luke Keele, and Teppei Yamamoto.
\newblock Identification, inference and sensitivity analysis for causal
  mediation effects.
\newblock \emph{Statistical Science}, 25\penalty0 (1):\penalty0 51--71,
  February 2010.
\newblock \doi{10.1214/10-STS321}.

\bibitem[Kane and Rouse(1993)]{Kane1993}
Thomas~J Kane and Cecilia~E Rouse.
\newblock Labor market returns to two-and four-year colleges: is a credit a
  credit and do degrees matter?
\newblock 1993.

\bibitem[Kang and Schafer(2007)]{Kang2007}
Joseph~D.Y. Kang and Joseph~L. Schafer.
\newblock Demystifying double robustness: A comparison of alternative
  strategies for estimating a population mean from incomplete data.
\newblock \emph{Statistical Science}, 22, 2007.
\newblock ISSN 08834237.
\newblock \doi{10.1214/07-STS227}.

\bibitem[Kennedy(2022)]{Kennedy2022}
Edward~H Kennedy.
\newblock Semiparametric doubly robust targeted double machine learning: a
  review.
\newblock \emph{arXiv preprint arXiv:2203.06469}, 2022.

\bibitem[Kenney(2004)]{Kenney2004}
Catherine Kenney.
\newblock Cohabiting couple, filing jointly? resource pooling and us poverty
  policies.
\newblock \emph{Family Relations}, 53:\penalty0 237--247, 2004.
\newblock ISSN 0197-6664.

\bibitem[Kirk and Sampson(2013)]{Kirk2013}
David~S Kirk and Robert~J Sampson.
\newblock Juvenile arrest and collateral educational damage in the transition
  to adulthood.
\newblock \emph{Sociology of education}, 86:\penalty0 36--62, 2013.
\newblock ISSN 0038-0407.

\bibitem[Lewis and Syrgkanis(2020)]{Lewis2020}
Greg Lewis and Vasilis Syrgkanis.
\newblock Double/debiased machine learning for dynamic treatment effects via
  g-estimation.
\newblock \emph{arXiv preprint arXiv:2002.07285}, 2020.

\bibitem[Lin and VanderWeele(2017)]{Lin2017}
Sheng-Hsuan Lin and Tyler VanderWeele.
\newblock Interventional approach for path-specific effects.
\newblock \emph{Journal of Causal Inference}, 5:\penalty0 20150027, 2017.
\newblock ISSN 2193-3685.

\bibitem[Mare(1980)]{Mare1980}
Robert~D Mare.
\newblock Social background and school continuation decisions.
\newblock \emph{Journal of the American Statistical Association}, 75:\penalty0
  295--305, 1980.
\newblock ISSN 0162-1459.

\bibitem[Miles et~al.(2017)Miles, Shpitser, Kanki, Meloni, and
  Tchetgen]{Miles2017}
Caleb~H Miles, Ilya Shpitser, Phyllis Kanki, Seema Meloni, and Eric J~Tchetgen
  Tchetgen.
\newblock Quantifying an adherence path-specific effect of antiretroviral
  therapy in the nigeria pepfar program.
\newblock \emph{Journal of the American Statistical Association}, 112:\penalty0
  1443--1452, 2017.
\newblock ISSN 0162-1459.

\bibitem[Miles et~al.(2020)Miles, Shpitser, Kanki, Meloni, and
  Tchetgen]{Miles2020}
Caleb~H Miles, Ilya Shpitser, Phyllis Kanki, Seema Meloni, and Eric J~Tchetgen
  Tchetgen.
\newblock On semiparametric estimation of a path-specific effect in the
  presence of mediator-outcome confounding.
\newblock \emph{Biometrika}, 107:\penalty0 159--172, 2020.
\newblock ISSN 0006-3444.

\bibitem[Mountjoy(2022)]{Mountjoy2022}
Jack Mountjoy.
\newblock Community colleges and upward mobility.
\newblock \emph{American Economic Review}, 112:\penalty0 2580--2630, 2022.
\newblock ISSN 0002-8282.

\bibitem[Mountjoy and Hickman(2021)]{Mountjoy2021}
Jack Mountjoy and Brent~R Hickman.
\newblock The returns to college (s): Relative value-added and match effects in
  higher education, 2021.

\bibitem[Newey and McFadden(1994)]{Newey1994}
Whitney~K Newey and Daniel McFadden.
\newblock Large sample estimation and hypothesis testing.
\newblock \emph{Handbook of econometrics}, 4:\penalty0 2111--2245, 1994.
\newblock ISSN 1573-4412.

\bibitem[Pearl(2001)]{Pearl01}
Judea Pearl.
\newblock Direct and indirect effects.
\newblock In J.S. Breese and D.~Koller, editors, \emph{Proceedings of the
  seventeenth conference on uncertainty in artificial intelligence}, pages
  411--420, San Francisco, CA, 2001. Morgan Kaufmann Publishers.

\bibitem[Pearl(2009)]{Pearl2009}
Judea Pearl.
\newblock Causal inference in statistics: An overview.
\newblock \emph{Statistics Surveys}, 3, 2009.
\newblock ISSN 19357516.
\newblock \doi{10.1214/09-SS057}.

\bibitem[Robins(1986)]{Robins1986}
James Robins.
\newblock A new approach to causal inference in mortality studies with a
  sustained exposure period-application to control of the healthy worker
  survivor effect.
\newblock \emph{Mathematical modelling}, 7:\penalty0 1393--1512, 1986.
\newblock ISSN 0270-0255.

\bibitem[Rotnitzky et~al.(2017)Rotnitzky, Robins, and Babino]{Rotnitzky2017}
Andrea Rotnitzky, James Robins, and Lucia Babino.
\newblock On the multiply robust estimation of the mean of the g-functional.
\newblock \emph{arXiv preprint arXiv:1705.08582}, 2017.

\bibitem[Scott-Clayton and Wen(2019)]{Scott-Clayton}
Judith Scott-Clayton and Qiao Wen.
\newblock Estimating returns to college attainment: Comparing survey and state
  administrative data–based estimates.
\newblock \emph{Evaluation Review}, 43:\penalty0 266--306, 2019.
\newblock ISSN 0193-841X.

\bibitem[Smith et~al.(2020)Smith, Goodman, and Hurwitz]{Smith2020}
Jonathan Smith, Joshua Goodman, and Michael Hurwitz.
\newblock The economic impact of access to public four-year colleges, 2020.

\bibitem[Snyder et~al.(2016)Snyder, de~Brey, and Dillow]{Snyder2016}
Thomas~D Snyder, Cristobal de~Brey, and Sally~A Dillow.
\newblock Digest of education statistics 2014, 50th edition. nces 2016-006,
  2016.

\bibitem[Steen et~al.(2017)Steen, Loeys, Moerkerke, and
  Vansteelandt]{Steen2017}
Johan Steen, Tom Loeys, Beatrijs Moerkerke, and Stijn Vansteelandt.
\newblock Flexible mediation analysis with multiple mediators.
\newblock \emph{American journal of epidemiology}, 186:\penalty0 184--193,
  2017.
\newblock ISSN 0002-9262.

\bibitem[Sugie and Turney(2017)]{Sugie2017}
Naomi~F Sugie and Kristin Turney.
\newblock Beyond incarceration: Criminal justice contact and mental health.
\newblock \emph{American Sociological Review}, 82:\penalty0 719--743, 2017.
\newblock ISSN 0003-1224.

\bibitem[Sullivan et~al.(2019)Sullivan, Castleman, and Bettinger]{Sullivan2019}
Zachary Sullivan, Benjamin~L Castleman, and Eric Bettinger.
\newblock College advising at a national scale: Experimental evidence from the
  collegepoint initiative.
\newblock 2019.

\bibitem[Sweeney and Phillips(2004)]{Sweeney2004}
Megan~M Sweeney and Julie~A Phillips.
\newblock Understanding racial differences in marital disruption: Recent trends
  and explanations.
\newblock \emph{Journal of Marriage and Family}, 66:\penalty0 639--650, 2004.
\newblock ISSN 0022-2445.

\bibitem[Turner and Gurantz(2024)]{Turner2024}
Lesley~J Turner and Oded Gurantz.
\newblock Experimental estimates of college coaching on postsecondary
  re-enrollment, 2024.

\bibitem[VanderWeele(2014)]{vanderweele2014unification}
Tyler~J VanderWeele.
\newblock A unification of mediation and interaction: a 4-way decomposition.
\newblock \emph{Epidemiology}, 25\penalty0 (5):\penalty0 749--761, 2014.

\bibitem[VanderWeele and Arah(2011)]{VanderWeele2011}
Tyler~J VanderWeele and Onyebuchi~A Arah.
\newblock Bias formulas for sensitivity analysis of unmeasured confounding for
  general outcomes, treatments, and confounders.
\newblock \emph{Epidemiology}, pages 42--52, 2011.
\newblock ISSN 1044-3983.

\bibitem[VanderWeele and Vansteelandt(2009)]{VanderWeele2009a}
Tyler~J VanderWeele and Stijn Vansteelandt.
\newblock Conceptual issues concerning mediation, interventions and
  composition.
\newblock \emph{Statistics and its Interface}, 2:\penalty0 457--468, 2009.
\newblock ISSN 1938-7997.

\bibitem[VanderWeele et~al.(2014)VanderWeele, Vansteelandt, and
  Robins]{VanderWeele2014}
Tyler~J VanderWeele, Stijn Vansteelandt, and James~M Robins.
\newblock Effect decomposition in the presence of an exposure-induced
  mediator-outcome confounder.
\newblock \emph{Epidemiology (Cambridge, Mass.)}, 25:\penalty0 300, 2014.

\bibitem[Vansteelandt and Daniel(2017)]{Vansteelandt2017}
Stijn Vansteelandt and Rhian~M Daniel.
\newblock Interventional effects for mediation analysis with multiple
  mediators.
\newblock \emph{Epidemiology}, 28:\penalty0 258--265, 2017.
\newblock ISSN 1044-3983.

\bibitem[Viviano and Bradic(2021)]{Viviano2021}
Davide Viviano and Jelena Bradic.
\newblock Dynamic covariate balancing: estimating treatment effects over time.
\newblock \emph{arXiv preprint arXiv:2103.01280}, 2021.

\bibitem[Weaver and Lerman(2010)]{Weaver2010}
Vesla~M Weaver and Amy~E Lerman.
\newblock Political consequences of the carceral state.
\newblock \emph{American Political Science Review}, 104:\penalty0 817--833,
  2010.
\newblock ISSN 1537-5943.

\bibitem[Yu et~al.(2024)Yu, Ge, and Elwert]{yu2024natural}
Ang Yu, Li~Ge, and Felix Elwert.
\newblock When do natural mediation effects differ from their randomized
  interventional analogues: Test and theory.
\newblock \emph{arXiv preprint arXiv:2407.02671}, 2024.

\bibitem[Zhou(2022{\natexlab{a}})]{Zhou2022a}
Xiang Zhou.
\newblock Attendance, completion, and heterogeneous returns to college: A
  causal mediation approach.
\newblock \emph{Sociological Methods \& Research}, page 00491241221113876,
  2022{\natexlab{a}}.
\newblock ISSN 0049-1241.

\bibitem[Zhou(2022{\natexlab{b}})]{Zhou2022b}
Xiang Zhou.
\newblock Semiparametric estimation for causal mediation analysis with multiple
  causally ordered mediators.
\newblock \emph{Journal of the Royal Statistical Society Series B: Statistical
  Methodology}, 84:\penalty0 794--821, 2022{\natexlab{b}}.
\newblock ISSN 1369-7412.

\bibitem[Zhou and Pan(2023)]{Zhou2023}
Xiang Zhou and Guanghui Pan.
\newblock Higher education and the black-white earnings gap.
\newblock \emph{American Sociological Review}, 88:\penalty0 154--188, 2023.
\newblock ISSN 0003-1224.

\bibitem[Zimmerman(2014)]{Zimmerman2014}
Seth~D Zimmerman.
\newblock The returns to college admission for academically marginal students.
\newblock \emph{Journal of Labor Economics}, 32:\penalty0 711--754, 2014.
\newblock ISSN 0734-306X.

\end{thebibliography}

\newpage
\setcounter{section}{0}
\renewcommand{\thesection}{\Alph{section}}
\setcounter{table}{0}
\setcounter{figure}{0}
\renewcommand{\thetable}{S\arabic{table}}
\renewcommand{\thefigure}{S\arabic{figure}}
\section*{Supplemental Materials (to appear online)}

\section{Parametric, regression-with-residuals (RWR) estimation\protect\label{App:rwr}}

In this section, I propose a linear regression-with-residuals (RWR)
approach for the MPSE decomposition. The approach relies on two steps.
The first involves residualizing pre-treatment confounders with respect
to their marginal means, and intermediate confounders on all causally
prior confounders, i.e., $X^{\perp}\triangleq X-\mathbb{E}[X]$, and
$Z_{k}^{\perp}\triangleq M_{k-1}\big[Z_{k}-\mathbb{E}[Z_{k}\mid X,\overline{Z}_{k-1},M_{k-1}=1]\big]$
for all $k\in[K]$, $M_{0}\triangleq A$. For now, we are agnostic
about the functional form used for $\mathbb{E}[Z_{k}\mid X,\overline{Z}_{k-1},M_{k-1}=1]$.
The second step involves fitting three sets of models. The first is
simply a model for the outcome given pre-treatment covariates and
the treatment, namely,

\begin{equation}
\begin{aligned}\mathbb{E}[Y\mid X,A] & =\lambda_{0}+\lambda_{1}A+\alpha_{1}^{T}X^{\perp}+\alpha_{2}^{T}AX^{\perp};\end{aligned}
\label{eq:rwr-ate}
\end{equation}

The second is a set of models for the outcome given covariates, the
treatment and $M_{k}$ for all $k\in[K]$, i.e.,

\begin{align}
\mathbb{E}[Y|X,\overline{Z}_{k},A,\overline{M}_{k}]= & \beta_{k,0}+c_{k,0}A+\sum_{j=1}^{k}\beta_{k,j}M_{j}+\eta_{k,1}^{\top}X^{\perp}+c_{k,1}AX^{\perp}+\sum_{j=1}^{k-1}\eta_{k,j}^{T}M_{j}X^{\perp}\label{eq:rwr-outcome-1}\\
 & +\sum_{j=1}^{k}\gamma_{k,j}^{T}Z_{j}^{\perp}+\sum_{j=1}^{k-1}M_{j}\sum_{l=1}^{j}\xi_{k,k,l}^{\top}Z_{l}^{\perp},\nonumber 
\end{align}

 while the third is a set of models for each mediator given covariates,
the treatment, conditional on the treatment and all prior mediators,
i.e., for all $k\in[K-1]$,

\begin{align}
\mathbb{E}[M_{k+1}\mid X,\overline{Z}_{k},\overline{1}_{k+1}]= & \theta_{k,0}+\delta_{k,1}^{T}X^{\perp}+\sum_{j=1}^{k}\delta_{k,j+1}^{T}Z_{j}^{\perp}.\label{eq:rwr-mediator}
\end{align}

These models differ from conventional linear regression in that (i)
pre-treatment variables are centered around their marginal means,
and (ii) post-treatment confounders $Z_{k}\forall k\in\{1,\dots K\}$
are centered around their conditional means given all antecedent variables.
Under Assumptions \ref{assu:Consistency}-\ref{assu:Positivity} in
the main text, and assuming that the outcome and mediators are linear
in pre- and post-treatment confounders, the treatment, and prior mediators,
and that all necessary interaction terms have been accounted for,
then the ATE $\tau_{0}$ can be obtained from the linear model $\mathbb{E}[Y\mid X,A]$
as $\lambda_{1}$, and coefficients from the models $\mathbb{E}[Y\mid X,A,\overline{Z}_{k},\overline{M}_{k}]$
and $\mathbb{E}[M_{k+1}\mid X,A,\overline{Z}_{k},\overline{M}_{k}]$
yield estimates of the components of the decomposition as follows:

\begin{align*}
\tau_{k} & =\mathbb{E}[Y(\overline{1}_{k+1})-Y(\overline{1}_{k},0)]=\beta_{k,k},\forall k\in\{1,\dots,K\},\\
\Delta_{k} & =\mathbb{E}[Y(\overline{1}_{k+1},\underline{0}_{k+2})-Y(\overline{1}_{k},\underline{0}_{k+1})]=\beta_{k+1,k-1},\forall k\in\{0,\dots,K-1\},\\
\pi_{k+1} & =\mathbb{E}[M_{k+1}(\overline{1}_{k+1})]=\theta_{k,0},\forall k\in\{0,\dots,K-1\}.
\end{align*}

I state the RWR estimation procedure formally in the following algorithm:
\begin{algorithm}[H]
\caption{RWR}
\end{algorithm}
\begin{enumerate}
\item For each of the baseline confounders, compute $\hat{X}^{\perp}=X-\mathbb{P}_{n}[X]$,
where $\mathbb{P}_{n}[\cdot]$ denotes empirical average.
\item Fit $\mathbb{\hat{E}}[Y\mid X,A]$ using the linear specification
shown above; an estimate of $\tau_{0}$ is given by $\hat{\lambda_{1}}$.
\item For each set of post-treatment confounders $Z_{k}$, $k\in\{1,\dots,K\}$,
compute $Z_{k}^{\perp}=M_{k-1}\big[Z_{k}-\mathbb{E}[Z_{k}\mid X,\overline{Z}_{k-1},M_{k-1}=1]\big]$
where an overbar denotes a vector of variables such that $\overline{Z}_{k}=(Z_{1},\dots,Z_{k})$,
by fitting a regression of $Z_{k}$ on $X$ and $\overline{Z}_{k-1}$
among units with $M_{k-1}=1$ and then calculating the residuals.
\item For each $k\in\{1,\dots K\}$:
\begin{enumerate}
\item compute least squares estimates of equations \ref{eq:rwr-outcome-1}
and \ref{eq:rwr-mediator}, using estimates of $X^{\perp}$ and $Z_{k}^{\perp}$.
\item compute $\text{\ensuremath{\hat{\tau}}}_{k}^{\text{RWR}}=\hat{\beta}_{k,k}$,
$\hat{\Delta}_{k-1}^{\text{RWR}}=\hat{\beta}_{k,k-1}$, and $\hat{\pi}_{k}^{\text{RWR}}=\hat{\theta}_{k-1,0}$.
\end{enumerate}
\item Compute the decomposition using $\hat{\tau}_{k}$, $\hat{\Delta}_{k}$
and $\hat{\pi}_{k+1}$, and estimating the covariance terms as $\hat{\eta}_{k}^{\text{RWR}}=\hat{\beta}_{k-1,k}-\hat{\beta}_{k,k-1}-\hat{\beta}_{k,1}\hat{\theta}_{k-1,k}$,
and the continuation effects as $\hat{\theta}_{k}^{\text{RWR}}=(\Pi_{j=1}^{k}\hat{\pi}_{j}^{\text{RWR}})\hat{\Delta}_{k}^{\text{RWR}}+(\Pi_{j=1}^{k-1}\hat{\pi}_{j}^{\text{RWR}})\hat{\eta}_{k}^{\text{RWR}}$.
\end{enumerate}
Standard errors and confidence intervals can then be obtained via
the non-parametric bootstrap, or by using their asymptotic analytic
variance. Specifically, let $\hat{\theta_{k}^{*}}\triangleq(\hat{\beta}_{k,0},\hat{\beta}_{k,1},\theta_{k,0})$
denote a set of parameters. Under the above models, we have that $\hat{\theta}^{*}=\{\hat{\lambda}_{1},\theta_{1}^{*},\dots\theta_{K}^{*}\}$
solves $\mathbb{P}_{n}[g(O;\hat{\theta}^{*})]=0$, where $g(O;\theta^{*})$
is the set of stacked moment conditions with solution $\hat{\theta}^{*}$.
Under standard regularity conditions \citep{Newey1994}, under correct
specification of the models wherein all residualized quantities are
estimated via linear models, the set $\hat{\theta}^{*}$ is consistent
and asymptotically normal, such that $\sqrt{n}\big(\hat{\theta}^{*}-\theta^{*})$
converges to a mean-zero normal distribution with finite variance
$V=G^{-1}\Omega(G^{-1})^{\top}$, where $\Omega=\mathbb{E}[g(O;\theta^{*})g(O;\theta^{*})^{\top}]$,
and where $G=\mathbb{E}[\frac{\partial g(O;\theta^{*})}{\partial\theta^{\top}}]$.
It follows by a simple application of the Delta Method that the set
$\hat{\gamma_{k}^{*}}\triangleq(\hat{\tau}_{k}^{\text{RWR}},\hat{\Delta}_{k-1}^{\text{RWR}},\hat{\pi}_{k}^{\text{RWR}},\hat{\eta}_{k}^{\text{RWR}},\hat{\theta}_{k}^{\text{RWR}})\forall k\in[K]$
is also consistent and asymptotically normal.

\clearpage

\section{Further details on simulation study\protect\label{app:Simulation}}

In this section, I provide further details on the simulation exercise described in the main text. To recall, the assumed data-generating process is

\[
\begin{aligned}U_{1},U_{2},U_{3},U_{4} & \sim\text{MVN}\big(0_{4},I_{4}\big)\\
X_{1} & \sim\text{N}\big((U_{1},U_{2},U_{3},U_{4})\beta_{X_{1}},1\big)\\
X_{2} & \sim\text{N}\big((U_{1},U_{2},U_{3},U_{4})\beta_{X_{2}},1\big)\\
A & \sim\text{Bern}\big(g^{-1}\big[(1,X)\beta_{A}\big]\big)\\
Z|A=1 & \sim\text{N}\big((1,X)\beta_{Z},1\big)\\
M_1|A=1\sim\text{Bern}\big(g^{-1}\big[(1,X,Z)\beta_{M_1}\big]\big) & \quad M|A=0=0\\
M_2|M_1=1\sim\text{Bern}\big(g^{-1}\big[(1,X,Z)\beta_{M_2}\big]\big) & \quad M_2|M_1=0=0\\
Y & \sim\text{N}\big((1,A,X,AZ,AM_1, A M_1 M_2)\beta_{Y},1\big).
\end{aligned}
\]

The coefficients $(\beta_{X_{1}},\beta_{X_{2}},\beta_{Y})$ are drawn
from a $\text{Unif}\big(-1,1\big)$ distribution, while the coefficient $\beta_{A}$
is drawn from a $\text{Unif}\big(-0.5,0.5\big)$ distribution. When evaluating the DML estimation procedure, I set $g$ to be the logistic link: $g^{-1}(x) \;=\; \frac{\exp(x)}{1+\exp(x)}$. While the DML estimator is agnostic to the functional form of $g$---all nuisance functions, including propensity scores and mediator models, are
estimated using flexible machine-learning methods---the RWR estimator relies on parametric linear regressions for
both mediator and outcome models. When evaluating the RWR estimator, I therefore set $g$ to be the identity link.\footnote{
More subtly, even if the outcome model is linear in
$(M_1,M_2)$ by construction, the regression of $Y$ on
$(X,A,M_1,Z)$ used by RWR is generally misspecified when mediator models
are nonlinear. To see this, note that
\[
\mathbb{E}[Y\mid X,A=1,M_1=m,Z]
=
\mathbb{E}\!\left[
\mathbb{E}[Y\mid X,A=1,M_1=m,M_2,Z]
\;\middle|\;
X,A=1,M_1=m,Z
\right],
\]
which involves integrating a linear function of $M_2$ with respect to a
nonlinear conditional distribution of
$M_2\mid X,Z,A=1,M_1=1$.
Unless $\mathbb{E}[M_2\mid X,Z,A=1,M_1=1]$ is linear in $X$,
this marginal conditional expectation is itself nonlinear in $X$.
Thus, even when the structural outcome model is linear, the reduced-form
outcome regression used by RWR will be misspecified unless $g$ is the identity link.}

For the DML
estimator, for each component involved in the decomposition, I construct
a Neyman-orthogonal ``signal'' using its EIF. The recentered EIFs
for each component are as follows:

\begin{align*}
M_1^{*}(1)
&=
\gamma_{1}(X)
+\frac{\mathbb{I}(A=1)}{\pi_{0}(X,1)}\bigl(M_1-\gamma_{1}(X)\bigr),\\[6pt]
M_2^{*}(1,1)
&=
\gamma_{2}(X)
+\frac{\mathbb{I}(A=1)\mathbb{I}(M_1=1)}{\pi_{0}(X,1)\,\pi_{1}(X,Z,1)}
\bigl(M_2-\gamma_{2}(X)\bigr),\\[10pt]
Y^{*}(a)
&=
\mu_{0}(X,a)
+\frac{\mathbb{I}(A=a)}{\pi_{0}(X,a)}\bigl(Y-\mu_{0}(X,a)\bigr),
\qquad a\in\{0,1\},\\[10pt]
Y^{*}(1,m_1)
&=
\nu_{1}(X,m_1)
+\frac{\mathbb{I}(A=1)\mathbb{I}(M_1=m_1)}{\pi_{0}(X,1)\,\pi_{1}(X,Z,m_1)}
\bigl(Y-\mu_{1}(X,Z,m_1)\bigr)\\
&\qquad
+\frac{\mathbb{I}(A=1)}{\pi_{0}(X,1)}\bigl(\mu_{1}(X,Z,m_1)-\nu_{1}(X,m_1)\bigr),
\qquad m_1\in\{0,1\},\\[10pt]
Y^{*}(1,1,m_2)
&=
\nu_{2}(X,m_2)
+\frac{\mathbb{I}(A=1)\mathbb{I}(M_1=1)\mathbb{I}(M_2=m_2)}
{\pi_{0}(X,1)\,\pi_{1}(X,Z,1)\,\pi_{2}(X,Z,m_2)}
\bigl(Y-\mu_{2}(X,Z,m_2)\bigr)\\
&\qquad
+\frac{\mathbb{I}(A=1)\mathbb{I}(M_1=1)}{\pi_{0}(X,1)\,\pi_{1}(X,Z,1)}
\bigl(\mu_{2}(X,Z,m_2)-\nu_{2}(X,m_2)\bigr),
\qquad m_2\in\{0,1\}.
\end{align*}

where
\[
\begin{aligned}
\pi_{0}(X,a) &\triangleq \Pr(A=a\mid X),\\
\pi_{1}(X,Z,m_1) &\triangleq \Pr(M_1=m_1\mid X,Z,A=1),\\
\pi_{2}(X,Z,m_2) &\triangleq \Pr(M_2=m_2\mid X,Z,A=1,M_1=1),\\
\gamma_{1}(X) &\triangleq \mathbb{E}[M_1\mid X,A=1],\\
\gamma_{2}(X) &\triangleq \mathbb{E}[M_2\mid X,A=1,M_1=1],\\
\mu_{0}(X,a) &\triangleq \mathbb{E}[Y\mid X,A=a],\\
\mu_{1}(X,Z,m_1) &\triangleq \mathbb{E}[Y\mid X,A=1,Z,M_1=m_1],\\
\nu_{1}(X,m_1) &\triangleq \mathbb{E}\!\left[\mu_{1}(X,Z,m_1)\mid X,A=1\right],\\
\mu_{2}(X,Z,m_2) &\triangleq \mathbb{E}[Y\mid X,A=1,Z,M_1=1,M_2=m_2],\\
\nu_{2}(X,m_2) &\triangleq \mathbb{E}\!\left[\mu_{2}(X,Z,m_2)\mid X,A=1,M_1=1\right].
\end{aligned}
\]

When both the outcome model and the mediator models are linear probability
models, these quantities can be derived analytically.
For example, under linearity, the continuation effect associated with the
second mediator coincides with the coefficient on $M_2$ in the outcome
model, and is used to calculate the MPSE via $M_2$.  However, under a nonlinear DGP, such closed-form
expressions for the $\theta_k$ terms no longer exist because the mediator counterfactuals are
nonlinear functions of $(X,Z)$. In this case, I recover the true values of
$\left(\theta_0,\theta_1,\theta_2\right)$ by Monte Carlo integration under
the known data-generating process.
Specifically, for each replication, I simulate a large population from
the structural equations under the relevant interventions in order to
recover the counterfactual mediator probabilities.
I then plug these counterfactual probabilities into the corresponding
linear expressions implied by the outcome model to obtain the true values
of the path-specific effects.
These Monte Carlo quantities are treated as the ground truth against
which finite-sample bias and coverage are evaluated.


It is worthwhile briefly clarifying how the proposed monotonic path-specific
effect (MPSE) decomposition relates to—and departs from—conventional mediation
analyses of the ATE with multiple causally ordered
mediators. Without mediator monotonicity, the ATE admits an
algebraic decomposition into a collection of path-specific effects (PSEs)
corresponding to the causal paths $A \to Y$ and $A \to M_k \rightsquigarrow Y$
for $k = 1, \dots, K$ \citep{Avin2005,Daniel2015,Zhou2022b}. However, when
mediators are causally ordered, only composite effects of the form
$A \to M_k \rightsquigarrow Y$—which aggregate all downstream pathways from $M_k$
to $Y$—are generally identifiable.

Under mediator monotonicity, this structure collapses in an important way.
Because later mediators are deterministically zero whenever earlier mediators are
not realized, path-specific effects of the form $A \to M_k \to Y$ for $k \ge 2$
are identically zero (see Section 2.3 in the main text). Consequently, the general PSE decomposition simplifies to a
single-mediator natural effect decomposition in which the total indirect effect
operates exclusively through the full causal chain
$A \to M_1 \to \cdots \to M_k \to Y$. This collapse has two implications for the simulation results. First, it explains
why conventional multi-mediator decompositions do not provide a meaningful
benchmark in this setting: once monotonicity is imposed, the distinction between
multiple indirect paths disappears, and the estimand reduces to a single composite
mediation effect. Second, and more importantly, even this reduced decomposition is
not identifiable under standard mediation assumptions when intermediate
confounders are present. In the
simulation design considered here, these assumptions are deliberately violated by
allowing for observed post-treatment confounders $Z$ that affect both later
mediators and the outcome.

\newpage{}
\begin{figure}[h]
\begin{centering}
\includegraphics[width=\linewidth,keepaspectratio]{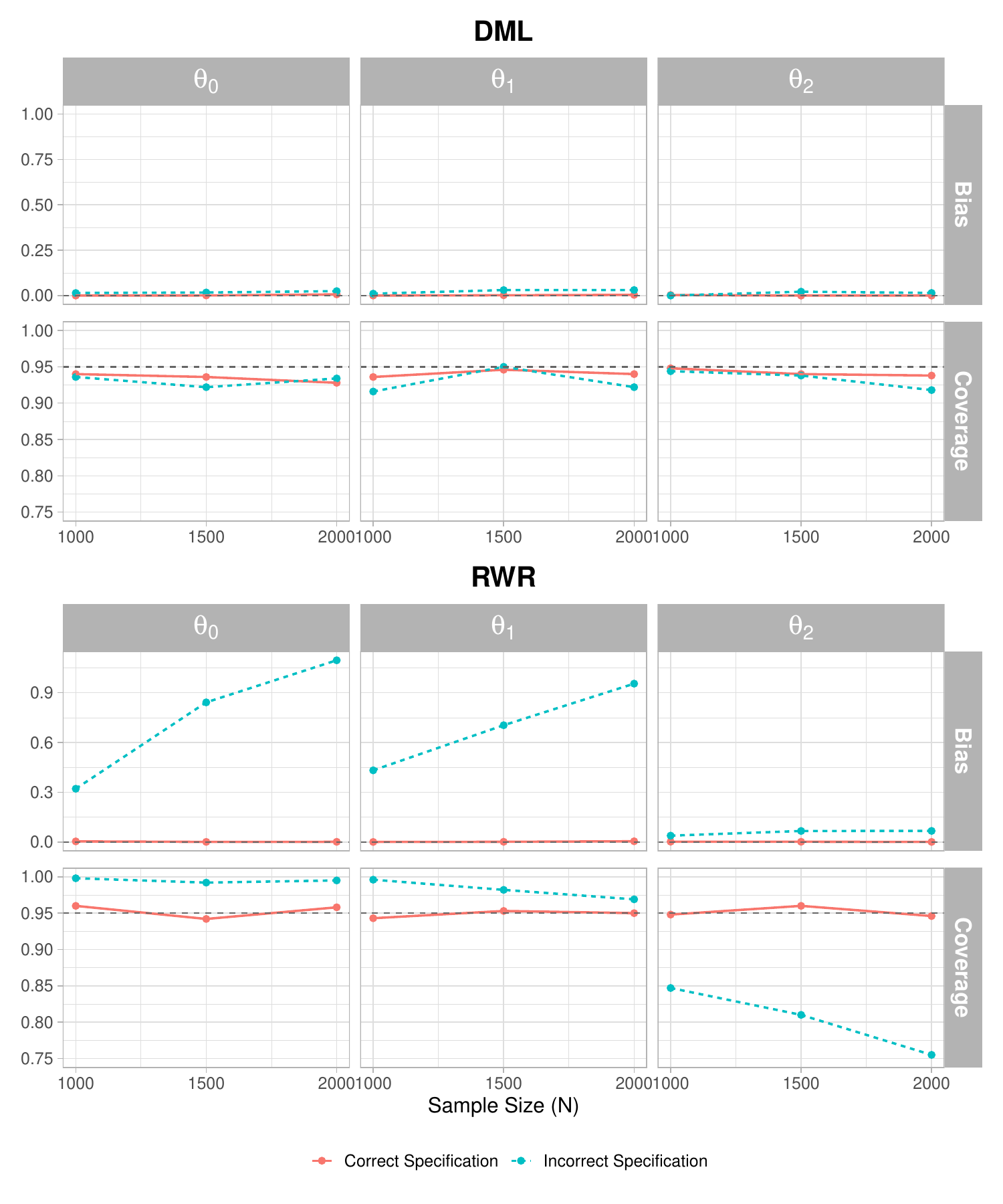}
\par\end{centering}
\caption{Bias, RMSE, and coverage of DML and RWR estimators for $n=1000,1500,2000$.
The red dots show the performance of the DML and RWR estimators
when the correct feature matrix is supplied to the estimators; the
blue dots show the performance of the two estimators when an incorrect
feature matrix is supplied to the estimators. \protect\label{fig:Simulation}}
\end{figure}

\FloatBarrier
\newpage{}

\section{Sensitivity analysis\protect\label{App:sens}}

How do my estimates of returns to different educational stages, as
presented in the main text, tally with previous findings on the labor
market returns to education? While previous work does not estimate
quantities analogous to the direct and indirect effects of interest
(i.e., the $\theta_{k}$ terms), some prior educational returns estimates
are closely related to the net effect ($\tau_{k}$) terms that
inform the total, direct and indirect components of the decomposition.
My estimate of the net effect of 4-year college enrollment ($\tau_{1}$)
is large, at $54\%$, but not implausibly so. While \citet{Zimmerman2014}
and \citet{Smith2020} recover college earnings returns at around
$20\%$ by age 30 (exploiting admissions discontinuities in the Florida and
Georgia state university systems), both of these studies estimate
the earnings premium from attending a less selective 4-year college
rather than a community college, for the marginally qualified university
attendee. By contrast, $\tau_{1}$ captures the effect of 4-year college
enrollment compared with community college \textit{and} no college
enrollment, pooling across the less selective colleges examined in
\citet{Zimmerman2014} and \citet{Smith2020}, as well as over more
selective colleges which could have greater earnings effects. Moreover,
since $\tau_{1}$ represents an effect averaged over all individuals,
it reflects a return among a broader population than the marginal
college-goers examined in previous studies.\footnote{My estimate of $\Delta_{1}$ (the direct effect of 4-year college
attendance on earnings) further tallies with a similar quantity estimated
by \citet{Scott-Clayton}. On the intensive margin of employment (i.e.
dropping respondents with zero observed earnings), the authors estimate
a return to college attendance without degree completion of 0.21.
While, theoretically, one might expect my estimate---which corresponds
to the extensive employment margin (including respondents with zero
observed earnings)---to be larger, the fact that it is slightly smaller
could reflect several factors, including the richer array of pre-college
controls I use in my models, model mis-specification resulting from
linearities imposed in prior work and, perhaps most importantly, collider-stratification
biases induced by conditioning on BA completion in \citet{Scott-Clayton}'s
models (biases that are likely reduced by the inclusion of time-varying
controls).} As I discuss in the main text, an additional, especially instructive point of
comparison are instrumental variable (IV) estimates of returns to
years of schooling. These results are in fact quite consistent with
those I report in the main text. 

Of course, an alternative explanation is that my estimates are upwardly
biased by a large degree of unobserved confounding. While the sequential
ignorability assumption facilitates identification of educational
effect pathways under a weaker set of conditions than might be typically
invoked in mediation settings, it is still strong and fundamentally
unverifiable. To assess potential bias of the estimated MPSEs due
to unobserved confounders not picked up in my covariate set $(X,\overline{Z}_{K})$,
I propose a sensitivity analysis for each of the MPSEs.

Assume first that we have a binary unobserved confounder, $U$, for
the treatment-outcome relationship. Assuming that $\alpha_{0}=\mathbb{E}[Y|x,a,U=1]-\mathbb{E}[Y|x,a,U=0]$
does not depend on $x$ or $a$, and further that $\beta_{0}=\text{Pr}[U=1|x,A=1]-\text{Pr}[U=1|x,A=0]$
does not depend on $x$, for $\tau_{0}=\mathbb{E}[Y(1)-Y(0)]\triangleq\text{ATE}$,
then $\text{bias}(\tau_{0})=\alpha\beta$ \citep{VanderWeele2011}.

Next, consider an unobserved binary confounder, $U_{k}$ that affects
both $M_{k}$ and $Y$ for any $k\in\{1,\dots,K\}$. Then, under a
weaker instantiation of Assumption \ref{assu:SI} (Sequential Ignorability),
i.e.,

\begin{equation}
Y(\overline{1}_{k},m_{k})\perp\!\!\!\perp(A,\overline{M}_{k})|X,A,U_{k},\overline{Z}_{k},\overline{M}_{k-1}\forall k\in[K],
\end{equation}

which states that potential outcomes under an arbitrary transition
sequence are independent of observed treatment and mediator values
conditional on observed confounders $(X,\overline{Z}_{k})$ \textit{and}
unobserved confounders $U_{k}$. Under the following set of assumptions:
(Assumption $A_{k}$) $\alpha_{k}=\mathbb{E}[Y|x,\overline{z}_{k},\overline{1}_{k},m_{k},U_{k}=1]-\mathbb{E}[Y|x,\overline{z}_{k},\overline{1}_{k},m_{k},U_{k}=0]$
does not depend on ($x,\overline{z}_{k},\overline{1}_{k},m_{k})$,
and (Assumption $B_{k}$), $\beta_{k}=\text{Pr}[U_{k}=1|x,\overline{z}_{k},\overline{1}_{k},m_{k}]-\text{Pr}[U_{k}=1|x,\overline{z}_{k},\overline{1}_{k}]$
does not depend on ($x,\overline{z}_{k}$), we can show that, for
any $k\in\{1,\dots K\}$,

\begin{align*}
\text{bias}(\tau_{k}) & =\alpha_{k}\beta_{k},
\end{align*}

and, further, that
\[
\text{bias}(\Delta_{k-1})=-\alpha_{k}\beta_{k}\pi_{k},
\]

where $\pi_{k}=\int_{x}\int_{\overline{z}_{k}}\text{Pr}[M_{k}=1|x,\overline{z}_{k},\overline{1}_{k}]\prod_{j=1}^{k}dP(z_{j}|x,\overline{z}_{j-1},\overline{1}_{j-1})dP(x)$,
and is estimable from observed data using the estimation strategies described
previously. A contour plot showing bias-adjusted estimates of $\Delta_{k^{*}}$
and $\tau_{k}$ then enables assessment of how strong the unobserved
confounder would need to be to reduce estimates of the direct and
gross effects to zero. I illustrate these techniques in my empirical
illustration below.

In order to assess the robustness of my empirical findings in the
main text to potential violations of Assumption \ref{assu:SI} (Sequential
Ignorability), I implement this sensitivity analysis discussed above.

Figure \ref{fig:sens} below displays a set of contour plots, which
capture the bias-corrected estimates of the $(\Delta_{k},\tau_{k})$
terms under varying degrees of confounding (that is, under different
values of $\alpha_{k}$ and $\beta_{k}$). For example, the level
set marked ``0'' corresponds to values of $(\alpha_{k},\beta_{k})$
required in order for the unobserved confounder to fully ``explain
away'' estimates $(\Delta_{k},\tau_{k})$ (i.e., to reduce their
true values to zero). Importantly, each row corresponds to a different
set of $(\alpha_{k},\beta_{k})$ terms for a given $U$, such that
the top row corresponds to $(\alpha_{0},\beta_{0})$, while the second
row corresponds to $(\alpha_{1},\beta_{1})$, and so on.

For simplicity, I consider $U$ to be an unmeasured binary confounder
that is (marginally) positively associated with each transition $A,M_{1},\dots M_{3}$
as well as with adult earnings $Y$. To benchmark the hypothetical
behavior of $U$, for each plot, I also display the values of $(\alpha_{k},\beta_{k})$
that would correspond to a $U$ that behaved similarly to a given
confounder that I do observe in the data: an indicator for whether
an individual's test score on the ASVAB is above the median. In each
plot, I mark this point and label it ``Ability''. For each plot,
I also mark the point on the zero contour that corresponds to $\alpha_{k}=\beta_{k}$
(i.e, the point at which the unobserved confounder's associations
with the treatment and with the outcome are equal, and reduce the
true value of the parameter to zero). 

I focus on estimates of $\tau_{0}$ and $\Delta_{0}$ to assess how
robust my primary conclusion---that the ATE of high school completion
is overwhelmingly mediated by high school's direct effect on earnings---is to unobserved confounding. Bias-adjusted estimates of the ATE
$\tau_{0}$ are presented in the top row of Figure \ref{fig:sens}.
Since $U$ is assumed to be positively associated with both $A$ and
$Y$, $\tau_{0}$ is overestimated and suffers from a bias of $\alpha_{0}\beta_{0}$.
My estimate of $\tau_{0}$ at $0.67$ is nevertheless quite robust:
if $U$ had similar effects to ability, the effect would be reduced
by $0.06$ log points, to $0.61$, still implying a high earnings
premium to high school completion overall in excess of $84\%$.

How do my estimates of the direct effect of high school completion
$\Delta_{0}$ (and, in particular, about the proportion of the total
effect that is direct) fare under unobserved confounding? The second
row of Figure \ref{fig:sens} considers bias-adjusted estimates of
$\Delta_{0}=\theta_{0}$ under different values of $(\alpha_{1},\beta_{1})$,
which correspond to the effects of an unobserved confounder $U$ (marginally)
positively associated with both $M_{1}$ and with $Y$. As described
above, in this scenario, $\Delta_{0}$ is affected by a bias of $-\alpha_{1}\beta_{1}\pi_{1}$.
Importantly, even if the unobserved confounder $U$ is \textit{marginally}
positively associated with high school graduation ($A$), the conditional
association between $U$ and $A$ may be zero or even negative since
$M_{1}$ is a collider of $A$ and $U$. In the case that the conditional
association between $U$ and $A$ is negative, $-\alpha_{1}\beta_{1}\pi_{1}$
would be positive, implying an overestimation of the direct effect
$\theta_{0}$. On the plot, I show estimates of $U$ if it behaved
similarly to the ability variable. Indeed, despite the fact that ability
is marginally positively associated with high school completion (top
row of Figure \ref{fig:sens}), its conditional association---conditional
on college attendance---is depressed to zero. Thus, it would take
an extreme form of confounding for $\theta_{0}$ to be largely different
from its estimated value of $.46$. In this way, my primary finding
that the ATE of high school graduation is overwhelmingly mediated
via its direct effect remains highly robust to patterns of unobserved
confounding, under my set of simplifying assumptions.

\newpage{}

\begin{figure}[h]
\begin{centering}
\includegraphics[scale=0.58]{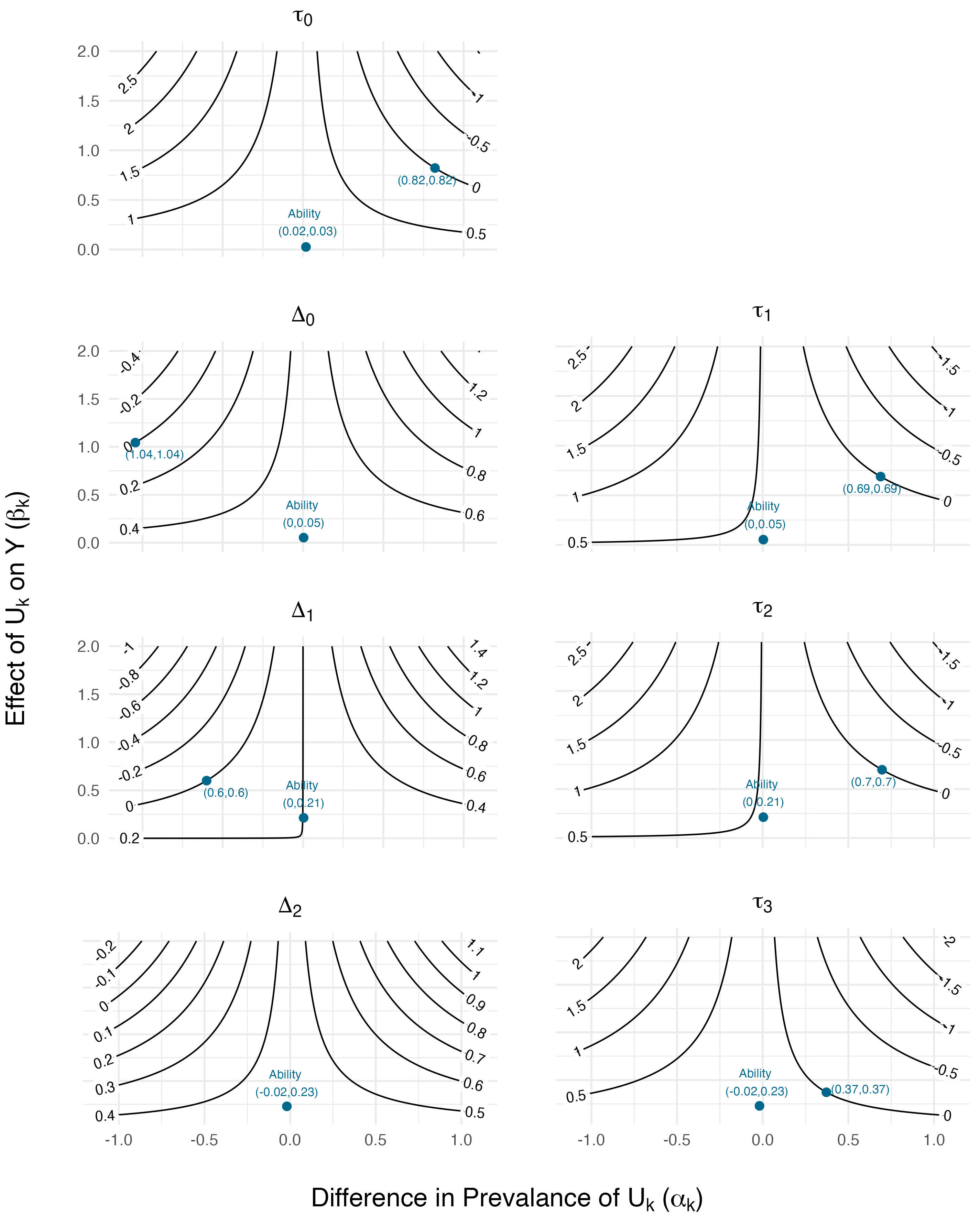}
\par\end{centering}
\caption{Sensitivity Analysis for the ``gross effect'' $(\tau_{k})$ and
``direct effect'' $(\Delta_{k})$ terms in decomposition. Each row
corresponds to a different set of ($\alpha_{k},\beta_{k})$ terms,
where $\alpha_{k}=\mathbb{E}[Y|x,\overline{z}_{k},\overline{1}_{k},m_{k},U_{k}=1]-\mathbb{E}[Y|x,\overline{z}_{k},\overline{1}_{k},m_{k},U_{k}=0]$
parameterizes the effect of $U_{k}$ on Y, and $\beta_{k}=\text{Pr}[U_{k}=1|x,\overline{z}_{k},\overline{1}_{k},m_{k}]-\text{Pr}[U_{k}=1|x,\overline{z}_{k},\overline{1}_{k}]$
parameterizes the effect of $M_{k}$ on $U_{k}$. Each row corresponds
to a different set of $(\alpha_{k},\beta_{k})$ terms. For example,
the top row corresponds to $(\alpha_{0},\beta_{0})$, while the second
row corresponds to $(\alpha_{1},\beta_{1})$, and so on.\protect\label{fig:sens}}
\end{figure}

The bias-factor sensitivity analysis presented above has the drawback of
assessing only the influence of a single unmeasured confounder at a time, for a given
transition. In practice, however, confounders may be unmeasured not only at a given
transition but also at earlier ones. Unfortunately, the bias-factor approach
scales poorly here, requiring $2K$ parameters (a sensitivity-parameter pair $(\alpha_k,\beta_k)$ for each of the $K$ transitions).

Therefore, to assess how sensitive later continuation effects are to
multiple-transition confounding, I undertake a calibrated simulation analysis.
Specifically, I simulate a data-generating process that mirrors the empirical
setting ($A=$ high school completion, $M_{1}=$ college attendance, $M_{2}=$ BA
completion, $M_{3}=$ graduate schooling, $Y=$ log earnings), augmented with one
unobserved baseline confounder ($X_{3}$) and one unobserved intermediate
confounder revealed at each subsequent transition ($Z_{1},Z_{2},Z_{3}$). The
strengths of these confounders are \emph{calibrated} to one of two observed
covariates in the NLSY97 analytic sample, and I then re-estimate $\Delta_{3}$
while cumulatively omitting them from the adjustment set.

I simulate from an expanded version of the data-generating process presented in
the Simulation Study in the main text, additionally drawing a third mediator $M_{3}$ and an intermediate confounder $Z_3$, and letting $Y$ depend
on these variables as well. I report two calibrations, calibrating all four
unobserved confounders to the partial associations of (i) AFQT and (ii)
high-school GPA with the outcome and with each transition.

For each of $B=500$ simulated samples, I estimate
$\Delta_{3}$ using the regression-with-residuals (RWR) estimator under five
nested adjustment sets: full adjustment ($S_{0}$; includes
$X_{3},Z_{1},Z_{2},Z_{3}$), then cumulatively omitting $X_{3}$ ($S_{1}$),
$Z_{1}$ ($S_{2}$), $Z_{2}$ ($S_{3}$), and $Z_{3}$ ($S_{4}$). Under $S_{0}$ the
estimator recovers the true value of $\Delta_{3}$.

Figure~\ref{fig:calib-sens} presents the results. Under either calibration,
even the full history of unobserved confounders---one at every stage, each as
strong (at every transition it affects) as AFQT or high-school GPA---inflates
the estimate by only $0.012$--$0.015$, i.e., roughly $8$--$9\%$ of the observed
value. The implied true continuation effect remains at approximately
$.15$. The estimate is therefore fairly robust to unobserved confounding
of a magnitude comparable to that of observed covariates. In any case, because
$\Delta_{3}$ remains a modest component of the total effect, confounding would not alter the central takeaway -- that the effect of high-school completion on earnings operates
primarily through its direct effect rather than through graduate schooling -- and if anything, would reinforce this conclusion.

\begin{figure}[!htbp]
\centering
\includegraphics[width=0.85\textwidth]{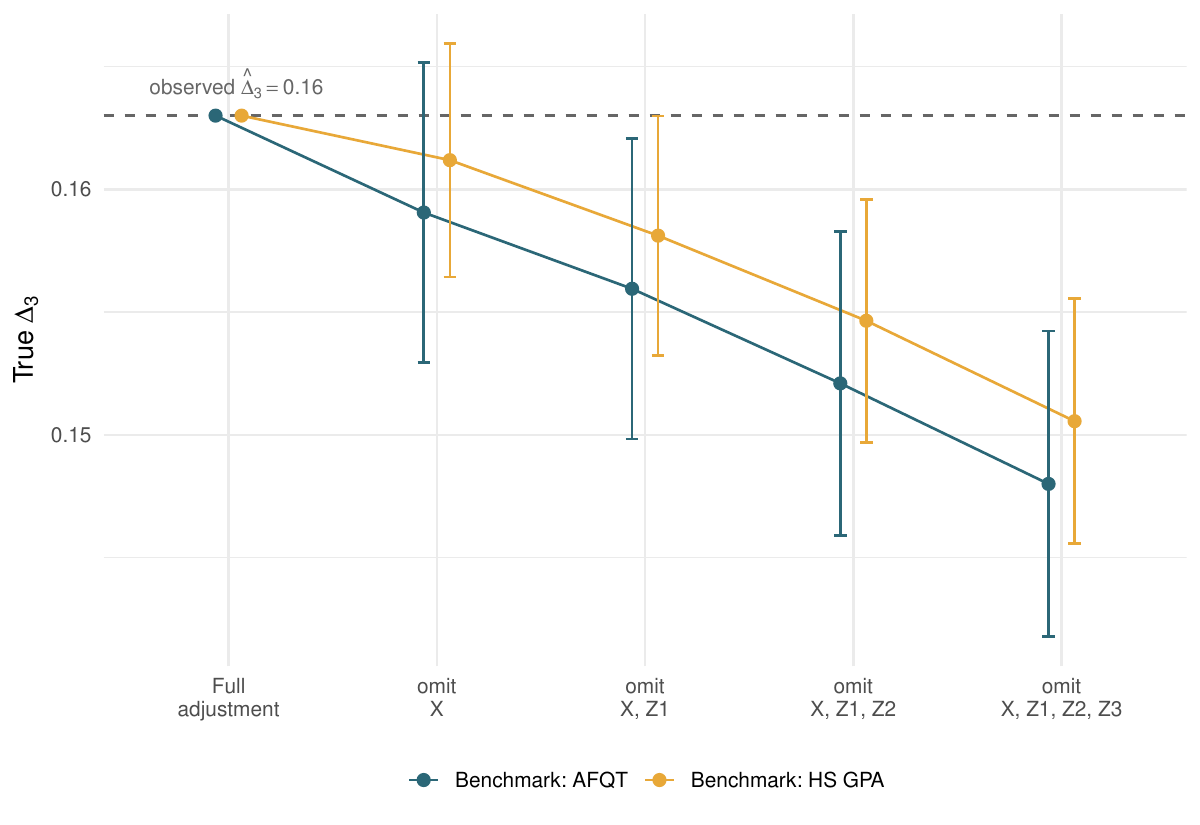}
\caption{Calibrated sensitivity analysis of 
$\Delta_{3}$ to cumulative multiple-transition confounding. The observed
estimate $\hat{\Delta}_{3}$ (dashed line) is held fixed; each point plots the
implied true $\Delta_{3}$ under a scenario that leaves the labeled confounders
unmeasured, for confounders calibrated to AFQT (dark teal) and high-school GPA
(gold). Bars are 95\% Monte Carlo intervals ($B=500$, $n=2{,}000$; RWR
estimator).}
\label{fig:calib-sens}
\end{figure}

\FloatBarrier
\newpage{}
\clearpage

\section{Further details on variable construction and education groups\protect\label{app:further-variable}}

\subsection*{Variable construction}

In an effort to satisfy the sequential ignorability assumption (Assumption
\ref{assu:SI}), I include a large array of covariates in my models
for the effects of completing educational transitions on labor market
outcomes. Figure \ref{fig:DAG-empirical-1} summarizes my assumed
data-generating process for the empirical example. In addition to
including information on respondent demographics (gender, race, ethnicity,
age at 1997), and observed pre-college performance such as overall
high school GPA and test score on the Armed Services Vocational Aptitude
Battery (ASVAB), I include detailed information on socioeconomic background
(parental education, parental income, parental assets, co-residence
with both biological parents, presence of a paternal figure, rural
residence, southern residence), an index of substance use, an index
of delinquency, and whether the respondent had any children by age 18,
and peer and school-level characteristics (measures of peers' college
expectations and behaviors). Both parental income and parental asset
variables are transformed to 2023 dollars. 

Since my proposed decomposition also facilitates the inclusion of
a distinct set of observed intermediate confounders for each transition
to adjust for selection processes that may confound the causal effects
of each transition on earnings (i.e., the $A-Y$ and $M_{k}-Y$ relationships,
for $k\in\{1,\dots K\}$), I include two postsecondary characteristics
($Z$) to adjust for confounders of the effect of BA completion and
graduate school attendance on earnings, namely, field of study, and
college GPA. Specifically, I use college self-reported major field
of study, drawing on the NLSY survey instrument asking respondents
about their choice of major in each month in which they were enrolled
in college, and using a dummy variable to denote whether a
respondent majored in a STEM or non-STEM field by age 29. Finally,
college GPA is measured using the respondent's cumulative GPA from
the Post-Secondary Transcript Study. I treat two of the $Z_{k}$ sets
as empty (namely, $Z_{1}$ and $Z_{3}$), assuming that the effects
of the first mediator (college attendance, $M_{1}$) on subsequent
transitions and adult earnings are unconfounded given background characteristics
$(X)$, and that the effects of the third mediator (graduate school
attendance, $M_{3}$) are unconfounded given background characteristics
($X$) and postsecondary characteristics ($Z$). \footnote{To be clear, assuming that $Z_{1}$ and $Z_{3}$ are empty is not
to say that $M_{1}$ and $M_{3}$ are marginally unconfounded; rather,
it means that the set of covariates that confound the effects of $M_{1}$
is assumed to be the same as those that confound the effects of $A$,
and that the set of covariates that confound the effects of $M_{3}$
are assumed to be the same as those that confound the effects of $M_{2}$.
This assumption is in part data-driven, given the few variables observed
chronologically post high school graduation and pre college attendance.} 

How convincingly do I satisfy the sequential ignorability assumption?
Despite the inclusion of a comprehensive set of background covariates
in my models, it is possible that observed variables do not perfectly
proxy for all important confounders jointly affecting education and
earnings. In particular, researchers often argue that important variables,
such as students' innate ability, ambition, and detailed forms of
socioeconomic advantage, confound observational estimates of educational
returns \citep[e.g.][]{Carneiro2011}. While some research suggests
that observational estimates of earnings returns may well capture
actual returns to education and that the degree of observational bias
may be rather small \citep{Card1999}, it is of course impossible
to quantify the true extent of the bias in the estimates I produce.
The sensitivity analysis described above provides a step towards this
goal.

I note that my assumption of ignorability of $M_{3}$ without conditioning
on intermediate variables $Z_{3}$ is perhaps the strongest assumption
I make. For example, many individuals take time off to work before
enrolling in graduate school, and labor market experience and earnings
gained in the interim period between college completion and graduate
school enrollment may confound the latter variable's effects on earnings.
Nevertheless, including a measure of labor market characteristics
for this period is difficult because some respondents enroll directly
in graduate school after BA completion, such that pre-graduate school
earnings variables would be undefined for these individuals.

Table \ref{tab:Conditional-means} shows conditional means of respondent
attributes $X$ and $Z$ for the full (imputed) and restricted (non-imputed)
samples, showing first the mean among the full population of high
school goers, and progressively restricting the sample from (i) high
school (HS) non-completers, to (ii) HS graduates, to (iii) college
attendees and, finally, to (iv) BA completers. Imputed and non-imputed
means---shown without and with brackets, respectively---are highly
similar across variables. As I progressively restrict the sample to
those who attained higher educational levels, variables capturing
components of socioeconomic advantage (such as parental income, parental
education and household net worth) increase monotonically in value.
Background covariates measuring aspects of the school environment
(such as peers' college expectations---which is an indicator for whether
over $90\%$ of a respondent's peers expected to go to college) behave similarly.
I also see that students who progress to higher educational levels
have higher levels of pre-college ability: HS non-completers have
on average an ASVAB Percentile score of 22.3, compared with only in
excess of 70 among BA completers. Similarly, college-goers' average
high school GPA is approximately $.5$ higher than high school graduates
overall (regardless of whether or not they proceed to college). Nevertheless,
the association between high school GPA and attainment declines at
higher educational levels: BA completers have only on average a $.11$
higher a high school GPA than the pooled group of college goers, irrespective
of their BA completion status. At this stage, college GPA appears
to matter more: college goers overall have on average a college GPA
of $2.77$, while BA completers' average college GPA is $3.07$.

\subsection*{Educational groups: raw mean earnings}

Table \ref{tab:Earnings-Descriptive} (column 2) presents the proportion
of individuals who have attained each level of education constructed
above. By age 22, a small, but not insignificant, proportion of individuals
who enroll in high school do not complete their studies ($13\%$), and by this same age, just over $40\%$ of individuals have attended
a 4-year college. By age 29, $29\%$ of individuals have attained a Bachelor's degree or higher. These
estimates of high school completion and BA completion align closely
both with those reported in previous studies that employ the NLSY97
(e.g. \citet{Scott-Clayton}), as well as with those reported in the
Current Population Survey (CPS). Table \ref{tab:Earnings-Descriptive}
(columns 3-4) also shows mean log earnings by educational group (column
3), alongside the estimated gap between these means and mean log earnings
among high school non-completers (column 4). High school dropouts
earn an average of $\text{9.07}$ log earnings, while groups with
higher levels of attainment earn successively more than high school
dropouts, though at a decreasing rate. High school graduates earn
on average $1.11$ log earnings more than high school non-completers,
implying an earnings premium in excess of $200\%$ ($\exp(1.11)-1$),
while college goers earn on average $1.5$ log earnings more than
high school non-completers (or 0.6 log earnings more than high school
graduates). At the highest end, graduate school goers earn on average
$10.88$ log earnings. These educational premia are extremely high,
since they reflect both the causal effect of a given educational level
as well as the effects of individual, geographic and family factors
correlated both with attainment and with adult earnings. To net out
these patterns of selection, we need to turn to estimates of the MPSE
decomposition, as well as its constituent components.

\newpage{}
\begin{figure}[H]
\begin{centering}
\includegraphics[width=\linewidth,keepaspectratio]{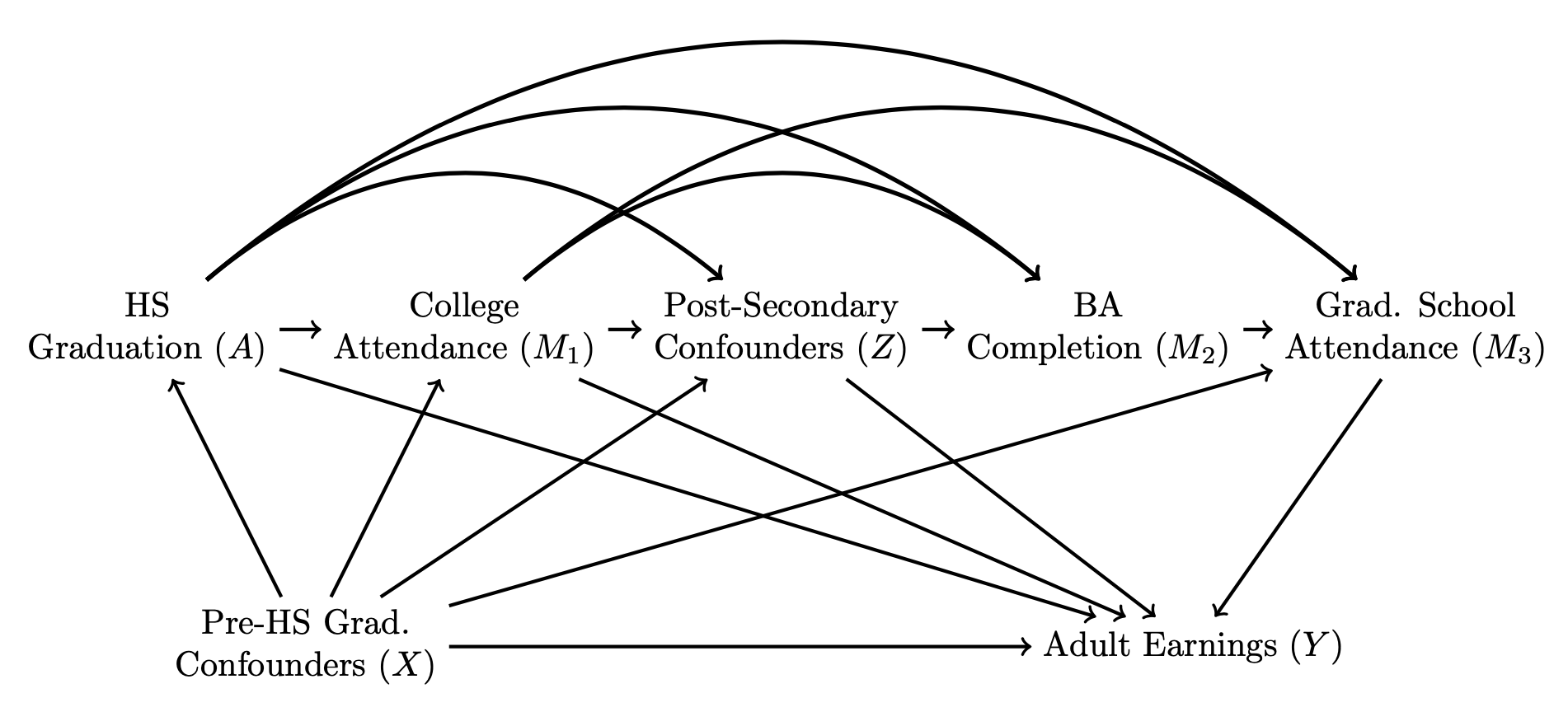}
\par\end{centering}
\caption{DAG showing the hypothesized causal relationships between high school
completion $A$ and adult earnings $Y$ via mediators $M_{1}$, $M_{2}$
and $M_{3}$. \protect\label{fig:DAG-empirical-1}}
\end{figure}

\newpage{}

\begin{table}[h!]
\rotatebox{90}{ \begin{minipage}{\textheight}

\caption{Conditional means of background and college-level attributes by sample
type (imputed dataset and non-imputed dataset). \protect\label{tab:Conditional-means}}

\renewcommand{\arraystretch}{1}
 \renewcommand\theadfont{}
\vspace{0.5em}
\noindent\begin{raggedright}

\centering
\begingroup\fontsize{11pt}{13pt}\selectfont
\begin{tabular}{llllll}
   &  & & & & \\ \hline  \multicolumn{1}{c}{ }  & \multicolumn{1}{c}{Full Population}  & \multicolumn{1}{c}{HS Non-Completers} & \multicolumn{1}{c}{HS Graduates} & \multicolumn{1}{c}{College Goers} & \multicolumn{1}{c}{BA Completers} \\ \hline  \hline Female & 0.50 (0.50) & 0.44 (0.47) & 0.51 (0.51) & 0.55 (0.56) & 0.57 (0.58) \\ 
  Black & 0.17 (0.16) & 0.25 (0.24) & 0.16 (0.15) & 0.13 (0.10) & 0.10 (0.09) \\ 
  Hispanic & 0.14 (0.12) & 0.20 (0.18) & 0.13 (0.12) & 0.09 (0.08) & 0.08 (0.07) \\ 
  Parental Income & 80,235 (79,243) & 43,308 (43,863) & 85,898 (83,865) & 110,240 (109,636) & 118,026 (114,329) \\ 
  Parental Education & 12.85 (12.84) & 11.18 (11.24) & 13.11 (13.04) & 14.06 (14.13) & 14.38 (14.34) \\ 
  Household Net Worth & 175,513 (176,097) & 61,997 (62,318) & 192,920 (190,959) & 272,606 (284,029) & 306,177 (307,101) \\ 
  Lived w Biological Parents & 0.52 (0.51) & 0.30 (0.28) & 0.55 (0.55) & 0.67 (0.69) & 0.72 (0.73) \\ 
  Father Figure Present & 0.75 (0.75) & 0.61 (0.61) & 0.77 (0.76) & 0.84 (0.83) & 0.86 (0.86) \\ 
  Lived in Rural Area & 0.28 (0.30) & 0.25 (0.27) & 0.28 (0.31) & 0.28 (0.32) & 0.28 (0.31) \\ 
  Lived in South & 0.36 (0.35) & 0.44 (0.42) & 0.35 (0.34) & 0.33 (0.30) & 0.31 (0.29) \\ 
  Children by 18 & 0.07 (0.07) & 0.21 (0.24) & 0.05 (0.05) & 0.01 (0.01) & 0.01 (0.01) \\ 
  Substance Abuse Score & 1.09 (1.10) & 1.32 (1.37) & 1.06 (1.07) & 0.86 (0.82) & 0.82 (0.77) \\ 
  Delinquency Score & 1.38 (1.39) & 2.07 (2.13) & 1.27 (1.29) & 0.91 (0.88) & 0.81 (0.79) \\ 
  Peers' College Expectations & 0.56 (0.56) & 0.40 (0.38) & 0.59 (0.58) & 0.69 (0.69) & 0.72 (0.71) \\ 
  Property Stolen at School & 0.24 (0.23) & 0.28 (0.29) & 0.23 (0.22) & 0.21 (0.20) & 0.20 (0.19) \\ 
  Threatened at School & 0.22 (0.23) & 0.30 (0.33) & 0.20 (0.21) & 0.14 (0.13) & 0.13 (0.11) \\ 
  In a Fight at School & 0.16 (0.16) & 0.32 (0.34) & 0.14 (0.13) & 0.07 (0.07) & 0.06 (0.06) \\ 
  ASVAB Percentile & 48.36 (48.52) & 22.27 (21.77) & 52.37 (52.02) & 67.22 (69.66) & 70.66 (71.82) \\ 
  High School GPA & 2.84 (2.81) & 2.12 (2.14) & 2.95 (2.90) & 3.30 (3.33) & 3.41 (3.43) \\ 
   \hline Stem Major &  &  &  & 0.17 (0.19) & 0.18 (0.19) \\ 
  College GPA &  &  &  & 2.78 (2.86) & 3.07 (3.12) \\ 
   \hline Earnings (\$) & 46,505 (45,151) & 21,597 (21,566) & 46,505 (45,151) & 64,984 (64,595) & 72,374 (70,212) \\ 
  Log ( Earnings + c ) & 10.03 (10.06) & 9.07 (9.14) & 10.03 (10.06) & 10.57 (10.65) & 10.73 (10.78) \\ 
   \hline \end{tabular}
\endgroup

\par\end{raggedright}
\vspace{-0em}

Note: Numbers denote means for the imputed sample (non-imputed sample).
Means are adjusted for multiple imputation via Rubin\textquoteright s
(1987) method, and all statistics are calculated using NLSY97 sampling
weights.

\end{minipage}}
\end{table}

\newpage{}
\begin{table}[h]
\caption{Means of observed log earnings by educational participation, and earnings
gaps (versus high school non-completers).\protect\label{tab:Earnings-Descriptive}}

\renewcommand{\arraystretch}{1.2}

\vspace{0.5em}
\noindent\begin{raggedright}

\centering
\begingroup\fontsize{12pt}{15pt}\selectfont\setlength{\tabcolsep}{3pt}
\begin{tabular}{lrll}
   \hline  \multicolumn{1}{c}{Group }  &  \multicolumn{1}{c}{Population Proportion} & \multicolumn{1}{c}{Log Earnings}  & \multicolumn{1}{c}{Gap (vs HS Non-Completers)} \\ \hline  \hline HS Non-Completers & 0.13 & 9.07 (0.05) &  \\ 
  HS Graduates & 0.87 & 10.18 (0.02) & 1.11 (0.05) \\ 
  College Goers & 0.41 & 10.57 (0.03) & 1.5 (0.05) \\ 
  BA Completers & 0.29 & 10.73 (0.03) & 1.66 (0.06) \\ 
  Grad. School Goers & 0.09 & 10.88 (0.05) & 1.81 (0.07) \\ 
  \end{tabular}
\endgroup
\par\end{raggedright}
\vspace{-0em}

{\footnotesize
Note: The category \textquotedbl High School Non-Completers\textquotedbl{}
captures all individuals who attended high school but did not obtain
a high school diploma; \textquotedbl High School Graduates\textquotedbl{}
refers to those individuals who graduated high school, regardless
of their subsequent educational experiences (i.e., whether or not
they proceeded to college); \textquotedbl College Goers\textquotedbl{}
refers to individuals who attended a 4-year college, irrespective
of whether they completed their degree; \textquotedbl BA Completers\textquotedbl{}
denotes individuals who completed a Bachelor's degree, while \textquotedbl Grad.
School Goers\textquotedbl{} captures individuals who participated
in a graduate-level degree program. A small constant of $\$1,000$
is added to observed earnings before taking the log. All statistics
are computed with a monotonicity assumption imposed on the observed
data (i.e. such that all individuals who complete a given educational
level are coded as having completed all prior levels). All statistics
are calculated using NLSY97 sampling weights, and standard errors
are in parentheses.\par}
\end{table}

\FloatBarrier
\newpage{}

\section{Results without imputation of missing covariates\protect\label{App:Non-MI}}

In the main text, I report estimates of the MPSE decomposition for
the ATE of high school completion on logged annual earnings for the
full sample of NLSY97 respondents with non-missing educational information
and non-missing earnings ($N=7,305$). A very large number (approximately
$50\%$) of these respondents are missing information on one or more
of the covariates $(X,Z)$ used in the models in order to identify
the decomposition components. Table~\ref{tab:missing_X_by_ed} summarizes 
missingness patterns for the background covariates $X$ by respondents' educational
attainment---which are all self-reported by respondents. Missingness is generally modest for most pretreatment covariates, though certain 
variables (e.g., household net worth and ASVAB percentile) exhibit higher rates of 
non-response (at $25\%$ and $20\%$ of the sample overall, respectively). Non-response is more common among respondents with lower levels of schooling.

Table~\ref{tab:missing_Z_by_ed} clarifies the nature of missingness in the intermediate 
covariates  $Z$ (STEM major and college GPA). These intermediate covariates are observed only for respondents who attend college, but even within this group missingness 
is substantial. One source of missingness is transcript non-receipt: these variables are derived from the NLSY97 
Postsecondary Transcript Study, and are therefore observed only for respondents for whom a 
transcript was successfully obtained. Transcript non-response is the primary driver of 
missingness: institutions sometimes did not supply transcripts, supplied incomplete records, 
or provided degree-program information without course-level grades.\footnote{See 
\emph{NLSY97 Appendix 12: Postsecondary Transcript Study Documentation} 
(https://www.nlsinfo.org/content/cohorts/nlsy97/other-documentation/codebook-supplement/appendix-12-post-secondary-transcript-study).} Another potential source of missingness may be that students drop out of college before declaring a major. Consistent with this, among college-goers who do not complete a 
BA, approximately 10\% are missing information on degree major, compared with 1.5\% of BA completers. This pattern raises a natural concern for my strategy to impute missing values for STEM. If missingness reflects a student's failure to declare a major, then this missingness should be reflected as \textit{true} missingness in an additional level of the STEM variable, rather than being approached as a missing data problem. At the same time, it is difficult to disentangle true non-declaration of a 
major from missingness caused by institutional non-compliance with transcript requests; 
students who leave college early may also disproportionately attend 
institutions that are less responsive to the NLSY's Postsecondary Transcript Study. As a result, the 
observed missingness conflates measurement limitations with underlying educational 
progression in ways that cannot be fully separated.

To assess the sensitivity of my primary conclusions
to these issues, I undertake two exercises. First, I replicate my DML and RWR estimates on a non-imputed
analytic sample, dropping observations
with missing values ($N=3,735)$. Figure \ref{fig:Decomp-non-mi} below
shows the results of this exercise. For both estimation procedures,
results under multiple imputation and non-imputation are highly similar.
As is to be expected, imputation reduces standard errors significantly,
especially for the parametric RWR procedure. Further, the greatest
variability between imputed and non-imputed results comes from effects
pertaining to high school completion, perhaps because patterns of
missingness are correlated with educational attainment. Despite this,
because the total effect $\tau_{0}$ and direct effect $\theta_{0}$
are similarly attenuated in the imputed sample, the overall conclusion
about the importance of the direct effect in explaining the ATE remains
unaffected.

Second, to assess whether missingness in field of study reflects meaningful differences 
in educational progression versus measurement limitations, I re-estimated the 
entire MPSE decomposition excluding the STEM indicator from the intermediate 
confounder set $Z$. The resulting estimates (presented in Table \ref{tab:Decomp-no-stem}) are nearly identical to those 
reported in the main text, indicating that the main empirical findings 
are not sensitive to whether field of study is included in the intermediate 
confounder set. Nevertheless, the question of whether major declaration itself constitutes an important intermediate transition (prior to BA completion) is an important direction for future research.

\newpage{}

\begin{table}[h]
\caption{Proportion of missing background covariates by educational level}
\label{tab:missing_X_by_ed}
\centering
\begingroup\fontsize{11pt}{13pt}\selectfont
\renewcommand{\arraystretch}{1.3}

\begin{tabular}{l r 
                r r r r r r r r}
\hline
& 
& \rotatebox{90}{Parental income}
& \rotatebox{90}{Parental education}
& \rotatebox{90}{Household net worth}
& \rotatebox{90}{ASVAB percentile}
& \rotatebox{90}{High school GPA}
& \rotatebox{90}{Peers' expectations (75th)}
& \rotatebox{90}{Peers' expectations (90th)}
& \rotatebox{90}{Property stolen at school} \\
\hline
\hline
HS Non-Completers       & 1144 & 0.058 & 0.069 & 0.233 & 0.302 & 0.097 & 0.021 & 0.021 & 0.034 \\
HS Graduates       & 3443 & 0.033 & 0.036 & 0.238 & 0.203 & 0.008 & 0.015 & 0.015 & 0.011 \\
College Goers  &  885 & 0.040 & 0.030 & 0.270 & 0.159 & 0.016 & 0.006 & 0.006 & 0.005 \\
BA Completers       & 1246 & 0.019 & 0.023 & 0.255 & 0.127 & 0.009 & 0.007 & 0.007 & 0.006 \\
Grad. School Goers    &  587 & 0.029 & 0.020 & 0.307 & 0.148 & 0.019 & 0.007 & 0.007 & 0.005 \\
\hline
\end{tabular}
\endgroup
\end{table}

\begin{table}[h]
\caption{Proportion of missing intermediate covariates (transcript-based) by educational level}
\label{tab:missing_Z_by_ed}
\centering
\begingroup\fontsize{11pt}{13pt}\selectfont
\renewcommand{\arraystretch}{1.3}

\begin{tabular}{l r r r}
\hline
& $N$ 
& \rotatebox{0}{STEM major} 
& \rotatebox{0}{College GPA} \\
\hline
\hline
HS Non-Completers   & 1144 & 1.000 & 1.000 \\
HS Graduates        & 3443 & 1.000 & 1.000 \\
College Goers       &  885 & 0.103 & 0.357 \\
BA Completers       & 1246 & 0.015 & 0.272 \\
Grad. School Goers  &  587 & 0.019 & 0.274 \\
\hline
\end{tabular}

\endgroup
\end{table}

\begin{table}[h]
\caption{Direct Effects ($\Delta_{k}$), Probabilities ($\pi_{k}$) and Covariance
Terms ($\eta_{k}$) Involved in Decomposition via Debiased Machine-Learning
(DML), without STEM degree. \label{tab:Decomp-no-stem}}
\centering
\begingroup\fontsize{11pt}{13pt}\selectfont\setlength{\tabcolsep}{3pt}
\renewcommand{\arraystretch}{1.5} 
\begin{tabular}{lllllllllll}
   \hline & $\Delta_0$ & $\Delta_1$ & $\Delta_2$ & $\Delta_3$ & $\pi_1$ & $\pi_2$ & $\pi_3$ & $\eta_1$ & $\eta_2$ & $\eta_3$  \\ \hline DML & 0.462 & 0.197 & 0.442 & 0.145 & 0.427 & 0.554 & 0.313 & 0.006 & 0.007 & -0.001 \\ 
   & (0.059) & (0.034) & (0.043) & (0.039) & (0.009) & (0.014) & (0.017) & (0.007) & (0.007) & (0.009) \\ 
   \hline
\end{tabular}
\endgroup
\end{table}

\begin{figure}[h]
\begin{centering}
\includegraphics[width=\linewidth,keepaspectratio]{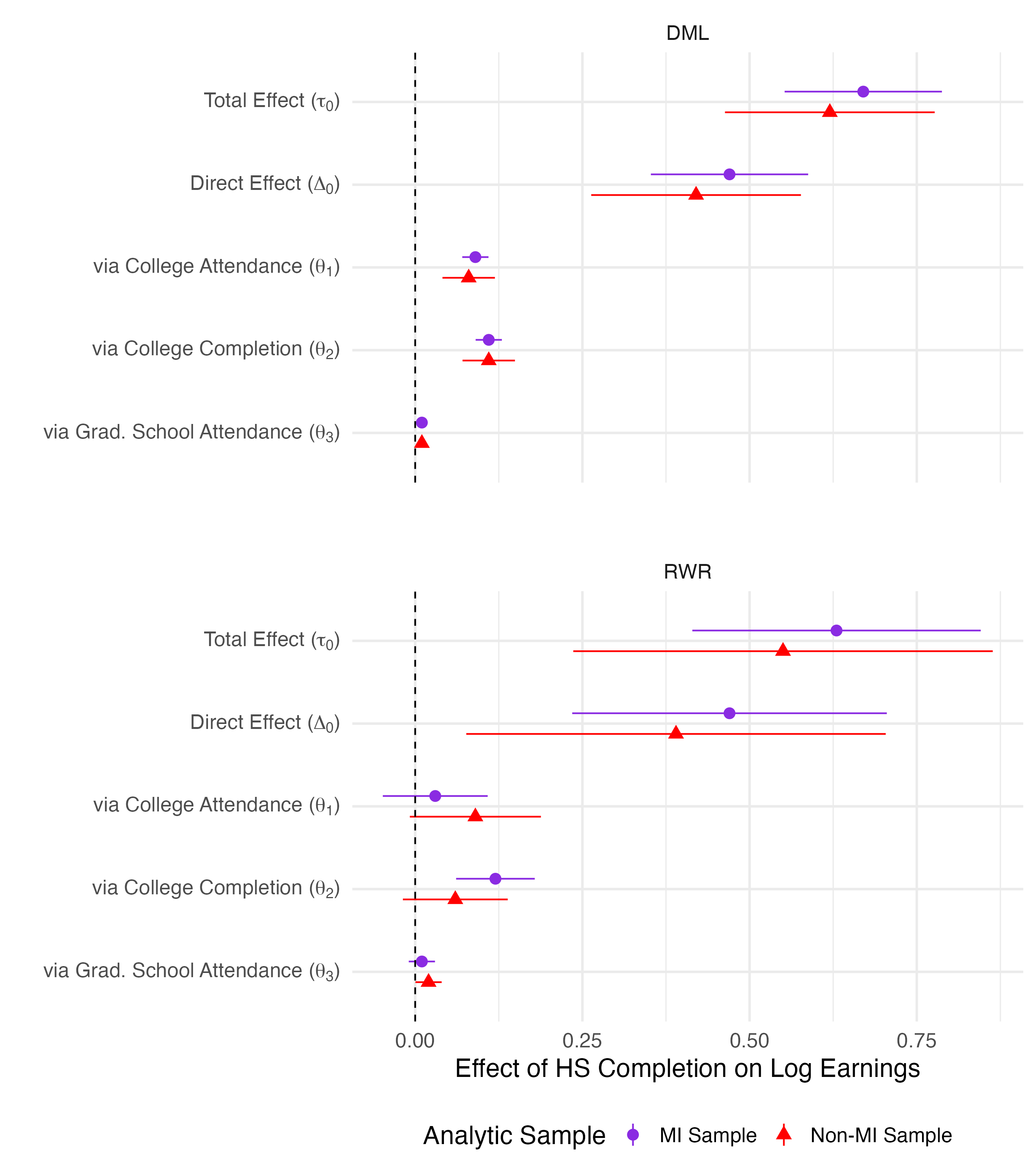}
\par\end{centering}
\caption{Decomposition of the Average Total Effect (ATE) of High School Graduation
on Logged Earnings Under Multiple Imputation (MI) and Under Dropping
Observations with Missing ($X,Z$) Values. Results with multiple imputation
(purple lines) are reproduced from the main text ($N=7,305$); results
without multiple imputation (red lines) employ a sample restricted
to respondents with observed values for all covariates used ($N=3,735$).\protect\label{fig:Decomp-non-mi}}
\end{figure}

\FloatBarrier
\newpage{}

\section{Results under alternative definitions of earnings\protect\label{App:Constant}}

In the main text, I report estimates of the MPSE decomposition components
for the ATE of high school completion on logged annual earnings. Logged
annual earnings in the main text are defined as the log of observed
annual earnings plus a small constant of $\$1,000$, in order to accommodate
respondents with zero observed annual earnings. In order to assess
the sensitivity of the reported results to the choice of this constant,
I replicate the main analyses under alternative definitions of earnings.
Figure \ref{fig:Decomp-constant} reports estimates of the direct
and indirect effects under a series of different constants $c$ added
to pre-logged annual earnings, for $c\in\{10,100,1000\}$, while Figure
\ref{fig:Decomp-raw} shows estimates of these direct and indirect
effects for observed annual earnings in dollar values. Beginning with
Figure \ref{fig:Decomp-constant}, we see that, while for the indirect
effects $\theta_{1}$, $\theta_{2},$ and $\theta_{3}$, both DML
and RWR estimates are quite consistent under these different constants,
estimates of the total effect $\tau_{0}$ as well as the direct effect
$\Delta_{0}$ are quite sensitive to the choice of constant. Specifically,
lower constant values correspond with large increases in the DML estimate
of $\tau_{0}$ from $0.67$ ($c=1000$) to $1.20$ ($c=10$), and
of $\Delta_{0}$ from $0.47$ ($c=1000$) to $0.89$ ($c=10$). This
is because individuals with less than a high school degree are more
likely than their higher-educated counterparts to have zero or low
earnings, making their logged earnings rather sensitive to the choice
of constant. Nevertheless, because the total effect $\tau_{0}$ and
direct effect $\theta_{0}$ are similarly affected by the change in
constant value, the importance of the direct effect in explaining
the ATE of high school completion is reinforced. In particular, the
proportion of the total effect that is direct is estimated to be roughly $70\%$ under each of $c=10,100,1000$. Turning next to
Figure \ref{fig:Decomp-raw}, under DML, we estimate that high school
completion increases earnings in expectation by roughly $\$16,500$,
corresponding to an earnings return of approximately $53\%$ relative to a baseline of \textbf{$\$31,300$ }without high school
graduation ($\mathbb{E}[Y(0)]$). Almost half ($\frac{7824}{16454}\cdot100=47.5\%$)
of the total effect is estimated to operate directly, with $29\%$
and $23\%$ mediated via college attendance and college completion,
respectively. 

\begin{figure}[h]
\begin{centering}
\includegraphics[width=\linewidth,keepaspectratio]{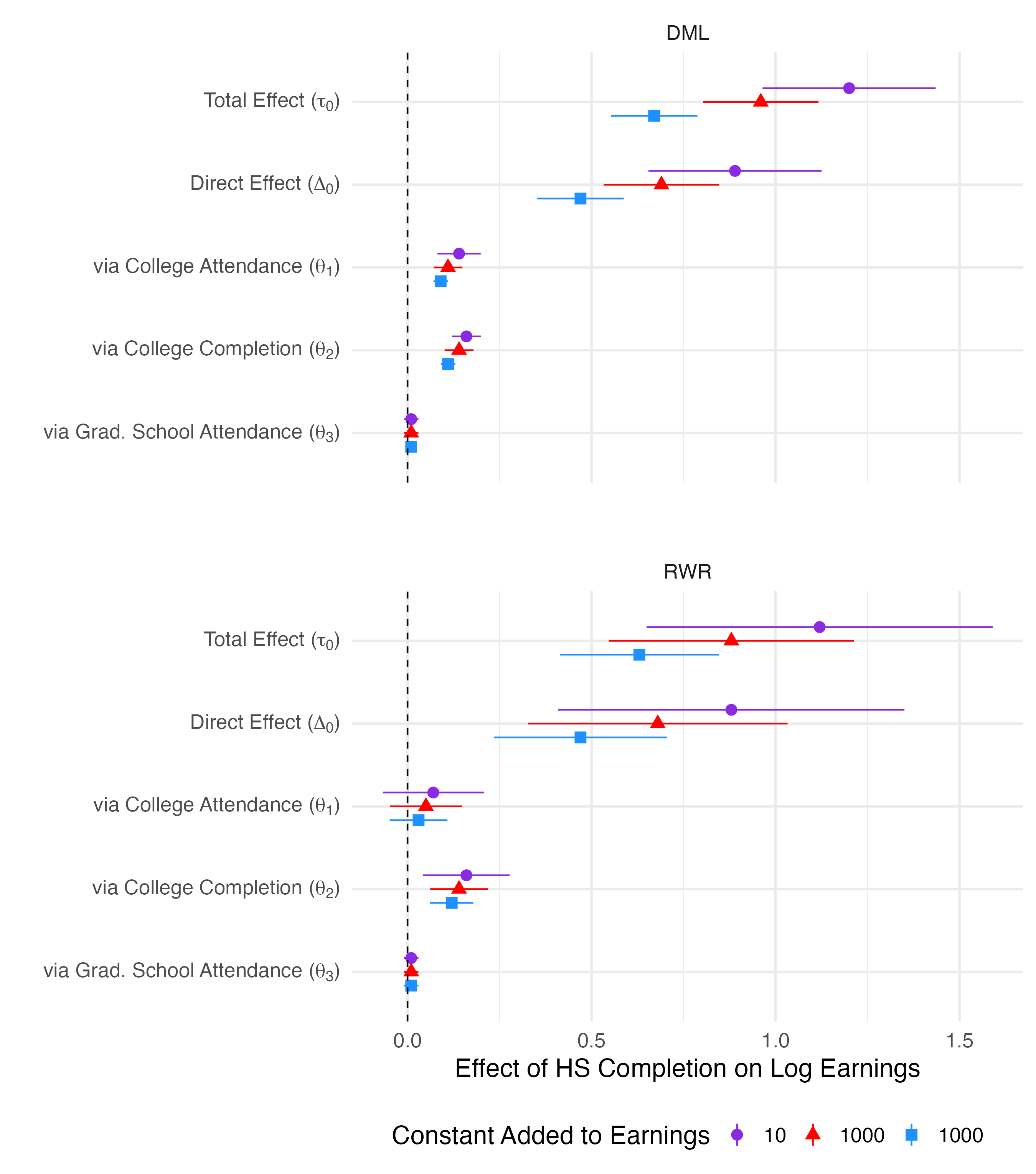}
\par\end{centering}
\caption{Decomposition of the Average Total Effect (ATE) of High School Graduation
on Logged Earnings Under Alternative Definitions of Earnings. The
figure shows estimates of the total effect ($\tau_{0}$) as well as
the continuation effects $\Delta_{0},\tau_{1},\dots,\tau_{K}$
when constants of 10, 100 and 1000, respectively, are added to raw
annual earnings (in dollar amounts) before taking the log. \protect\label{fig:Decomp-constant}}
\end{figure}

\newpage{}

\begin{figure}[h]
\begin{centering}
\includegraphics[width=\linewidth,keepaspectratio]{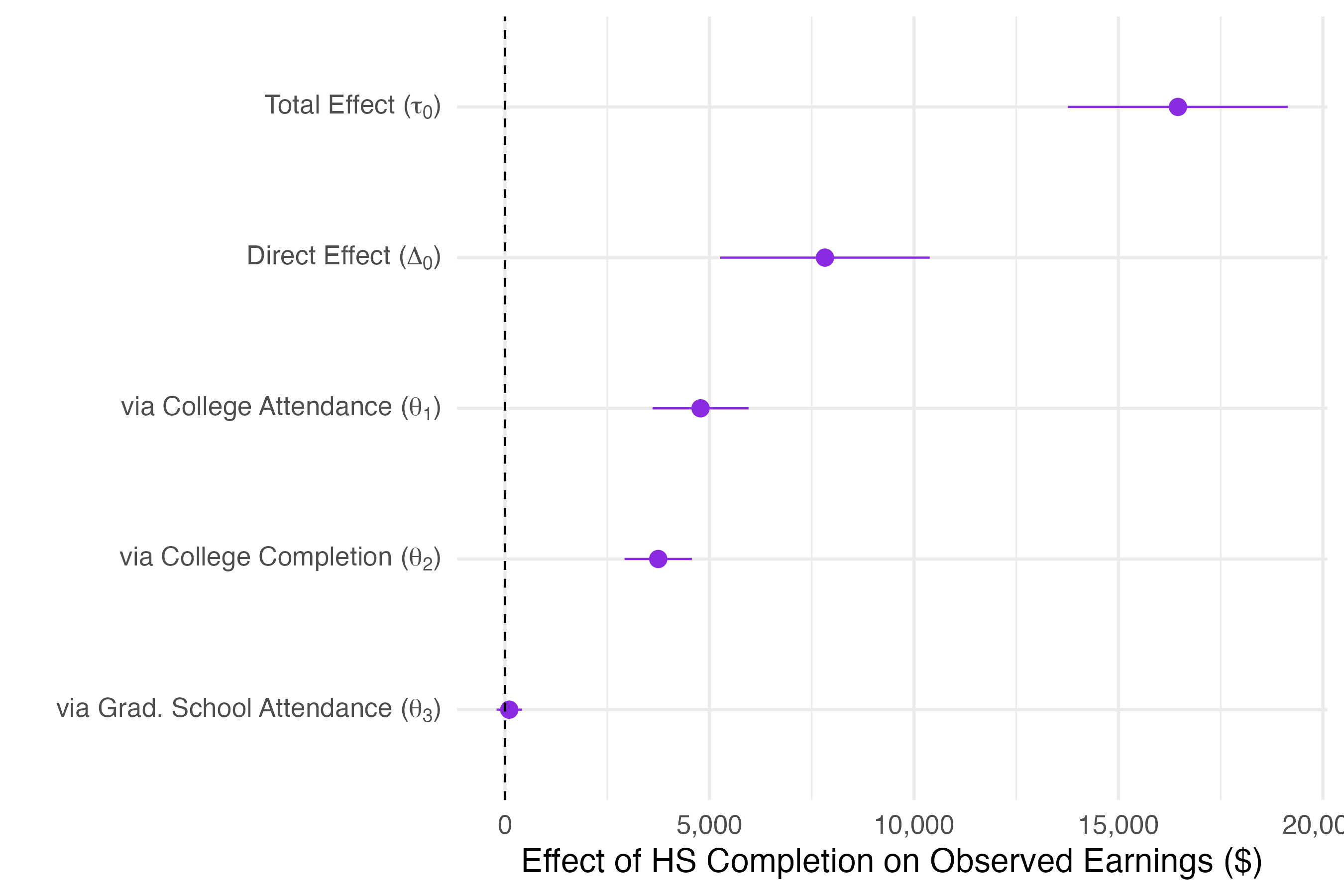}
\par\end{centering}
\caption{Decomposition of the Average Total Effect (ATE) of High School Graduation
on Logged Earnings with Different Definitions of Earnings. The figure
shows estimates of the total effect ($\tau_{0}$) as well as the continuation
effects $\Delta_{0},\tau_{1},\dots,\tau_{K}$ when constants
of 10, 100 and 1000 are added to raw annual earnings (in dollar amounts)
before taking the log. \protect\label{fig:Decomp-raw}}
\end{figure}

\FloatBarrier
\newpage{}

\section{Results under different positivity assumptions\protect\label{App:Positivity}}


Under the positivity assumption maintained in the main text, the decomposition components (the
$\Delta$, $\pi$, and $\eta$ terms, and hence $\tau_0$ and $\theta_1,\dots,\theta_3$) are
nonparametrically identified, and I estimate them by both parametric RWR and DML.

In a second analysis, I relax the positivity assumption by allowing some units to have transition propensities of exactly
zero or one. These units need to be excluded from the analysis, but because the true propensities are unknown, I
operationalize this stage by stage, trimming units whose estimated propensity $\hat{\pi}(X_i)$
falls below $\min_{i:\,T_i=1}\hat{\pi}(X_i)$ or above $\max_{i:\,T_i=0}\hat{\pi}(X_i)$ for that
transition's indicator $T$, retaining a unit only if it lies within the common-support region of
every transition it reaches. Concretely, I trim below/above
$\min_{i:\,A_i=1}\hat{\pi}(X_i)$ / $\max_{i:\,A_i=0}\hat{\pi}(X_i)$ at high-school completion;
$\min_{i:\,M_{1i}=1,\,A_i=1}\hat{\pi}(X_i)$ / $\max_{i:\,M_{1i}=0,\,A_i=1}\hat{\pi}(X_i)$ at
college attendance;
$\min_{i:\,M_{2i}=1,\,M_{1i}=1,\,A_i=1}\hat{\pi}(X_i)$ /
$\max_{i:\,M_{2i}=0,\,M_{1i}=1,\,A_i=1}\hat{\pi}(X_i)$ at BA completion; and
$\min_{i:\,M_{3i}=1,\,M_{2i}=1,\,M_{1i}=1,\,A_i=1}\hat{\pi}(X_i)$ /
$\max_{i:\,M_{3i}=0,\,M_{2i}=1,\,M_{1i}=1,\,A_i=1}\hat{\pi}(X_i)$ at graduate-school attendance.
The propensity scores $\hat{\pi}(X_i)$ are estimated using a super learner composed of the Lasso
and random forests.

Averaging over the imputed datasets, these bounds equal $(0.162,\,0.995)$ at high-school
completion, $(0.022,\,0.979)$ at college attendance, $(0.092,\,0.965)$ at BA completion, and
$(0.166,\,0.583)$ at graduate-school attendance, leading me to exclude, respectively, $590$
completers and $4$ non-completers; $94$ college attendees and $16$ non-attendees; $28$ BA
completers and $8$ non-completers; and $3$ graduate-school attendees and $5$ non-attendees---$626$
of $7{,}305$ units in total, leaving a common-support sample of $6{,}679$. Note that this trimmed sample no longer yields an estimate of the ATE, but rather a conditional treatment effect (and corollary decomposition) for units that have a non-zero probability of attaining or not attaining each transition.

Comparing results from the full sample (assuming positivity) and the common-support sample
(relaxing positivity) in Table~\ref{tab:relax-positivity} and Figure~\ref{fig:relax-positivity} shows that
trimming has little effect on the substantive conclusions. Under both parametric RWR and DML the
total effect of high-school completion on log earnings is essentially unchanged (RWR: $0.628$ vs.\
$0.653$; DML: $0.675$ vs.\ $0.643$), as is the direct effect $\theta_0$, which dominates the
decomposition and accounts for roughly $70$--$75\%$ of the total (RWR: $0.469$ vs.\ $0.507$;
DML: $0.466$ vs.\ $0.464$). The continuation effects are largely stable: in particular, the
graduate-school pathway $\theta_3$ is negligible throughout. The largest change among the underlying components is in $\Delta_3$, the graduate-school gross-effect, which falls from 0.138 to 0.080 under DML (it is essentially unchanged under RWR); this estimate is imprecise, however, and its standard error roughly doubles (from 0.040 to 0.075). More generally, standard errors are larger in the trimmed sample.
 
Overall,
excluding units with extreme estimated propensity scores does not materially alter the main conclusion that the effect of high-school completion on earnings is
predominantly direct, with secondary contributions operating through college completion and graduate school.


\begin{table}[h]
\caption{Direct effects ($\Delta_k$), transition probabilities ($\pi_k$), and covariance terms
($\eta_k$) in the decomposition, estimated by DML and parametric RWR, assuming positivity (full
sample) and relaxing positivity (common-support / trimmed sample). Standard errors, adjusted for
multiple imputation via Rubin's (1987) method, in parentheses. \label{tab:relax-positivity}}
\centering
\begingroup\fontsize{9.2pt}{9.2pt}\selectfont\setlength{\tabcolsep}{3pt}
\renewcommand{\arraystretch}{1.5}
\begin{tabular}{lllllllllll}
   \hline & $\Delta_0$ & $\Delta_1$ & $\Delta_2$ & $\Delta_3$ & $\pi_1$ & $\pi_2$ & $\pi_3$ & $\eta_1$ & $\eta_2$ & $\eta_3$ \\ \hline
   DML, Assuming & 0.466 & 0.200 & 0.444 & 0.138 & 0.427 & 0.555 & 0.314 & 0.007 & 0.004 & -0.002 \\
    & (0.057) & (0.031) & (0.034) & (0.040) & (0.007) & (0.008) & (0.008) & (0.008) & (0.016) & (0.011) \\
   DML, Relaxing & 0.464 & 0.239 & 0.464 & 0.080 & 0.360 & 0.519 & 0.307 & 0.003 & -0.006 & 0.000 \\
    & (0.062) & (0.040) & (0.045) & (0.075) & (0.023) & (0.018) & (0.011) & (0.008) & (0.015) & (0.021) \\
   RWR, Assuming & 0.469 & 0.116 & 0.466 & 0.163 & 0.374 & 0.515 & 0.159 & -0.016 & 0.083 & 0.028 \\
    & (0.115) & (0.093) & (0.097) & (0.114) & (0.007) & (0.015) & (0.051) & (0.015) & (0.077) & (0.022) \\
   RWR, Relaxing & 0.507 & 0.176 & 0.462 & 0.176 & 0.320 & 0.479 & 0.211 & -0.002 & 0.040 & 0.014 \\
    & (0.100) & (0.106) & (0.124) & (0.174) & (0.019) & (0.021) & (0.066) & (0.013) & (0.085) & (0.029) \\
   \hline
\end{tabular}
\endgroup
\end{table}

\vspace{-10em}
\begin{figure}[ht]
\centering
\includegraphics[width=\linewidth]{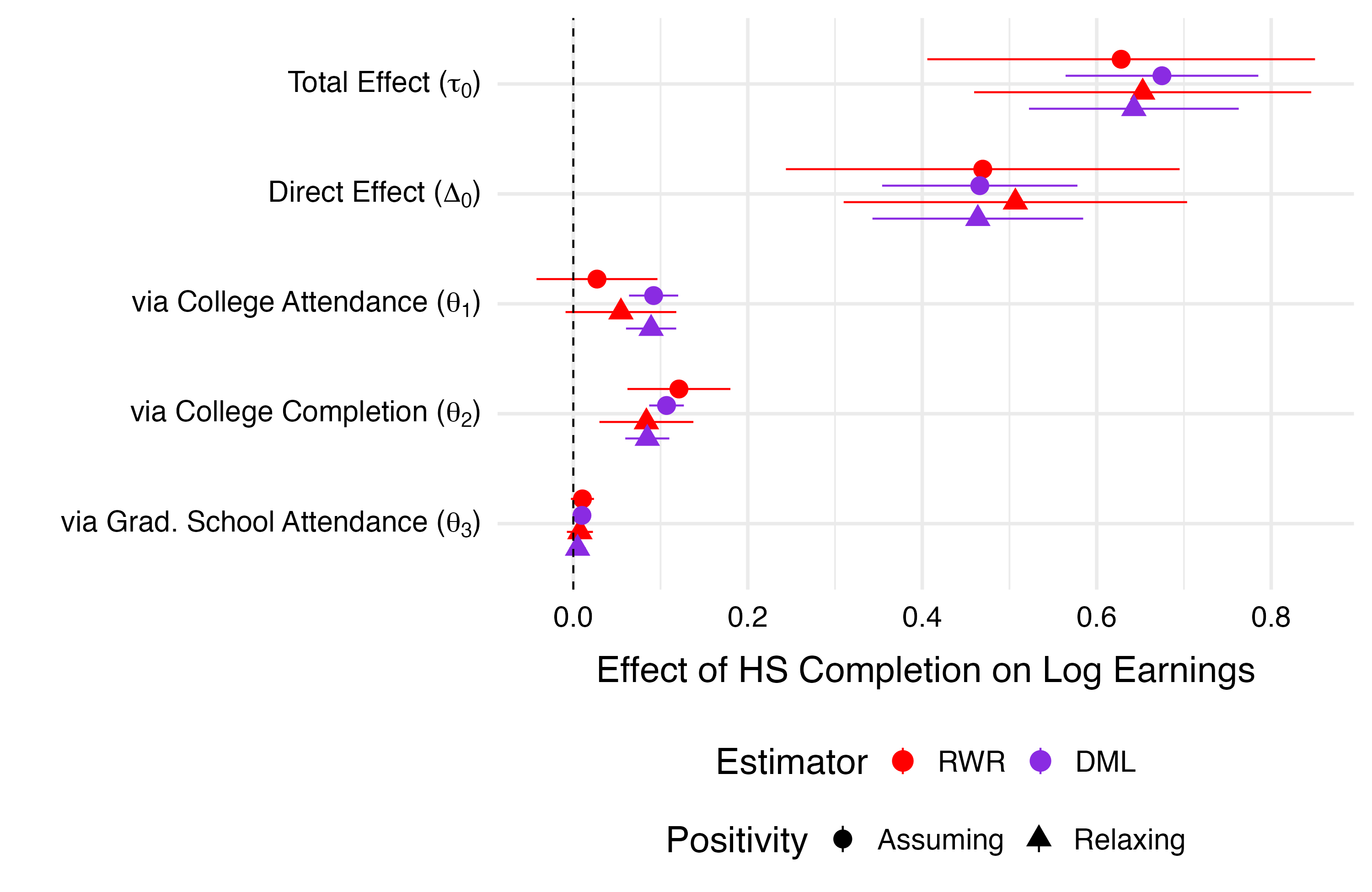}
\caption{Total effect ($\tau_0$), direct effect ($\theta_0$), and pathway-specific continuation
effects via college attendance ($\theta_1$), college completion ($\theta_2$), and graduate-school
attendance ($\theta_3$) of high-school completion on log earnings, estimated by parametric RWR and
DML, assuming positivity and relaxing positivity (common-support sample). Bars are $95\%$
confidence intervals. \label{fig:relax-positivity}}
\end{figure}

\FloatBarrier
\newpage{}
\clearpage
\section{Results for the NLSY79 cohort\protect\label{App:Extension}}

Because NLSY97 respondents are observed only through age 36, the main 
analysis uses logged average earnings over ages 32--36. To evaluate whether 
later-life earnings alter the decomposition---especially for graduate 
schooling, whose returns may rise after early adulthood—I also report results 
using the NLSY79 cohort, for whom we can observe earnings up until age 44.

All variables are defined analogously, although several differences between the 
two surveys are of note.  First, the NLSY79 contains a narrower set of baseline covariates. Whereas the 
NLSY97 includes high school GPA and a rich set of school characteristics 
(disciplinary environment, parental assets, school safety, and peer-context 
measures), the NLSY79 does not. Consequently, the covariate vector in this 
replication consists of a more limited set of demographic and family 
background characteristics: indicators of gender and race/ethnicity; parental 
socioeconomic status (income, educational attainment, and occupation); family 
structure and sibship size; a cognitive ability measure based on the Armed 
Forces Qualification Test (AFQT); and the respondent’s expectations for their 
own educational attainment. The NLSY79 also provides several proxies for 
cultural resources in the home at age 14—specifically, whether the household 
regularly received magazines or newspapers and whether any household member 
held a library card. These measures capture elements of the early educational 
environment but are considerably narrower than the school-context and academic 
performance variables available in the NLSY97.

Second, the NLSY79 lacks transcript-based measures of college major or GPA, so 
there are no intermediate covariates $Z$ between the BA and graduate-school 
transitions. Identification of the MPSE components therefore relies on 
stronger assumptions than in the NLSY97 analysis, as academic performance and 
field-of-study sorting cannot be adjusted for when isolating the contribution 
of BA and post-BA educational pathways.

Figures~\ref{fig:Decomp-constant-nlsy79} and \ref{fig:Decomp-raw-nlsy79} report 
estimates of the direct and indirect components of the MPSE decomposition for 
the NLSY79 cohort, for earnings at age 32-36 (top panel) and at age 35-44 (bottom panel). Figure~\ref{fig:Decomp-constant-nlsy79} presents results 
when a constant $c\in\{10,100,1000\}$ is added to annual earnings prior to 
logging, while Figure~\ref{fig:Decomp-raw-nlsy79} shows the decomposition for 
observed earnings in dollar values.  Table~\ref{tab:Decomp-nlsy79} also summarizes the estimated direct effects 
($\Delta_{k}$), stage-specific transition probabilities ($\pi_{k}$), and 
covariance components ($\eta_{k}$) for the NLSY79 cohort for earnings at age 32-36 and at age 35-44.
\footnote{
There are several differences in the estimated effects in the NLSY79 versus the NLSY97 cohorts. Notably, across specifications, the total effect of high school 
completion is smaller in the NLSY79 cohort ($\tau_{0}=0.4$ for log$(Y+1000)$, 
compared with $\tau_{0}=0.675$ in the NLSY97), consistent with lower marginal returns to schooling for the earlier cohort. The contribution of college attendance to the total 
effect is also estimated to be negligible ($\theta_{1}=0.03$) compared with in the NLSY97 cohort ($\theta_{1}=0.092$), which may reflect the fact that mediation via ``college attendance'' in this decomposition 
captures any postsecondary enrollment, as opposed to mediation via 4-year attendance in the NLSY97, as the NLSY79 public use dataset does not distinguish clearly between two-year and four-year 
colleges.
}

Within the NLSY79 cohort, across the earlier and later earnings windows, results are remarkably stable for the 
log-transformed specifications. As shown in 
Figure~\ref{fig:Decomp-constant-nlsy79}, estimates of $\tau_{0}$ and 
$\theta_{k}$ are nearly indistinguishable across log$(Y+10)$, log$(Y+100)$, and 
log$(Y+1000)$, with only mild divergence when the largest constant is applied. By contrast, decompositions based on raw earnings 
(Figure~\ref{fig:Decomp-raw-nlsy79}) exhibit larger differences across age 
windows. The larger total effect $\tau_{0}$ at ages 35–44 than at ages 32–36, for example, 
results from the fact that the same proportional earnings premium associated 
with high school completion translates into a larger dollar difference as 
overall earnings rise with age. Even so, the 
ordering and relative magnitudes of the direct and indirect pathways remain 
consistent across specifications.

Importantly, the graduate-school pathway contributes almost identically to the 
total effect at ages 32–36 and 35–44, even though one might reasonably expect 
its contribution to increase at later ages as the graduate-school earnings 
premium becomes larger. 

Table~\ref{tab:Decomp-nlsy79} clarifies the source of this stability. The continuation effect ($\Delta_{3}$) nearly doubles between the 
two windows, but the covariance component $\eta_{3}$ declines, indicating a 
weaker alignment between the propensity to attend graduate school and the 
incremental graduate-school earnings premium at later ages. Early in adulthood, 
individuals most likely to pursue graduate education tend to experience the 
largest immediate gains from doing so, generating a larger covariance term. By 
ages 35–44, graduate-school earnings advantages are more broadly distributed 
across degree holders rather than being concentrated among those with the 
highest propensity to attend. Because the overall probability of traversing the 
entire trajectory from high school completion to graduate school remains small, 
these offsetting forces produce a nearly unchanged $\theta_{3}$ across age 
windows.

\begin{figure}[h!]
\begin{centering}
\includegraphics[width=\linewidth,keepaspectratio]{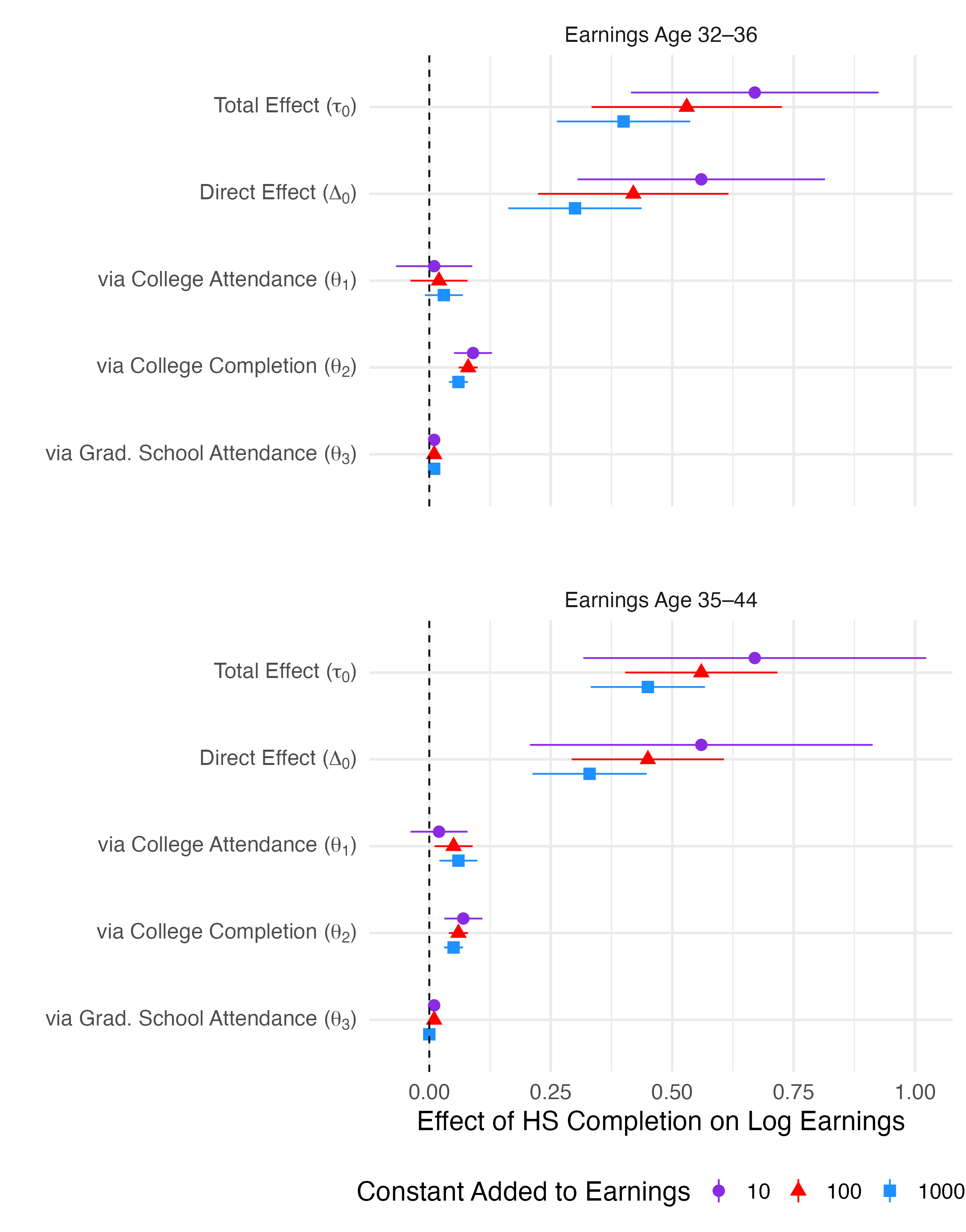}
\par\end{centering}
\caption{Decomposition of the Average Total Effect (ATE) of high school graduation
on logged earnings under alternative definitions of earnings. The figure 
shows estimates of the total effect ($\tau_{0}$) as well as the continuation
effects $\Delta_{0},\tau_{1},\dots,\tau_{K}$ when constants of 10, 100, 
and 1000 are added to annual earnings (in dollars) prior to taking the log.
\protect\label{fig:Decomp-constant-nlsy79}}
\end{figure}

\newpage{}

\begin{figure}[h!]
\begin{centering}
\includegraphics[width=\linewidth,keepaspectratio]{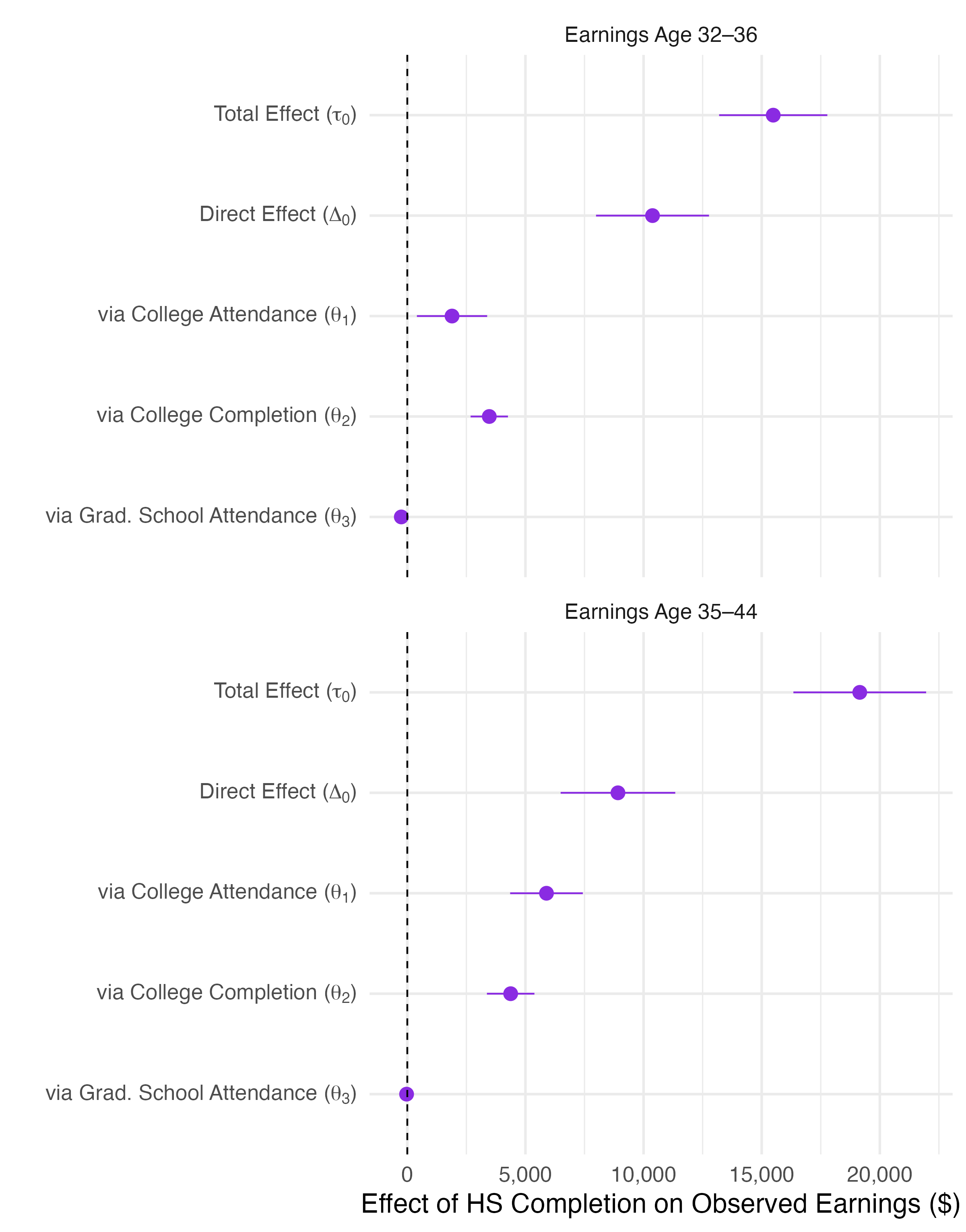}
\par\end{centering}
\caption{Decomposition of the Average Total Effect (ATE) of high school graduation
on observed annual earnings in dollars (no transformation). The figure 
shows estimates of the total effect ($\tau_{0}$) as well as the continuation
effects $\Delta_{0},\tau_{1},\dots,\tau_{K}$ using raw annual earnings.
\protect\label{fig:Decomp-raw-nlsy79}}
\end{figure}

\newpage{}
\begin{table}[h!]
\caption{Direct Effects ($\Delta_{k}$), Probabilities ($\pi_{k}$) and Covariance
Terms ($\eta_{k}$) Involved in Decomposition via Debiased Machine-Learning
(DML) for NLSY79 cohort. \label{tab:Decomp-nlsy79}}
\centering
\begingroup\fontsize{9.2pt}{9.2pt}\selectfont\setlength{\tabcolsep}{3pt}
\renewcommand{\arraystretch}{1.5} 
\begin{tabular}{lllllllllll}
   \hline & $\Delta_0$ & $\Delta_1$ & $\Delta_2$ & $\Delta_3$ & $\pi_1$ & $\pi_2$ & $\pi_3$ & $\eta_1$ & $\eta_2$ & $\eta_3$  \\ \hline Earnings Age 32–36 & 0.300 & 0.070 & 0.520 & 0.150 & 0.480 & 0.330 & 0.120 & -0.010 & -0.040 & 0.020 \\ 
  Earnings Age 35–44 & 0.330 & 0.140 & 0.440 & 0.290 & 0.480 & 0.330 & 0.120 & -0.010 & -0.030 & -0.010 \\ 
   & (0.070) & (0.040) & (0.040) & (0.050) & (0.010) & (0.010) & (0.010) & (0.010) & (0.010) & (0.010) \\ 
   & (0.060) & (0.040) & (0.040) & (0.040) & (0.010) & (0.010) & (0.010) & (0.010) & (0.010) & (0.010) \\ 
   \hline\end{tabular}
\endgroup
\end{table}

\newpage
\FloatBarrier
\section{Extension to categorical mediators} The main text considers a decomposition of the ATE in the case of
binary monotonic mediators (i.e., educational transitions), but the
framework naturally extends to settings with categorical transitions.
Such an extension is especially appealing in the context of U.S. higher
education, where individuals follow increasingly variegated pathways.
In the early 2010s, fewer than $40\%$ of high school graduates enrolled
directly in a four-year college, whereas roughly $30\%$ entered a
two-year college, and close to one third of these later transferred to a
four-year program. Nearly half of all BA recipients nationally had
attended a two-year college at some point in their educational careers.
A similar situation arises at the transition to postgraduate education:
``graduate school'' encompasses multiple substantively distinct routes
(e.g.\ Master’s degrees, professional programs, PhDs), each featuring
different patterns of selection and labor-market returns. Dichotomizing
such transitions discards precisely the heterogeneity that is often of
substantive interest.

To illustrate the extension, consider a single intermediate transition
$M_k$ that now takes values in a finite set of mutually exclusive and
exhaustive categories,
\[
\mathcal{H}_k=\{h_1,\ldots,h_H\},
\]
rather than a binary indicator.  For each category $h\in\mathcal{H}_k$,
let $M_k(h)$ denote the potential mediator value under category $h$, and
let $Y(1_k,h)$ denote the potential outcome under completion of all
prior transitions $1_k=(A=1,M_1=1,\dots,M_{k-1}=1)$ with the mediator $M_k$
set to $h$.

The net effect of transition $k$ can then be written as
\begin{equation}
\tau_k
=
\Delta_k
\;+\;
\sum_{h\in\mathcal{H}_{k+1}}
\bigl(
  \pi_{{k+1},h}\,\Delta_{{k+1},h}
  \;+\;
  \eta_{{k+1},h}
\bigr),
\label{eq:tau-multinomial}
\end{equation}

where

\begin{align*}
\Delta_{k,h}
&\equiv 
\mathbb{E}[
  Y(1_k,h)-Y(1_k,h_0)],
\qquad h_0\in\mathcal{H}_k
\ \text{(baseline)},\\[6pt]
\pi_{k,h}
&\equiv 
\mathbb{E}\!\left[\mathbf{1}\{M_k(1_k)=h\}\right],\\[6pt]
\eta_{k,h}
&\equiv 
\operatorname{cov}[
  \mathbf{1}\{M_k(1_k)=h\},
  \ Y(1_k,h)-Y(1_k,h_0)].
\end{align*}

These terms generalize the binary case:  
$\Delta_{k,h}$ is the direct effect of category $h$ relative to
$h_0$; $\pi_{k,h}$ is the counterfactual probability of attaining $h$
given completion of all prior transitions; and $\eta_{k,h}$ captures the
alignment between the payoff $Y(1_k,h)-Y(1_k,h_0)$ and the propensity to
realize category $h$. The direct effects $\Delta_{k,h}$ and path probabilities $\pi_{k,h}$
remain identified under Assumptions~\ref{assu:Consistency}–
\ref{assu:Positivity}.  However, when $|\mathcal{H}_k|>2$, the covariance
terms are no longer separately nonparametrically identified.
In the binary case,
\[
\tau_k
=
\Delta_k
+
\pi_k \Delta_k
+
\eta_k
\]
contains a single $\eta_k$ and one identifying restriction.
By contrast, \eqref{eq:tau-multinomial} contains
multiple unknowns $\{\eta_{k,h}\}$ but only one restriction, implying
that the vector of covariance components cannot be uniquely recovered
without additional structure—for example, assuming homogeneity across
categories, proportionality, or a parametric model for effect
heterogeneity.

\paragraph{Graduate school example (NLSY97)}
In my empirical application, the $k=3$ transition corresponds to the
education decision following BA completion.  
I disaggregate this transition into three categories:
\[
\mathcal{H}_3=\{\mathrm{BA\!-\!only},\;\mathrm{MA},\;\mathrm{PhD}\},
\]
with \text{BA-only} as the baseline $h_0$. Among BA completers in the NLSY97, roughly $73\%$ do not
pursue any graduate degree, $22\%$ complete a Master’s degree, and only
$5\%$ complete a professional or PhD degree. \footnote{Professional (DDS, MD and JD) and PhD
degree holders are pooled due to small cell sizes.}
Stage-specific causal contrasts are:
\[
\Delta_{3,h}
=
\mathbb{E}\!\bigl[
  Y(1_3,h)-Y(1_3,\mathrm{BA\!-\!only})
\bigr],
\qquad
h\in\{\mathrm{MA},\mathrm{PhD}\}.
\]

The continuation value that enters the full MPSE decomposition is
\begin{equation}
\theta_3
=
\left(\prod_{j=1}^{2}\pi_j\right)
\sum_{h\in\{\mathrm{MA},\mathrm{PhD}\}}
\bigl(
  \pi_{3,h}\,\Delta_{3,h} + \eta_{3,h}
\bigr).
\label{eq:theta3-multinomial}
\end{equation}

Figure~\ref{fig:Decomp-phd} displays the resulting decomposition for MA
and professional/PhD pathways, and Table~\ref{tab:grad-path} reports the estimated
components. The continuation effects for the graduate-school transition are small for both pathways, with
$\widehat{\theta}_{3,\mathrm{MA}} = 0.0119$ and 
$\widehat{\theta}_{3,\mathrm{PhD}} = 0.0077$, corresponding to earnings effects of high school graduation via these transitions of
$1.2\%$ and $0.8\%$, respectively. Although the continuation effect is slightly lower for the
PhD/professional pathway, this masks a much larger underlying causal contrast:
the net effect of completing a PhD or professional postgraduate program  relative to not pursuing postgraduate study is 
$\widehat{\Delta}_{3,\mathrm{PhD}} = 0.521$ (a $68\%$ earnings premium), compared to a more modest
$\widehat{\Delta}_{3,\mathrm{MA}} = 0.213$ (a $24\%$ premium) for Master’s degrees. The similarity in the
overall continuation effects instead reflects stark differences in the counterfactual probabilities of
each pathway: only about $2\%$ of BA completers pursue a PhD-level degree 
($\widehat{\pi}_{3,\mathrm{PhD}} \approx 0.021$), whereas Master's attainment is roughly six times more
common ($\widehat{\pi}_{3,\mathrm{MA}} \approx 0.134$).

\begin{table}[h]
\centering
\caption{Estimated components for MA and PhD pathways in the graduate-school transition.
\label{tab:grad-path}}
\begin{tabular}{lcccccc}
\hline
 & $\Delta_{3,MA}$ & $\Delta_{3,PhD}$ & $\pi_{3,MA}$ & $\pi_{3,PhD}$ & $\eta_{3,MA}$ & $\eta_{3,PhD}$ \\
\hline
DML & 0.213 & 0.521 & 0.134 & 0.021 & 0.022 & 0.022 \\
    & (0.032) & (0.032) & (0.007) & (0.002) & (0.008) & (0.008) \\
\hline
\end{tabular}
\end{table}

\begin{figure}[h!]
\begin{centering}
\includegraphics[scale=0.6]{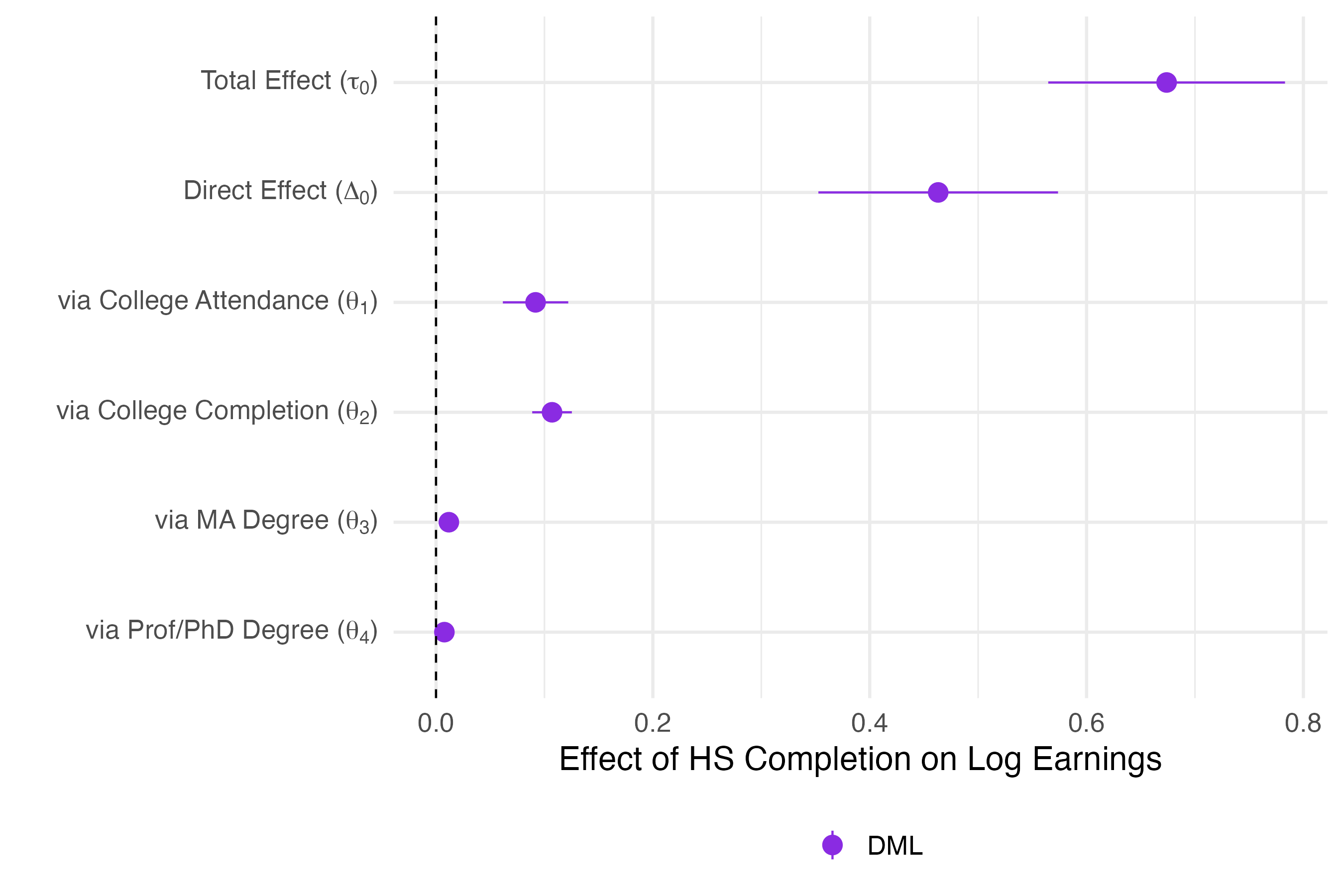}
\par\end{centering}
\caption{Continuation effects for Master's and Professional/PhD pathways in the NLSY97.
\label{fig:Decomp-phd}}
\end{figure}

\FloatBarrier
\newpage
\newpage

\clearpage
\section{Description of EIFs used in empirical illustration\protect\label{App:Illustration-Description}}

For each component involved in the MPSE, I construct a Neyman-orthogonal
``signal'' using its EIF, whose exact form depends on whether each
set of intermediate confounders is empty or not. Figure \ref{fig:DAG-empirical-1}
in the main text shows a potential data-generating process for the
direct and indirect (continuation) effects of high school graduation
on adult earnings, via three transitions: college attendance ($M_{1}$),
BA completion ($M_{2}$), and graduate school attendance ($M_{3}$).
I assume that a set of pre-college characteristics serve as confounders
for the $A-(M_{1},M_{2},M_{3},Y)$ relationships, and that a set of
post-secondary confounders $Z$ confound the $M_{2}-(M_{3},Y)$ relationships.

Under these assumptions for the various sets of confounders, my MPSE
decomposition implies that, in the case of the four transitions (one
treatment and three mediators), it suffices to estimate the following
three sets of parameters: (i) four direct effects $\Delta_{k},k\in\{0,\dots,3\}$,
where $\Delta_{3}=\tau_{3}$, (ii) four gross effects $\tau_{k},k\in\{0,\dots,3\}$,
where $\tau_{0}=\mathrm{ATE}$, (iii) three mediator terms, $\pi_{k},k\in\{1,\dots,3\}$.
All components in the three-mediator decomposition can then be estimated
as functions of these parameters. For each of these target parameters,
I construct a Neyman-orthogonal signal using its efficient influence
function. Because of my assumed data-generating process, which maintains
that there is only a single set of intermediate confounders (as opposed
to a separate set of confounders for each mediator), the EIF for each
estimand involved in the decomposition simplifies somewhat. Specifically,
the recentered EIFs for each component in the decomposition are shown
below:

\begin{align*}
\ensuremath{M_{1}^{*}(1)} & =\gamma_{1}(X)+\frac{\mathbb{I}(A=1)}{\pi_{0}(X,1)}(M_{1}-\gamma_{1}(X)),\\
\ensuremath{M_{2}^{*}(1,1)} & =\gamma_{2}(X)+\frac{\mathbb{I}(A=1)\mathbb{I}(M_{1}=1)}{\pi_{0}(X,1)\pi_{1}(X,1)}(M_{2}-\gamma_{2}(X)),\\
\ensuremath{M_{3}^{*}(1,1,1)} & =\mathbb{E}[\gamma_{3}(X,Z)|X,A=1,M_{1}=1]\\
 & +\frac{\mathbb{I}(A=1)\mathbb{I}(M_{1}=1)}{\pi_{0}(X,1)\pi_{1}(X,1)}(\gamma_{3}(X,Z)-\mathbb{E}[\gamma_{3}(X,Z)|X,A=1,M_{1}=1])\\
 & +\frac{\mathbb{I}(A=1)\mathbb{I}(M_{1}=1)\mathbb{I}(M_{2}=1)}{\pi_{0}(X,1)\pi_{1}(X,1)\pi_{2}(X,Z,1)}(M_{3}-\gamma_{3}(X,Z)),\\
\ensuremath{Y^{*}(a)} & =\mu_{0}(X,a)+\frac{\mathbb{I}(A=a)}{\pi_{0}(X,a)}(Y-\mu_{0}(X,a)),\;\text{for }a\in\{0,1\}\\
\ensuremath{Y^{*}(1,m_{1})} & =\mu_{1}(X,m_{1})+\frac{\mathbb{I}(A=1)\mathbb{I}(M_{1}=m_{1})}{\pi_{0}(X,1)\pi_{1}(X,m_{1})}(Y-\mu_{1}(X,m_{1})),\;\text{for }m_{1}\in\{0,1\}\\
\ensuremath{Y^{*}(1,1,m_{2})} & =\mathbb{E}[\mu_{2}(X,Z,m_{2})|X,A=1,M_{1}=1]\\
 & +\frac{\mathbb{I}(A=1)\mathbb{I}(M_{1}=1)}{\pi_{0}(X,1)\pi_{1}(X,1)}(\mu_{2}(X,Z,m_{2})-\mathbb{E}[\mu_{2}(X,Z,m_{2})|X,A=1,M_{1}=1])\\
 & +\frac{\mathbb{I}(A=1)\mathbb{I}(M_{1}=1)\mathbb{I}(M_{2}=m_{2})}{\pi_{0}(X,1)\pi_{1}(X,1)\pi_{2}(X,Z,m_{2})}(Y-\mu_{2}(X,Z,m_{2})),\;\text{for }m_{2}\in\{0,1\}\\
Y^{*}(1,1,1,m_{3})= & \mathbb{E}[\mu_{3}(X,Z,m_{3})|X,A=1,M_{1}=1]\\
 & +\frac{\mathbb{I}(A=1)\mathbb{I}(M_{1}=1)}{\pi_{0}(X,1)\pi_{1}(X,1)}\left(\mu_{3}(X,Z,m_{3})-\mathbb{E}[\mu_{3}(X,Z,m_{3})|X,A=1,M_{1}=1]\right)\\
 & +\frac{\mathbb{I}(A=1)\mathbb{I}(M_{1}=1)\mathbb{I}(M_{2}=1)\mathbb{I}(M_{3}=m_{3})}{\pi_{0}(X,1)\pi_{1}(X,1)\pi_{2}(X,Z,1)\pi_{3}(X,Z,m_{3})}\left(Y-\mu_{3}(X,Z,m_{3})\right)\;\text{for }m_{3}\in\{0,1\},
\end{align*}

where

\[
\begin{aligned}\pi_{0}(X,a) & \triangleq\mathrm{Pr}[A=a\mid X]\\
\pi_{1}(X,m_{1}) & \triangleq\mathrm{Pr}[M_{1}=m_{1}\mid X,A=1]\\
\pi_{2}(X,Z,m_{2}) & \triangleq\mathrm{Pr}[M_{2}=m_{2}\mid X,A=1,M_{1}=1,Z]\\
\pi_{3}(X,Z,m_{3}) & \triangleq\mathrm{Pr}[M_{3}=m_{3}\mid X,A=1,M_{1}=1,Z,M_{2}=1]\\
\gamma_{1}(X) & \triangleq\mathbb{E}[M_{1}\mid X,A=1]\\
\gamma_{2}(X) & \triangleq\mathbb{E}[M_{2}\mid X,A=1,M_{1}=1]\\
\gamma_{3}(X,Z) & \triangleq\mathbb{E}[M_{3}\mid X,A=1,M_{1}=1,Z,M_{2}=1]\\
\mu_{0}(X,a) & \triangleq\mathbb{E}[Y|X,A=a]\\
\mu_{1}(X,m_{1}) & \triangleq\mathbb{E}[Y|X,A=1,M_{1}=m_{1}]\\
\mu_{2}(X,Z,m_{2}) & \triangleq\mathbb{E}[Y|X,A=1,M_{1}=1,Z,M_{2}=m_{2}]\\
\mu_{3}(X,Z,m_{3}) & \triangleq\mathbb{E}[Y|X,A=1,M_{1}=1,Z,M_{2}=1,M_{3}=m_{3}].
\end{aligned}
\]

\newpage{}

\section{Proofs and technical details\protect\label{App:Derivation-Influence}}

\subsection{EIFs for $\eta_{k}$ and $\theta_{k}$ terms (Theorem \ref{prop:eif-cov-theta})}

Under Assumptions \ref{assu:Consistency}-\ref{assu:Positivity},
the covariance component $\eta_{k}$ is identified as $\eta_{k}=\tau_{k-1}-\Delta_{k-1}-\pi_{k}\tau_{k}$.
Following \citet[p.~15]{Kennedy2022}, I let $\mathbb{IF}:\Psi\rightarrow L_{2}(\mathbb{P})$
denote the operator mapping the functionals $\{\Delta_{k},\pi_{k},\eta_{k}\}:\mathcal{P}\rightarrow\mathbb{R},$
$\forall k\in[K]$ to their respective influence functions under the
nonparametric model $\mathcal{P}$. First, by linearity of the EIF,
$\mathbb{IF}(\eta_{k})$ is given by

\[
\mathbb{IF}(\eta_{k})=\mathbb{IF}(\tau_{k-1})-\mathbb{IF}(\Delta_{k-1})-\mathbb{IF}(\pi_{k}\tau_{k}).
\]

Since $\mathbb{IF}(\pi_{k}\tau_{k})$ can be written as follows $\mathbb{IF}(\pi_{k}\tau_{k})=\tau_{k}\mathbb{IF}(\pi_{k})+\pi_{k}\mathbb{IF}(\tau_{k})$,
$\mathbb{IF}(\eta_{k})$ can be written as

\begin{align*}
\mathbb{IF}(\eta_{k}) & =\mathbb{IF}(\tau_{k-1})-\mathbb{IF}(\Delta_{k-1})-\big(\tau_{k}\mathbb{IF}(\pi_{k})+\pi_{k}\mathbb{IF}(\tau_{k})\big)\\
 & =\mathbb{IF}(\tau_{k-1})-\mathbb{IF}(\Delta_{k-1})-\tau_{k}\mathbb{IF}(\pi_{k})-\pi_{k}\mathbb{IF}(\tau_{k}).
\end{align*}

Noticing that we can rewrite this expression as 

\begin{align*}
\mathbb{IF}(\eta_{k}) & =\mathbb{RIF}(\tau_{k-1})-\tau_{k-1}-\mathbb{RIF}(\Delta_{k-1})+\Delta_{k-1}-\tau_{k}\mathbb{RIF}(\pi_{k})+\tau_{k}\pi_{k}-\pi_{k}\mathbb{RIF}(\tau_{k})+\pi_{k}\tau_{k}\\
 & =\mathbb{RIF}(\tau_{k-1})-\mathbb{RIF}(\Delta_{k-1})-\tau_{k}\mathbb{RIF}(\pi_{k})-\pi_{k}\mathbb{RIF}(\tau_{k})+\pi_{k}\tau_{k}-\eta_{k},
\end{align*}

where $\mathbb{RIF}(\phi)=\mathbb{IF}(\phi)+\phi$, we can obtain
the corresponding EIF-based estimator for $\eta_{k}$ by solving the
empirical moment condition implied by setting the average of the above
equation equal to $0$, and plugging in the set of estimated nuisance
functions:

\[
\hat{\eta}_{k}^{\text{eif}}=\mathbb{\widehat{RIF}}(\tau_{k-1})-\mathbb{\widehat{RIF}}(\Delta_{k-1})-\tau_{k}\mathbb{\mathbb{\widehat{RIF}}}(\pi_{k})-\pi_{k}\mathbb{\mathbb{\widehat{RIF}}}(\tau_{k})+\pi_{k}\tau_{k},
\]

where $\mathbb{\hat{RIF}}(\phi)=\hat{\mathbb{IF}}(\phi)+\phi,$ and
$\hat{\mathbb{IF}}(\phi)$ denotes the influence function of a parameter
evaluated at estimates of its component nuisance functions.

Turning next to the influence functions for the continuation effects
$\theta_{k}$, $k\in\{1,\dots K\}$, following the same logic as the
above, we can write the EIF of $\theta_{k}$, $\mathbb{IF}(\theta_{k})$,
as

\[
\mathbb{IF}(\theta_{k})=\mathbb{IF}(\Delta_{k})\prod_{j=1}^{k}\pi_{j}+\Delta_{k}\sum_{j=1}^{k}\mathbb{IF}(\pi_{j})\prod_{l:l\ne j}^{k}\pi_{l}+\mathbb{IF}(\eta_{k})\prod_{j=1}^{k-1}\pi_{j}+\eta_{k}\sum_{j=1}^{k-1}\mathbb{IF}(\pi_{j})\prod_{l:l\ne j}^{k-1}\pi_{l}.
\]
Rewriting this expression as 

\begin{align*}
\mathbb{IF}(\theta_{k}) & =\mathbb{RIF}(\Delta_{k})\prod_{j=1}^{k}\pi_{j}+\Delta_{k}\sum_{j=1}^{k}\mathbb{RIF}(\pi_{j})\prod_{l:l\ne j}^{k}\pi_{l}+\mathbb{RIF}(\eta_{k})\prod_{j=1}^{k-1}\pi_{j}+\eta_{k}\sum_{j=1}^{k-1}\mathbb{RIF}(\pi_{j})\prod_{l:l\ne j}^{k-1}\pi_{l}\\
 & \;-k\Delta_{k}\prod_{j=1}^{k}\pi_{j}-(k-1)\eta_{k}\prod_{j=1}^{k-1}\pi_{j}-\theta_{k},
\end{align*}

we obtain the corresponding EIF-based estimator for $\theta_{k}$
as:

 \begin{adjustwidth}{-2cm}{-2cm}

\begin{align*}
\hat{\theta}_{k}^{\text{eif}} & =\widehat{\mathbb{RIF}}(\Delta_{k})\prod_{j=1}^{k}\hat{\pi}_{j}+\hat{\Delta}_{k}\sum_{j=1}^{k}\hat{\mathbb{RIF}}(\pi_{j})\prod_{l:l\ne j}^{k}\hat{\pi}_{l}+\widehat{\mathbb{RIF}}(\eta_{k})\prod_{j=1}^{k-1}\hat{\pi}_{j}+\hat{\eta}_{k}\widehat{\mathbb{RIF}}(\pi_{j})\prod_{l:l\ne j}^{k-1}\hat{\pi}_{l}\\
 & \;-k\hat{\Delta}_{k}\prod_{j=1}^{k}\hat{\pi}_{j}-(k-1)\hat{\eta}_{k}\prod_{j=1}^{k-1}\hat{\pi}_{j}.
\end{align*}

 \end{adjustwidth}

\newpage{}

\subsection{Semiparametric efficiency (Theorem \ref{prop:semiparametric-efficiency})}

In this section, I establish the conditions required for the semiparametric
efficiency of all terms featured in the decomposition. Before proceeding,
I establish some notational preliminaries. Let $||g||=(\int g^{\top}gdP)^{1/2}$
denote the $L_{2}(P)$ norm, and let $R_{n}\big(\cdot\big)$ denote
a mapping from a nuisance function to its $L_{2}(P)$ convergence
rate. Let $\hat{\varphi}_{km_{k}}^{\text{EIF}}=\mathbb{P}_{n}[m(O;\hat{\eta})]$,
where $m(O;\hat{\eta})$ is the quantity inside $\mathbb{P}_{n}[\cdot]$
in equation \ref{eq:eif-estimator}, and $\hat{\eta}=(\hat{\pi}_{0},\dots,\hat{\pi}_{K},\hat{\mu}_{0},\dots,\hat{\mu}_{K})$. For the purposes of the technical proofs in this appendix, I adopt a
slightly more compact notation than that used in the main text.
Specifically, I define $M_{0}\equiv A$, and let
$\overline{1}_{k}\equiv(A=1,M_{1}=1,\ldots,M_{k-1}=1)$. We have that 

\begin{align}
\hat{\varphi}_{km_{k}}^{\text{EIF}}-\varphi_{km_{k}} & =\mathbb{P}_{n}[m(O;\hat{\eta})]-P[m(O;\eta)]\nonumber \\
 & =\mathbb{P}_{n}[m(O;\eta)]+\underbrace{P[m(O;\hat{\eta})-m(O;\eta)]}_{\triangleq R_{2}(\hat{\eta})}+\left(\mathbb{P}_{n}-P\right)[m(O;\hat{\eta})-m(O;\eta)],\label{eq:von-mises}
\end{align}
where $Pg=\int gdP$ denotes the expectation of function $g$ at the
truth. The first term in equation \ref{eq:von-mises} is a sample
average, and can be analyzed with the central limit theorem. It has
an asymptotic variance of $\mathbb{E}\left[\left(\varphi_{km_{k}}(O;\eta)\right)^{2}\right].$
The last term is an empirical process term that will be $o_{p}(n^{-1/2})$
if either the nuisance functions fall in a Donsker class or if cross-fitting
is used to induce independence between $\hat{\eta}$ and $O$. Thus,$\hat{\theta}^{\text{EIF}}$
will be asymptotically normal and semiparametric efficient if $R_{2}(\hat{\eta})$
is $o(n^{-1/2})$. To analyze this term, I first note that 

\[
\begin{aligned}P[m(O;\eta)] & =P\bigg[\frac{A}{\pi_{01}\big(X\big)}\left(\frac{\mathbb{I}(M_{k}=m_{k})}{\pi_{km_{k}}\big(X,\overline{Z}_{k}\big)}\prod_{j=1}^{k-1}\frac{M_{j}}{\pi_{j1}\big(X,\overline{Z}_{j}\big)}\right)\left(Y-\mu_{km_{k}}^{k}\left(X,\bar{Z}_{k}\right)\right)\\
 & \quad+\sum_{j=1}^{k}\frac{A}{\pi_{01}\big(X\big)}\left(\prod_{l=1}^{j-1}\frac{M_{l}}{\pi_{l1}\big(X,\overline{Z}_{l}\big)}\right)\left(\mu_{jm_{k}}^{k}\left(X,\overline{Z}_{k}\right)-\mu_{j-1m_{k}}^{k}\left(X,\bar{Z}_{j-1}\right)\right)+\mu_{0}(X)\bigg]\\
 & =P\bigg[\frac{A}{\pi_{01}\big(X\big)}\left(\frac{\mathbb{I}(M_{k}=m_{k})}{\pi_{km_{k}}\big(X,\overline{Z}_{k}\big)}\prod_{j=1}^{k-1}\frac{M_{j}}{\pi_{j1}\big(X,\overline{Z}_{j}\big)}\right)\cdot\\
 & \quad\left(\underbrace{\mathbb{E}[Y-\mu_{km_{k}}^{k}\left(X,\bar{Z}_{k}\right)|X,\bar{Z}_{k},\overline{1}_{k},m_{k}]}_{=0}\right)\\
 & \quad+\sum_{j=1}^{k}\frac{A}{\pi_{01}\big(X\big)}\left(\prod_{l=1}^{j-1}\frac{M_{l}}{\pi_{l1}\big(X,\overline{Z}_{l}\big)}\right)\cdot\\
 & \quad\left(\underbrace{\mathbb{E}[\mu_{jm_{k}}^{k}\left(X,\overline{Z}_{j}\right)-\mu_{(j-1)m_{j}}\left(X,\bar{Z}_{j-1}\right)|X,\bar{Z}_{j-1},\overline{1}_{j}]}_{=0}\right)+\mu_{0}(X)\bigg]\\
 & =P\bigg[\mu_{0}(X)\bigg].
\end{aligned}
\]

Plugging this result into $R_{2}(\hat{\eta})$, we have that 

\[
\begin{aligned}R_{2}(\hat{\eta}) & =P[m(O;\hat{\eta})-m(O;\eta)]\\
 & =P\bigg[\sum_{j=0}^{k-1}\frac{A}{\hat{\pi}_{01}(X)}\left(\prod_{l=1}^{j-1}\frac{M_{j}}{\hat{\pi}_{l1}(X,\bar{Z}_{l})}\right)\cdot\\
 & \quad\big(\hat{\pi}_{j1}\big(X,\overline{Z}_{j}\big)-\pi_{j1}\big(X,\overline{Z}_{j}\big)\big)\big(\hat{\mu}_{jm_{k}}^{k}\big(X,\bar{Z}_{j}\big)-\mu_{jm_{k}}^{k}\big(X,\bar{Z}_{j}\big)\big)\\
 & \quad+\frac{A}{\hat{\pi}_{01}(X)}\left(\prod_{l=1}^{k-1}\frac{M_{j}}{\hat{\pi}_{l1}(X,\bar{Z}_{l})}\right)\cdot\\
 & \quad\big(\hat{\pi}_{km_{k}}\big(X,\overline{Z}_{j}\big)-\pi_{km_{k}}\big(X,\overline{Z}_{j}\big)\big)\big(\hat{\mu}_{km_{k}}^{k}\big(X,\bar{Z}_{k}\big)-\mu_{km_{k}}^{k}\big(X,\bar{Z}_{k}\big)\big)\bigg]\\
 & =\sum_{j=0}^{k}O_{p}\big(||\hat{\pi}_{j1}\big(X,\overline{Z}_{j}\big)-\pi_{j1}\big(X,\overline{Z}_{j}\big)||\cdot||\hat{\mu}_{jm_{k}}^{k}\big(X,\bar{Z}_{k}\big)-\mu_{jm_{k}}^{k}\big(X,\bar{Z}_{k}\big)||\big),
\end{aligned}
\]

where the last equality results from the positivity assumption that
$\hat{\pi}_{k1}(X,\bar{Z}_{k})$ is bounded away from zero, for all
$k\in[K]$, and from the Cauchy-Schwartz inequality. Then, assuming
that the empirical process term is of order $o_{p}(n^{-1/2})$, we
can write $\hat{\varphi}_{km_{k}}^{\text{EIF}}-\varphi_{km_{k}}$
as

\begin{align*}
\hat{\psi}_{km_{k}}^{\text{EIF}}-\psi_{km_{k}} & =\mathbb{P}_{n}[m(O;\eta)-\psi_{km_{k}}]\\
 & \quad+\sum_{j=0}^{k}O_{p}\big(||\hat{\pi}_{j1}\big(X,\overline{Z}_{j}\big)-\pi_{j1}\big(X,\overline{Z}_{j}\big)||\big)\cdot O_{p}\big(||\hat{\mu}_{jm_{k}}^{k}\big(X,\bar{Z}_{k}\big)-\mu_{jm_{k}}^{k}\big(X,\bar{Z}_{j}\big)||\big)\\
 & \quad+o_{p}(n^{-1/2}).
\end{align*}

Thus, letting $R_{n}(k,m_{k})\triangleq\sum_{j=0}^{k}R_{n}\big(\hat{\pi}_{j})R_{n}\big(\hat{\mu}_{km_{k}}^{k})$,
$\hat{\varphi}_{km_{k}}^{\text{rEIF}}$ is consistent if $R_{n}(k,m_{k})=o(1)$
and it is semiparametric efficient if $R_{n}(k,m_{k})=o(n^{-1/2})$.
Clearly, then, $\hat{\Delta}_{k}^{\text{rEIF }}$ is consistent if
$\sum_{j=k}^{k+1}R_{n}(j,0)=o(1)$ and it is semiparametric efficient
if $\sum_{j=k}^{k+1}R_{n}(j,0)=o(n^{-1/2})$. Similarly, $\hat{\tau}_{k}^{\text{rEIF}}$
is consistent if $\sum_{j=0}^{1}R_{n}(k,j)=o(1)$ and it is semiparametric
efficient if $\sum_{j=0}^{1}R_{n}(k,j)=o(n^{-1/2})$. 

Turning next to $\phi_{k}\triangleq\mathbb{E}[M_{k+1}(\overline{1}_{k+1})]$,
for all $k\in[K-1]$, we can similarly define $\ensuremath{\gamma_{k}\left(X,\bar{Z}_{k}\right)}$
iteratively as

\begin{align*}
\gamma_{k}\left(X,\bar{Z}_{k}\right) & \triangleq\mathbb{E}\left[M_{k+1}\mid X,\bar{Z}_{k},\overline{1}_{k+1}\right]\\
\gamma_{j}\left(X,\bar{Z}_{j}\right) & \triangleq\mathbb{E}\left[\gamma_{j+1}\left(X,\bar{Z}_{j+1}\right)\mid X,\bar{Z}_{j},\overline{1}_{j+1}\right]\forall j\in[k-1].
\end{align*}

Similarly to the previous case, the EIF of $\phi_{k}$ is equal to

\begin{align*}
\varphi_{k}(O) & =\frac{A}{\pi_{01}\big(X\big)}\left(\prod_{l=1}^{k}\frac{M_{l}}{\pi_{l1}\big(X,\overline{Z}_{l}\big)}\right)\left(Y-\gamma_{k}\left(X,\bar{Z}_{k}\right)\right).\\
 & \quad+\sum_{j=0}^{k}\frac{A}{\pi_{01}\big(X\big)}\left(\prod_{l=1}^{j-1}\frac{M_{l}}{\pi_{l1}\big(X,\overline{Z}_{l}\big)}\right)\left(\gamma_{j}\left(X,\overline{Z}_{j}\right)-\gamma_{j-1}\left(X,\bar{Z}_{j-1}\right)\right)\\
 & \quad+\gamma_{0}(X)-\phi_{k}.
\end{align*}

Following similar arguments to the above, we have that

\begin{align*}
\hat{\phi}_{k}^{\text{EIF}}-\phi_{k} & =\mathbb{P}_{n}[m_{2}(O;\eta)-\phi_{k}]\\
 & \quad+\sum_{j=0}^{k}O_{p}\big(||\hat{\pi}_{j1}\big(X,\overline{Z}_{j}\big)-\pi_{j1}\big(X,\overline{Z}_{j}\big)||\big)\cdot O_{p}\big(||\hat{\gamma}_{j}\big(X,\bar{Z}_{j}\big)-\gamma_{j}\big(X,\bar{Z}_{j}\big)||\big)\\
 & \quad+o_{p}(n^{-1/2}),
\end{align*}

where $m_{2}(O;\hat{\eta})=\varphi_{k}+\phi_{k}$. Thus, letting
$R_{n}(k,\gamma)\triangleq\sum_{j=0}^{k}R_{n}\big(\hat{\pi}_{j})R_{n}\big(\hat{\gamma}_{j})$,
$\hat{\phi}_{k}^{\text{EIF}}$ is consistent if $R_{n}(k,\gamma)=o(1)$
and it is semiparametric efficient if $R_{n}(k,\gamma)=o(n^{-1/2})$.
This result implies that if all nuisance functions are consistently
estimated and converge at faster than $n^{1/4}$ rates, then $\hat{\phi}_{k}^{\text{EIF}}$
is semiparametric efficient. I first establish the following lemma:

\begin{lem}
\label{lem:1}Let $X_{n}$ and $Y_{n}$ denote two convergent sequences,
where $X_{n}=O_{p}(n^{-1/2})$ and $Y_{n}=o_{p}(n^{-1/2})$. Then,
(a) $X_{n}Y_{n}=o_{p}(n^{-1/2})$, and (b) $X_{n}X_{n}=o_{p}(n^{-1/2})$.
\end{lem}
\begin{proof}
(a) $X_{n}=O_{p}(n^{-1/2})=n^{-1/2}O_{p}(1)=o_{p}(1)$. Thus, $X_{n}Y_{n}=o_{p}(1)o_{p}(n^{-1/2})=o_{p}(n^{-1/2})$.
(b) $X_{n}X_{n}=O_{p}(n^{-1/2})O_{p}(n^{-1/2})=O_{p}(n^{-1})=n^{-1/2}O_{p}(n^{-1/2})=n^{-1/2}o_{p}(1)$
(by (a)). Thus, $X_{n}X_{n}=o_{p}(n^{-1/2})$.
\end{proof}
Using this lemma, I establish rate conditions for the semiparametric
efficiency of $\eta_{k}=\tau_{k-1}-\Delta_{k-1}-\pi_{k}\tau_{k}.$
We can analyze the asymptotic behavior of $\hat{\eta}_{k}=\hat{\tau}_{k-1}-\hat{\Delta}_{k-1}-\hat{\pi}_{k}\hat{\tau}_{k}$
via a distributional expansion of each plug-in estimator:

\begin{align*}
\hat{\eta}_{k} & =\hat{\tau}_{k-1}-\hat{\Delta}_{k-1}-\hat{\pi}_{k}\hat{\tau}_{k}\\
 & =(\tau_{k-1}+\mathbb{P}_{n}[\tau_{k-1}^{\text{EIF}}]+\tau_{k-1}^{\text{EP}}+\tau_{k-1}^{\text{R2}}]-(\Delta_{k-1}+\mathbb{P}_{n}[\Delta_{k-1}^{\text{EIF}}]+\Delta_{k-1}^{\text{EP}}+\Delta_{k-1}^{\text{R2}}]\\
 & \quad-[\pi_{k}+\mathbb{P}_{n}[\pi_{k}^{\text{EIF}}]+\pi_{k}^{\text{EP}}+\pi_{k}^{\text{R2}}][\tau_{k}+\mathbb{P}_{n}[\tau_{k}^{\text{EIF}}]+\tau_{k}^{\text{EP}}+\tau_{k}^{\text{R2}}]\\
 & =(\tau_{k-1}+\mathbb{P}_{n}[\tau_{k-1}^{\text{EIF}}]+o_{p}(n^{-1/2})+\tau_{k-1}^{\text{R2}}]-(\Delta_{k-1}+\mathbb{P}_{n}[\Delta_{k-1}^{\text{EIF}}]+o_{p}(n^{-1/2})+\Delta_{k-1}^{\text{R2}}]\\
 & \quad-[\pi_{k}+\mathbb{P}_{n}[\pi_{k}^{\text{EIF}}]+o_{p}(n^{-1/2})+\pi_{k}^{\text{R2}}][\tau_{k}+\tau_{k}^{\text{EIF}}+o_{p}(n^{-1/2})+\tau_{k}^{\text{R2}}]\\
 & =[\tau_{k-1}-\Delta_{k-1}-\pi_{k}\tau_{k}]+\mathbb{P}_{n}[(\tau_{k-1}-\Delta_{k-1}-\pi_{k}\tau_{k})^{\text{EIF}}]\\
 & \quad+\tau_{k-1}^{\text{R2}}+\Delta_{k-1}^{\text{R2}}+\pi_{k}^{\text{R2}}+\tau_{k}^{\text{R2}}+O_{p}(n^{-1/2})O_{p}(n^{-1/2})+O_{p}(n^{-1/2})o_{p}(n^{-1/2})+o_{p}(n^{-1})+o_{p}(n^{-1/2})\\
 & =[\tau_{k-1}-\Delta_{k-1}-\pi_{k}\tau_{k}]+\mathbb{P}_{n}[(\tau_{k-1}-\Delta_{k-1}-\pi_{k}\tau_{k})^{\text{EIF}}]\\
 & \quad+\tau_{k-1}^{\text{R2}}+\Delta_{k-1}^{\text{R2}}+\pi_{k}^{\text{R2}}+\tau_{k}^{\text{R2}}+o_{p}(n^{-1/2}),
\end{align*}

where the penultimate equality follows from Theorem \ref{prop:eif-cov-theta},
and the final equality follows from Lemma \ref{lem:1}.

Thus, for any $k\in\{1,\dots,K\}$, $\hat{\eta}_{k}=\hat{\tau}_{k-1}-\hat{\Delta}_{k-1}-\hat{\pi}_{k}\hat{\tau}_{k}$
is semiparametric efficient if $\tau_{k-1}^{\text{R2}}+\Delta_{k-1}^{\text{R2}}+\pi_{k}^{\text{R2}}+\tau_{k}^{\text{R2}}=o_{p}(n^{-1/2})$. 

For the continuation terms ($\theta_{k}=(\Pi_{j=1}^{k}\pi_{j})\Delta_{k}+(\Pi_{j=1}^{k-1}\pi_{j})\eta_{k}$),
I proceed by induction. Let $\hat{\Delta}_{*}=\Delta_{*}+\mathbb{P}_{n}[\Delta_{*}^{\text{EIF}}]+\Delta_{*}^{\text{EP}}+\Delta_{*}^{\text{R2}}$
and $\hat{\eta}_{*}=\eta_{*}+\mathbb{P}_{n}[\eta_{*}^{\text{EIF}}]+\eta_{*}^{\text{EP}}+\eta_{*}^{\text{R2}}$
be asymptotically linear, where $*\in\{1,\dots K\}$. For $k=1$,
we can asymptotically expand $\hat{\pi}_{1}\hat{\Delta}_{^{*}}$ as

\begin{align*}
\hat{\pi}_{1}\hat{\Delta}_{*} & =(\pi_{1}+\mathbb{P}_{n}[\pi_{1}^{\mathrm{EIF}}]+\pi_{1}^{\mathrm{EP}}+\pi_{1}^{\mathrm{R}2})(\Delta_{*}+\mathbb{P}_{n}[\Delta_{*}^{\text{EIF}}]+\Delta_{*}^{\text{EP}}+\Delta_{*}^{\text{R2}})\\
 & =(\pi_{1}+\mathbb{P}_{n}[\pi_{1}^{\mathrm{EIF}}]+o_{p}(n^{-1/2})+\pi_{1}^{\mathrm{R}2})(\Delta_{*}+\mathbb{P}_{n}[\Delta_{*}^{\mathrm{EIF}}]+o_{p}(n^{-1/2})+\Delta_{*}^{\mathrm{R}2})\\
 & =\pi_{1}\Delta_{*}+\mathbb{P}_{n}[\pi_{1}^{\mathrm{EIF}}\Delta_{*}+\Delta_{*}^{\mathrm{EIF}}\pi_{1}]+\pi_{1}^{\mathrm{R}2}+\Delta_{*}^{\mathrm{R}2}+o_{p}(n^{-1/2})\quad\text{(by Lemma \ref{lem:1}})\\
 & =\sum_{j=1}^{k}\pi_{j}\Delta_{*}+\mathbb{P}_{n}[(\sum_{j=1}^{k}\pi_{j}\Delta_{*})^{\text{EIF}}]+\Delta_{*}^{\mathrm{R}2}+\sum_{j=1}^{k}\pi_{j}^{\mathrm{R}2}+o_{p}(n^{-1/2})\quad\text{(by Lemma }\ref{lem:1}),
\end{align*}

and expand $(\Pi_{j=1}^{k-1}\pi_{j})\eta_{*}$, similarly, as

\begin{align*}
\sum_{j=1}^{k-1}\hat{\pi}_{j}\hat{\eta}_{*} & =\eta_{*}+\mathbb{P}_{n}[\eta_{*}{}^{\text{EIF}}]+\tau_{*-1}^{\text{R2}}+\Delta_{*-1}^{\text{R2}}+\pi_{*}^{\text{R2}}+\tau_{*}^{\text{R2}}+o_{p}(n^{-1/2})\\
 & =\sum_{j=1}^{k-1}\pi_{j}\eta_{*}+\mathbb{P}_{n}[(\sum_{j=1}^{k-1}\pi_{j}\eta_{*})^{\text{EIF}}]\\
 & \quad+\tau_{*-1}^{\text{R2}}+\Delta_{*-1}^{\text{R2}}+\pi_{*}^{\text{R2}}+\tau_{*}^{\text{R2}}+\sum_{j\in\{1\dots,k^{*}-1\}:j\neq*}\pi_{j}^{\text{R2}}+o_{p}(n^{-1/2}),
\end{align*}

following a similar logic to above. Now, assume that, for $k^{*}\in\{1,\dots,K\}$,
$(\Pi_{j=1}^{k^{*}}\hat{\pi}_{j})\hat{\Delta}_{*}=\Delta_{k}\prod_{j=1}^{k^{*}}\pi_{j}+\mathbb{P}_{n}[(\Delta_{*}\prod_{j=1}^{k^{*}}\pi_{j})^{\text{EIF}}]+\Delta_{*}^{\mathrm{R}2}+\sum_{j=1}^{k^{*}}\pi_{j}^{\mathrm{R}2}+o_{p}(n^{-1/2})$
and, further, that $(\Pi_{j=1}^{k^{*}-1}\hat{\pi}_{j})\hat{\eta}_{*}=\Pi_{j=1}^{k^{*}-1}\pi_{j}\eta_{*}+\mathbb{P}_{n}[(\Pi_{j=1}^{k^{*}-1}\pi_{j}\eta_{*})^{\text{EIF}}]+\tau_{*-1}^{\text{R2}}+\Delta_{*-1}^{\text{R2}}+\pi_{*}^{\text{R2}}+\tau_{*}^{\text{R2}}+\sum_{j\in\{1\dots,k^{*}-1\}:j\neq*}\pi_{j}^{\text{R2}}+o_{p}(n^{-1/2})$.
Then, by induction, we have that

\begin{align*}
(\Pi_{j=1}^{k^{*}+1}\hat{\pi})\hat{\Delta}_{*} & =\bigg[(\pi_{k^{*}+1}+\mathbb{P}_{n}[\pi_{k^{*}+1}^{\mathrm{EIF}}]+o_{p}(n^{-1/2})+\pi_{k^{*}+1}^{\mathrm{R}2})\bigg]\\
 & \quad\bigg[\Delta_{*}\prod_{j=1}^{k^{*}}\pi_{j}+\mathbb{P}_{n}[(\Delta_{*}\sum_{j=1}^{k^{*}}\pi_{j})^{\text{EIF}}]+\Delta_{*}^{\mathrm{R}2}+\sum_{j=1}^{k^{*}}\pi_{j}^{\mathrm{R}2}+o_{p}(n^{-1/2})\bigg]\\
 & =\Delta_{*}\prod_{j=1}^{k^{*}+1}\pi_{j}+\mathbb{P}_{n}[(\Delta_{*}\prod_{j=1}^{k^{*}+1}\pi_{j})^{\text{EIF}}]+\Delta_{*}^{\mathrm{R}2}\\
 & \quad+\sum_{j=1}^{k^{*}+1}\pi_{j}^{\mathrm{R}2}+O_{p}(n^{-1/2})O_{p}(n^{-1/2})\\
 & \quad+O_{p}(n^{-1/2})o_{p}(n^{-1/2})+o_{p}(n^{-1})+o_{p}(n^{-1/2})\\
 & =\Delta_{*}\sum_{j=1}^{k^{*}+1}\pi_{j}+\mathbb{P}_{n}[(\Delta_{*}\prod_{j=1}^{k^{*}+1}\pi_{j})^{\text{EIF}}]+\Delta_{*}^{\mathrm{R}2}+\sum_{j=1}^{k^{*}+1}\pi_{j}^{\mathrm{R}2}+o_{p}(n^{-1/2}),
\end{align*}

and that
\begin{align*}
(\Pi_{j=1}^{k^{*}}\hat{\pi}_{j})\hat{\Delta}_{*} & =\bigg[(\pi_{k^{*}+1}+\mathbb{P}_{n}[\pi_{k^{*}+1}^{\mathrm{EIF}}]+o_{p}(n^{-1/2})+\pi_{k^{*}+1}^{\mathrm{R}2})\bigg]\\
 & \quad\bigg[\Pi_{j=1}^{k^{*}-1}\pi_{j}\eta_{*}+\mathbb{P}_{n}[(\Pi_{j=1}^{k^{*}-1}\pi_{j}\eta_{*})^{\text{EIF}}]+\tau_{*-1}^{\text{R2}}+\Delta_{*-1}^{\text{R2}}+\pi_{*}^{\text{R2}}+\tau_{*}^{\text{R2}}+\sum_{j\in\{1\dots,k^{*}-1\}:j\neq*}\pi_{j}+o_{p}(n^{-1/2})\bigg]\\
 & =\Pi_{j=1}^{k^{*}}\pi_{j}\eta_{*}+\mathbb{P}_{n}[(\Pi_{j=1}^{k^{*}}\pi_{j}\eta_{*})^{\text{EIF}}]+\tau_{*-1}^{\text{R2}}+\Delta_{*-1}^{\text{R2}}+\pi_{*}^{\text{R2}}+\tau_{*}^{\text{R2}}+\sum_{j\in\{1\dots,k^{*}\}:j\neq*}\pi_{j}^{\text{R2}}+o_{p}(n^{-1/2})
\end{align*}

It follows that, for any $k\in\{1,\dots,K\}$,

\[
(\Pi_{j=1}^{k}\hat{\pi}_{j})\hat{\Delta}_{k}=\Delta_{k}\sum_{j=1}^{k}\pi_{j}+\Pi_{j=1}^{k-1}\pi_{j}\eta_{k}+\mathbb{P}_{n}[(\Delta_{k}\sum_{j=1}^{k}\pi_{j}+\Pi_{j=1}^{k-1}\pi_{j}\eta_{k})^{\text{EIF}}]+\sum_{j=k-1}^{k}\Delta_{j}^{\mathrm{R}2}+\sum_{j=k-1}^{k}\tau_{j}^{\mathrm{R}2}+\sum_{j=1}^{k}\pi_{j}^{\mathrm{R}2}.
\]

Thus, $\hat{\theta}_{k}=(\Pi_{j=1}^{k}\hat{\pi}_{j})\hat{\Delta}_{k}+(\Pi_{j=1}^{k-1}\hat{\pi}_{j})\hat{\eta}_{k}$
is semiparametric efficient if $\sum_{j=k-1}^{k}\Delta_{j}^{\mathrm{R}2}+\sum_{j=k-1}^{k}\tau_{j}^{\mathrm{R}2}+\sum_{j=1}^{k}\pi_{j}^{\mathrm{R}2}=o(n^{-1/2})$.
Theorem 3.2 then follows immediately by recognizing the rate conditions
required for each of the constituent functionals of $(\pi_{k},\Delta_{k},\tau_{k})$
to be semiparametric efficient.

\newpage{}

\section{Derivation of RWR procedures}

For simplicity, throughout the following I let $Z_{0}=X$ and $M_{0}=A$.
I assume the following linear specification of the outcome model:

\begin{equation}
\begin{aligned}\mathbb{E}[Y\mid\overline{Z}_{k},A,\overline{M}_{k}] & =\beta_{k,0}+c_{k,0}A+\sum_{j=1}^{k}\beta_{k,j}M_{j}+\eta_{k,1}^{\top}X^{\perp}+c_{k,1}AX^{\perp}+\sum_{j=1}^{k-1}\eta_{k,j}^{T}M_{j}X^{\perp}+\sum_{j=1}^{k}\gamma_{k,j}^{T}Z_{j}^{\perp}\\
 & \quad+\sum_{j=1}^{k-1}M_{j}\sum_{l=1}^{j}\xi_{k,k,l}^{\top}Z_{l}^{\perp},
\end{aligned}
\label{eq:rwr-outcome-proof}
\end{equation}

where $Z_{k}^{\perp}=Z_{k}-\mathbb{E}[Z_{k}|\overline{Z}_{k-1},M_{k-1}=1_{k-1}]$,
$\forall k\in[0,\dots,K]$. In the following derivations, I use the
fact that, $\forall k\in\{1,\dots,K\},$

\begin{align*}
\int z_{k}^{\perp}dP(z_{k}|\overline{z}_{k-1},m_{k-1}=1)\\
=\mathbb{E}[Z_{k}-\mathbb{E}[Z_{k}|\overline{z}_{k-1},m_{k-1}=1]|\overline{z}_{k-1},m_{k-1}=1]\\
=0.
\end{align*}

Letting $X=Z_{0}$, the above also implies that $\int z_{0}^{\perp}dP(z_{0})=\mathbb{E}[Z_{0}-\mathbb{E}[Z_{0}]]=0$.
Under sequential ignorability and assuming linearity of the outcome
with respect to all antecedent variables, we have that

\begin{align*}
\Delta_{k-1} & =\int\mathbb{E}[Y|\overline{M}_{k-1}=\overline{1}_{k-1},\overline{z}_{k},M_{k}=0]\prod_{j=0}^{k}dP(z_{j}|\overline{z}_{j-1},m_{j-1}=1)\\
 & -\int\mathbb{E}[Y|\overline{M}_{k-2}=\overline{1}_{k-2},\overline{z}_{k},M_{k-1}=0]\prod_{j=1}^{k}dP(z_{j}|\overline{z}_{j-1},m_{j-1}=1)\\
 & =\int\bigg[\beta_{k,0}+c_{k,0}+\sum_{j=1}^{k-1}\beta_{k,j}+\eta_{k,1}^{\top}X^{\perp}+c_{k,1}X^{\perp}+\sum_{j=1}^{k-2}\eta_{k,j}^{T}X^{\perp}+\sum_{j=1}^{k}\gamma_{k,j}^{T}Z_{j}^{\perp}+\sum_{j=1}^{k-2}\sum_{l=1}^{j}\xi_{k,k,l}^{\top}Z_{l}^{\perp})\\
 & -(\beta_{k,0}+c_{k,0}+\sum_{j=1}^{k-2}\beta_{k,j}M_{j}+\eta_{k,1}^{\top}X^{\perp}+c_{k,1}AX^{\perp}+\sum_{j=1}^{k-2}\eta_{k,j}^{T}X^{\perp}+\sum_{j=1}^{k}\gamma_{k,j}^{T}Z_{j}^{\perp}+\sum_{j=1}^{k-3}\sum_{l=1}^{j}\xi_{k,k,l}^{\top}Z_{l}^{\perp})\bigg]\\
 & \prod_{j=0}^{k}dP(z_{j}|\overline{z}_{j-1},m_{j-1}=1)\\
 & =\beta_{k,k-1}.
\end{align*}

Further, for $\tau_{k}\forall k\in\{1,\dots,K\}$ we have that

\begin{align*}
\tau_{k} & =\int\mathbb{E}[Y|A=1,\overline{M}_{k}=\overline{1}_{k},\overline{z}_{k}]\prod_{j=0}^{k}dP(z_{j}|\overline{z}_{j-1},m_{j-1}=1)\\
 & -\int\mathbb{E}[Y|A=1,\overline{M}_{k-1}=\overline{1}_{k-1},\overline{z}_{k},M_{k}=0]\prod_{j=0}^{k}dP(z_{j}|\overline{z}_{j-1},m_{j-1}=1)\\
 & =\int\bigg[(\beta_{k,0}+c_{k,0}+\sum_{j=1}^{k}\beta_{k,j}+\eta_{k,1}^{\top}X^{\perp}+c_{k,1}X^{\perp}+\sum_{j=1}^{k-1}\eta_{k,j}^{T}X^{\perp}+\sum_{j=1}^{k}\gamma_{k,j}^{T}Z_{j}^{\perp}+\sum_{j=1}^{k-1}\sum_{l=1}^{j}\xi_{k,j,l}^{\top}Z_{l}^{\perp})\\
 & -(\beta_{k,0}+c_{k,0}+\sum_{j=1}^{k-1}\beta_{k,j}+\eta_{k,1}^{\top}X^{\perp}+c_{k,1}X^{\perp}+\sum_{j=1}^{k-2}\eta_{k,j}^{T}X^{\perp}+\sum_{j=1}^{k}\gamma_{k,j}^{T}Z_{j}^{\perp}+\sum_{j=1}^{k-2}\sum_{l=1}^{j}\xi_{k,j,l}^{\top}Z_{l}^{\perp})\bigg]\\
 & \prod_{j=0}^{k}dP(z_{j}|\overline{z}_{j-1},m_{j-1}=1)\\
 & =\beta_{k,k}.
\end{align*}
Finally, I assume that

\[
\begin{aligned}\mathbb{E}[M_{k+1}\mid A=1,\overline{Z}_{k},\overline{M}_{k}=\overline{1}_{k}] & =\theta_{0}+\sum_{k=0}^{k}\delta_{k+1}^{T}Z_{k}^{\perp}\end{aligned}
.
\]

Then:

\begin{align*}
\mathbb{E}[M(\overline{1}_{k+1})] & =\theta_{0}+\int\bigg[\sum_{k=0}^{k}\delta_{k+1}^{T}Z_{k}^{\perp}\prod_{j=0}^{k}dP(z_{j}|\overline{z}_{j-1},m_{j-1}=1)\bigg]\\
 & =\theta_{0}.
\end{align*}

\newpage{}

\section{Derivation of bias formulae for sensitivity analysis}

In this section, I derive the bias formulae for the set ($\tau_{k},\Delta_{k}$)
for all $k\in[K]$, where $K$ denotes the number of mediators considered
in the decomposition, under a sequence of simplifying assumptions.
Assume first that we have a binary unobserved confounder, $U$, for
the treatment-outcome relationship. Assuming that $\alpha_{0}=\mathbb{E}[Y|x,a,U=1]-\mathbb{E}[Y|x,a,U=0]$
does not depend on $x$ or $a$, and further that $\beta_{0}=\text{Pr}[U=1|x,A=1]-\text{Pr}[U=1|x,A=0]$
does not depend on $x$, for $\tau_{0}=\mathbb{E}[Y(1)-Y(0)]\triangleq\text{ATE}$,
I then have that $\text{bias}(\tau_{0})=\alpha\beta$ \citep{VanderWeele2011}.

Next, consider an unobserved confounder, $U_{k}$ that affects both
$M_{k}$ and $Y$ for any $k\in\{1,\dots,K\}$. Then, under a weaker
iteration of Assumption \ref{assu:SI} (Sequential Ignorability),
i.e.,

$Y(\overline{1}_{k},m_{k})\perp\!\!\!\perp(A,\overline{M}_{k})|X,A,U_{k},\overline{Z}_{k},\overline{M}_{k-1}\forall k\in[K],$$\mathbb{E}[Y(\overline{1}_{k},m_{k})]$
is identified as

\[
\mathbb{E}[Y(\overline{1}_{k},m_{k})]=\int_{x}\int_{\overline{z}_{k}}\mathbb{E}[Y|x,\overline{z}_{k},\overline{1}_{k},m_{k},u_{k}]\big[dP(u_{k}|x,\overline{z}_{k},\overline{1}_{k-1})\prod_{j=1}^{k}dP(z_{j}|x,\overline{z}_{j-1},\overline{1}_{j-1})\big]dP(x).
\]

By contrast, under Assumption \ref{assu:SI}, my estimator of $\mathbb{E}[Y(\overline{1}_{k},m_{k})]$,
$\tilde{\mathbb{E}}[Y(\overline{1}_{k},m_{k})]$, converges to 

\[
\tilde{\mathbb{E}}[Y(\overline{1}_{k},m_{k})]=\int_{x}\int_{\overline{z}_{k}}\mathbb{E}[Y|x,\overline{z}_{k},\overline{1}_{k},m_{k},u_{k}]\big[dP(u_{k}|x,\overline{z}_{k},\overline{1}_{k})\prod_{j=1}^{k}dP(z_{j}|x,\overline{z}_{j-1},\overline{1}_{j-1})\big]dP(x).
\]
I invoke the following three assumptions: (Assumption $A_{k}$) $\alpha_{k}=\mathbb{E}[Y|x,\overline{z}_{k},\overline{1}_{k},m_{k},U_{k}=1]-\mathbb{E}[Y|x,\overline{z}_{k},\overline{1}_{k},m_{k},U_{k}=0]$
does not depend on ($x,\overline{z}_{k},\overline{1}_{k},m_{k})$;
(Assumption $B_{k}$) $\beta_{k}=\text{Pr}[U_{k}=1|x,\overline{z}_{k},\overline{1}_{k},m_{k}]-\text{Pr}[U_{k}=1|x,\overline{z}_{k},\overline{1}_{k}]$
does not depend on ($x,\overline{z}_{k}$); Assumption ($C$) $U_{k}$
is binary. Taking the difference between the quantities in the above
two equations thus gives that, for any $m_{k}\in\{0,1\}$, for any
$k\in[K]$, we have that

\begin{align}
\text{bias}(\tilde{\mathbb{E}}[Y(\overline{1}_{k},m_{k})]) & =\int\big(\mathbb{E}[Y|x,\overline{z}_{k},\overline{1}_{k},m_{k},U_{k}=1]-\mathbb{E}[Y|x,\overline{z}_{k},\overline{1}_{k},m_{k},U_{k}=0]\big)\cdot\nonumber \\
 & \quad\big(\text{Pr}[U_{k}=1|x,\overline{z}_{k},\overline{1}_{k},m_{k}]-\text{Pr}[U_{k}=1|x,\overline{z}_{k},\overline{1}_{k}]\big)\prod_{j=1}^{k}dP(z_{j}|x,\overline{z}_{j-1},\overline{1}_{j-1})dP(x).\label{eq:bias_1kmk}
\end{align}

Consider first $\text{bias}(\Delta_{k-1})=\text{bias}(\tilde{\mathbb{E}}[Y(\overline{1}_{k},0)-Y(\overline{1}_{k-1},0)])$.
Under mediator monotonicity (Assumption 1), I immediately have that
$\text{Pr}[U_{k}=1|x,\overline{z}_{k},\overline{1}_{k},m_{k}]-\text{Pr}[U_{k}=1|x,\overline{z}_{k},\overline{1}_{k}]$,
and thus that $\text{bias}(\Delta_{k-1})=\text{bias}(\tilde{\mathbb{E}}[Y(\overline{1}_{k},0)])$,
which can be written as

\begin{align*}
\text{bias}(\tilde{\mathbb{E}}[Y(\overline{1}_{k},0)]) & =\int\big(\mathbb{E}[Y|x,\overline{z}_{k},\overline{1}_{k},0,U_{k}=1]-\mathbb{E}[Y|x,\overline{z}_{k},\overline{1}_{k},0,U_{k}=0]\big)\\
 & \quad\big(\text{Pr}[U_{k}=1|x,\overline{z}_{k},\overline{1}_{k},0]-\text{Pr}[U_{k}=1|x,\overline{z}_{k},\overline{1}_{k}]\big)\prod_{j=1}^{k}dP(z_{j}|x,\overline{z}_{j-1},\overline{1}_{j-1})dP(x)\\
 & =\int\big(\mathbb{E}[Y|x,\overline{z}_{k},\overline{1}_{k},0,U_{k}=1]-\mathbb{E}[Y|x,\overline{z}_{k},\overline{1}_{k},0,U_{k}=0]\big)\cdot\\
 & \quad\big(\text{Pr}[U_{k}=1|x,\overline{z}_{k},\overline{1}_{k},0]-\big(\text{Pr}[U_{k}=1|x,\overline{z}_{k},\overline{1}_{k+1}]\text{Pr}[M_{k}=1|x,\overline{z}_{k},\overline{1}_{k}]\\
 & +\text{Pr}[U_{k}=1|x,\overline{z}_{k},\overline{1}_{k},0]-\text{Pr}[U_{k}=1|x,\overline{z}_{k},\overline{1}_{k},0]\text{Pr}[M_{k}=1|x,\overline{z}_{k},\overline{1}_{k}]\big)\big)\\
 & \quad\prod_{j=1}^{k}dP(z_{j}|x,\overline{z}_{j-1},\overline{1}_{j-1})dP(x)\\
 & =-\int\big(\mathbb{E}[Y|x,\overline{z}_{k},\overline{1}_{k},0U_{k}=1]-\mathbb{E}[Y|x,\overline{z}_{k},\overline{1}_{k},0_{k},U_{k}=0]\big)\cdot\\
 & \quad\big(\big(\text{Pr}[U_{k}=1|x,\overline{z}_{k},\overline{1}_{k+1}]-\text{Pr}[U_{k}=1|x,\overline{z}_{k},\overline{1}_{k},0]\big)\text{Pr}[M_{k}=1|x,\overline{z}_{k},\overline{1}_{k}]\big)\\
 & \quad\prod_{j=1}^{k}dP(z_{j}|x,\overline{z}_{j-1},\overline{1}_{j-1})dP(x).
\end{align*}

Next, applying assumptions $A_{k}$ and $B_{k}$, we can write

\[
\text{bias}(\tilde{\mathbb{E}}[Y(\overline{1}_{k+1},0)])=-\alpha_{k}\beta_{k}\int_{x}\int_{\overline{z}_{k}}\text{Pr}[M_{k}=1|x,\overline{z}_{k},\overline{1}_{k}]\prod_{j=1}^{k}dP(z_{j}|x,\overline{z}_{j-1},\overline{1}_{j-1})dP(x).
\]

Second, to compute $\text{bias}(\tau_{k})=\text{bias}(\mathbb{E}[Y(\overline{1}_{k+1})-Y(\overline{1}_{k},0)])$
for any $k\in\{1,\dots K\}$, beginning with Equation \ref{eq:bias_1kmk}
and applying assumptions $A_{k}$ and $B_{k}$ once again, we have
that

\begin{align*}
\text{bias}(\tau_{k}) & =\alpha_{k}\beta_{k}.
\end{align*}

\end{document}